\documentclass[11pt,letterpaper]{amsart}
\usepackage{amsaddr}

\usepackage[margin=1in]{geometry}

\usepackage[utf8]{inputenc}

\usepackage[T1]{fontenc}
\usepackage{lmodern}

\usepackage[colorinlistoftodos,bordercolor=orange,backgroundcolor=orange!20,linecolor=orange,textsize=normalsize]{todonotes}

\usepackage{amsmath}  
\usepackage{amssymb}     
\usepackage{bbm}

\usepackage{bm}
\usepackage{hyperref}
\usepackage{mathrsfs}
\usepackage{mathtools} 
\usepackage{paralist}
\usepackage{fixmath}
\usepackage{interval}
\usepackage{diffcoeff}
\usepackage{enumitem}
\usepackage{algpseudocode}

\usepackage[capitalize]{cleveref}

\Crefname{figure}{Figure}{Figures}

\crefformat{equation}{\textup{#2(#1)#3}}
\crefrangeformat{equation}{\textup{#3(#1)#4--#5(#2)#6}}
\crefmultiformat{equation}{\textup{#2(#1)#3}}{ and \textup{#2(#1)#3}}
{, \textup{#2(#1)#3}}{, and \textup{#2(#1)#3}}
\crefrangemultiformat{equation}{\textup{#3(#1)#4--#5(#2)#6}}%
{ and \textup{#3(#1)#4--#5(#2)#6}}{, \textup{#3(#1)#4--#5(#2)#6}}{, and \textup{#3(#1)#4--#5(#2)#6}}

\Crefformat{equation}{#2Equation~\textup{(#1)}#3}
\Crefrangeformat{equation}{Equations~\textup{#3(#1)#4--#5(#2)#6}}
\Crefmultiformat{equation}{Equations~\textup{#2(#1)#3}}{ and \textup{#2(#1)#3}}
{, \textup{#2(#1)#3}}{, and \textup{#2(#1)#3}}
\Crefrangemultiformat{equation}{Equations~\textup{#3(#1)#4--#5(#2)#6}}%
{ and \textup{#3(#1)#4--#5(#2)#6}}{, \textup{#3(#1)#4--#5(#2)#6}}{, and \textup{#3(#1)#4--#5(#2)#6}}

\usepackage{tikz}

\usetikzlibrary{arrows}
\usetikzlibrary{patterns}
\usetikzlibrary{calc}
\usetikzlibrary{shapes}
\usetikzlibrary{positioning}
\usetikzlibrary{arrows.meta,math,shapes.geometric,decorations.pathmorphing}

\usepackage{etoolbox,xparse}
\usepackage{amssymb,amsthm,amsmath}

\makeatletter
\def\@instring#1#2{TT\fi\begingroup
  \edef\x{\endgroup\noexpand\in@{#1}{#2}}\x\ifin@}

\newcommand\@newmathsfone[1]{%
  \expandafter\newcommand\csname #1\endcsname{\mathsf{#1}}%
}

\newcommand\@newmathsftwo[2]{%
  \expandafter\newcommand\csname #1\endcsname{\mathsf{#2}}%
}

\NewDocumentCommand\@newmathsfsplit{ > { \SplitArgument { 1 } { : } } m}
{ \@newmathsftwo #1 }

\newcommand{\newmathsf}[1]{%
\if\@instring{:}{#1}%
\@newmathsfsplit{#1}
\else%
\@newmathsfone{#1}
\fi}

\makeatother

\newcommand*{\newmathsfs}{%
  \let\do\newmathsf%
  \docsvlist%
}

\usepackage[scr=boondox,scrscaled=1.05]{mathalfa}

\renewcommand{\geq}{\geqslant}
\renewcommand{\leq}{\leqslant}

\renewcommand{\le}{\leq}
\renewcommand{\ge}{\geq}

\DeclareMathAlphabet{\mathbfcal}{OMS}{cmsy}{b}{n}
\DeclareMathAlphabet{\mathbbold}{U}{bbold}{m}{n}

\newcommand{\E}{\mathbb{E}}

\newcommand{\N}{\mathbb{N}} 
\newcommand{\R}{\mathbb{R}}

\newcommand{\trop}[1][]{\ifthenelse{\equal{#1}{}}{ \mathbb{T} }{ \mathbb{T}(#1) }}

\newcommand{\abs}[1]{|{#1}|}

\newcommand{\card}[1]{|{#1}|}

\newtheorem{theorem}{Theorem}[section]
\newtheorem{proposition}[theorem]{Proposition}

\newtheorem{corollary}[theorem]{Corollary}

\newtheorem{lemma}[theorem]{Lemma}
\theoremstyle{definition}
\newtheorem{definition}[theorem]{Definition}
\newtheorem{assumption}{Assumption}

\theoremstyle{remark}
\newtheorem{remark}[theorem]{Remark}
\newtheorem{example}[theorem]{Example}

\tikzset{grid/.style={gray!30,very thin}}
\tikzset{axis/.style={gray!50,->,>=stealth'}}
\tikzset{convex/.style={draw=none,fill=lightgray,fill opacity=0.7}}
\tikzset{convexborder/.style={very thick}}
\tikzset{point/.style={blue!50}}
\tikzset{hs/.style={fill opacity=0.3,fill=orange,draw=none}}
\tikzset{hsborder/.style={orange,ultra thick,dashdotted}}

\newcommand{\polyh}{\mathcal{P}}

\newcommand{\argmax}{\arg\max}
\newcommand{\argmin}{\arg\min}

\newcommand{\dgraph}{\vec{\mathcal{G}}}

\newcommand{\vertices}{V}
\newcommand{\vertex}{v}
\newcommand{\edges}{E}

\newcommand{\cycle}{C}

\newcommand{\Max}{\mathrm{Max}}
\newcommand{\Min}{\mathrm{Min}}

\newcommand{\Maxvertices}{\vertices_{\Max}}
\newcommand{\Minvertices}{\vertices_{\Min}}

\newcommand{\dunion}{\uplus}

\newcommand{\payoff}{r}

\newcommand{\shapley}{F}

\newcommand{\gameval}{\lambda}
\newcommand{\dominion}{\mathcal{D}}

\newcommand{\Prob}{\mathbb{P}}

\newcommand{\dpath}[3][]{(#1\ifthenelse{\equal{#1}{}}{}{,} #2 \rightarrow #3   )}
\newcommand{\ndpath}[3][]{(#1\ifthenelse{\equal{#1}{}}{}{,} #2 \nrightarrow #3   )}

\newcommand{\onenorm}[1]{\| #1 \|_{\mathrm{1}}}
\newcommand{\supnorm}[1]{\| #1 \|_{\infty}}

\newcommand{\bias}{u}

\newcommand{\poly}{\mathrm{poly}}

\newcommand{\Normal}{\mathcal{N}}
\newcommand{\Generic}{\mathcal{U}}
\newcommand{\dens}{f}

\newcommand{\Var}{\mathrm{Var}}

\newcommand{\Cond}{\Delta}
\newcommand{\disc}{\gamma}

\newcommand{\SwitchMax}{\textsc{SwitchMax}}
\newcommand{\SwitchMin}{\textsc{SwitchMin}}
\newcommand{\DiscPI}{\textsc{DiscountedPI}}
\newcommand{\IncDiscPI}{\textsc{IncreasingDiscountPI}}
\newcommand{\IncDiscPItwo}{\textsc{IncreasingDiscountPI2}}
\newcommand{\IncDiscPIor}{\textsc{IncreasingDiscountPIOracle}}
\newcommand{\IncDiscPIortwo}{\textsc{IncreasingDiscountPIOracle2}}

\newcommand{\eps}{\varepsilon}

\title{Smoothed analysis of deterministic discounted and mean-payoff games}

\author[Bruno Loff]{Bruno Loff}
\address{LASIGE, Faculdade de Ciências, Universidade de Lisboa}
\email{bruno.loff@gmail.com}

\author[Mateusz Skomra]{Mateusz Skomra}
\address{LAAS-CNRS, Universit{\'e} de Toulouse, CNRS, Toulouse, France}
\email{mateusz.skomra@laas.fr}

\begin{document}

\makeatletter
\def\@tocline#1#2#3#4#5#6#7{\relax
  \ifnum #1>\c@tocdepth 
  \else
    \par \addpenalty\@secpenalty\addvspace{#2}%
    \begingroup \hyphenpenalty\@M
    \@ifempty{#4}{%
      \@tempdima\csname r@tocindent\number#1\endcsname\relax
    }{%
      \@tempdima#4\relax
    }%
    \parindent\z@ \leftskip#3\relax \advance\leftskip\@tempdima\relax
    \rightskip\@pnumwidth plus4em \parfillskip-\@pnumwidth
    #5\leavevmode\hskip-\@tempdima
      \ifcase #1
       \or\or \hskip 1em \or \hskip 2em \else \hskip 3em \fi%
      #6\nobreak\relax
    \hfill\hbox to\@pnumwidth{\@tocpagenum{#7}}\par
    \nobreak
    \endgroup
  \fi}
\makeatother

\begin{abstract}
  We devise a policy-iteration algorithm for deterministic two-player discounted and mean-payoff games, that runs in polynomial time with high probability, on any input where each payoff is chosen independently from a sufficiently random distribution.

  This includes the case where an arbitrary set of payoffs has been perturbed by a Gaussian, showing for the first time that policy-iteration algorithms are efficient, in the sense of smoothed analysis.

  More generally, we devise a \emph{condition number} for deterministic discounted and mean-payoff games, and show that our algorithm runs in time polynomial in this condition number.

  \medskip
  Our result confirms a previous conjecture of Boros et al., which was claimed as a theorem \cite{BorosElbassioniFouzGurvich:2011} and later retracted \cite{boros2018approximation}. It stands in contrast with a recent counter-example by Christ and Yannakakis \cite{christ2023smoothed}, showing that Howard's policy-iteration algorithm does \emph{not} run in smoothed polynomial time on \emph{stochastic} single-player mean-payoff games.

  Our approach is inspired by the analysis of random optimal assignment instances by Frieze and Sorkin \cite{frieze2007probabilistic}, and the analysis of bias-induced policies for mean-payoff games by Akian, Gaubert and Hochart \cite{generic_uniqueness}.
\end{abstract}

\maketitle

\bigskip\noindent\textbf{Acknowledgements.} MS would like to thank Xavier Allamigeon, St\'{e}phane Gaubert, and Ricardo D. Katz for many useful discussions on mean-payoff games, policy iteration, for explaining the operator approach of \cite{generic_uniqueness}, for exchanging ideas about the problem of smoothed analysis, for their remarks on a preliminary version of this paper, and for being a perpetual source of friendship and inspiration.

Funded by the European Union (ERC, HOFGA, 101041696). Views and opinions expressed are however those of the author(s) only and do not necessarily reflect those of the European Union or the European Research Council. Neither the European Union nor the granting authority can be held responsible for them. Also supported by FCT through the LASIGE Research Unit, ref.\ UIDB/00408/2020 and ref.\ UIDP/00408/2020.

\tableofcontents

\section{Introduction}

\subsection{A history of discounted and mean-payoff games} John von Neumann proved his minimax theorem in 1928, founding game theory by showing the existence of optimal strategies in zero-sum matrix games. In 1953, Lloyd Shapley~\cite{shapley_stochastic} considered what happened if two players repeatedly played a zero-sum matrix game. The overall game proceeds as follows. We have $n$ states, and to each state $i\in[n]$ corresponds a zero-sum matrix game $G_i$. At each round, the two players are in some state $i\in[n]$ and play the corresponding game $G_i$, with each player simultaneously choosing an action out of a finite set of possible actions. The two players' choice of actions determines not just the payoff, but also the state at the next round. This is repeated ad-infinitum.

In these games, randomness is possible in two different ways. First, the state at the next round can be chosen stochastically, or deterministically, based on the current state and on the players' chosen actions. This gives us two variants: \emph{stochastic} games, and \emph{deterministic} games, with the latter being a special case of the former. Second, the players' choice of action can itself be \emph{pure} (a single action), or \emph{mixed} (a distribution over the possible actions).

In an infinite game such as this there are two natural ways of determining the winning player. In the \emph{discounted} variant, payoffs received at round $t$ are multiplied by a \emph{discount factor} of $\disc^t$ (for some $0 \le \disc < 1$), and we wish to know the total discounted payoff in the limit as the number of rounds goes to infinity. This is equivalent to saying that the game is forced to stop after every round with probability $1-\disc$, and asking for the expected payoff at the limit. In the \emph{mean-payoff} variant, we measure the liminf or limsup, as the number of rounds goes to infinity, of the average payoff received so far (i.e. total payoff divided by the number of rounds). Shapley~\cite{shapley_stochastic} proved the existence of a value and optimal (mixed) strategies for the stochastic, discounted variant.

Concurrently to Shapley's work, Bellman \cite{bellman1957dynamic} studied a class of problems which he termed \emph{Markov Decision Processes} (MDPs). MDPs model decision making when the result of one's actions can be partly random, and they can be seen as \emph{single-player} variant of stochastic games. At each round, the player finds himself in a given state out of a finite number of states, and chooses an action. Depending on his choice he receives a payoff, and transitions to a different state. This transition can be either deterministic or stochastic. The player's goal is to maximize the discounted payoff or mean payoff at the limit, as the number of rounds goes to infinity. Bellman provided a method to find an optimal pure strategy in the discounted variant.


In both MDPs and in discounted games the optimal strategies can be made \emph{memoryless}, in that the choice what to do only depends on the current state $i$, and not on the past history. In the general case of discounted stochastic games where the players play simultaneously, the optimal memoryless strategies must be mixed. In the case of MDPs, the optimal memoryless strategy can further be made \emph{pure}.

As for the mean-payoff variant, Gillette~\cite{gillette} gave an example of a mean-payoff two-player game, where the players play simultaneously at each round, whose optimal strategies cannot be memoryless.\footnote{In fact, it was only in the 1980s \cite{mertens_neyman} that simultanous-move mean-payoff games were proven to have a value. For every $\eps > 0$, optimal (not memoryless) strategies exist for each player ensuring that the payoff is $\eps$-close to the value.} With this in mind, Gillette introduced a \emph{turn-based} variant of Shapley's infinite game, where two players play in turns. At each round, the game is in some state $i$, and one of the players (depending on $i$) chooses an action, which determines the next state (stochastically or deterministically) and a resulting payoff. One player is trying to maximize the payoff at the limit, and the other tries to minimize it. Gillette claimed that turn-based two-player games have a value and that optimal strategies exist for both players which are both memoryless and pure. Gillette's proof was actually wrong, but it was later corrected by Liggett and Lippman~\cite{liggett_lippman}, so the statement \emph{is} true. It also implies the corresponding statement for the mean-payoff variant of MDPs, as a special case.

A pure, memoryless strategy for such games is called a \emph{policy}, and it is a finite object: one can represent it as a finite function $\sigma:\mathsf{States} \to \mathsf{Actions}$ specifying the chosen action at each state. In this paper, we will not concern ourselves with simultaneous two-player games, and only consider single-player games and turn-based two-player games. We will use the informal nomenclature ``deterministic/stochastic single-player/two-player discounted\allowbreak /\allowbreak mean-payoff game'' to denote each of the eight variants of (non-simultaneous) games just mentioned. Let us also use the term ``discounted and mean-payoff games'' to refer to these games as whole.

\subsection{Algorithms} The above results show that all eight variants have a value, and that optimal strategies are policies, hence finite objects. It then makes sense to consider the algorithmic problem of \emph{solving} such a game: given as input a specification of the game with rational weights (and with rational discount factor, if applicable), compute the value of the game, and optimal policies for the players.\footnote{It can be shown that the value of such a game with rational weights (and discount factor) is a rational number of comparable size.}
The study of discounted and mean-payoff games has always been accompanied by the development of algorithms for solving them. Most algorithms for solving discounted and mean-payoff games can be broadly classified into three families: value-iteration algorithms\footnote{The first algorithm ever invented was a value iteration algorithm \cite{bellman1957dynamic}. For a modern value-iteration algorithm for single-player games, see \cite{sidford2018variance}, which contains a historical overview in Section 2. For two players see, e.g., \cite{zwick_paterson,ibsen-jensen_miltersen,chatterjee_ibsen-jensen,AllamigeonGaubertKatzSkomra:2022,AkianGaubertNaepelsTerver:2023}. 
}, algorithms for MDPs that use linear programming\footnote{See, e.g., Chapter 2 of \cite{filar_vrieze}, or various sections of \cite{puterman}.}, and, of particular importance to us, policy iteration algorithms.

Policy iteration algorithms have been invented for solving all variants of games described above. These algorithms maintain a policy in memory, and proceed by repeatedly modifying the policy, so that its quality improves monotonically according to some measure, until it can no longer be improved, at which point the measure must guarantee that we have found an optimal strategy for both players.

The first policy iteration algorithm was invented by Howard \cite{howard1960dynamic}, and finds an optimal strategy for deterministic (and some stochastic) Markov Decision Processes. This was later extended by Denardo and Fox \cite{denardo1968multichain} to work on all stochastic MDPs. The method was first extended to two-player mean-payoff games by Hoffman and Karp \cite{hoffman_karp}, and two-player discounted games by Denardo \cite{denardo1967contraction}, with later developments by many other authors \cite{rao1973algorithms,puri1995theory,CochetGaubertGunawardena:1999,gaubert1998duality,raghavan2003policy,cochetterrasson2006policy}. A good historical overview with more technical details appears in \cite{detournay_policy}, where a policy iteration algorithm first appeared that can handle all the variants of mean-payoff games. In the case of discounted games, an optimal strategy can be found in time polynomial in $\frac{1}{1-\gamma}$ \cite{ye2011simplex,hansen2013strategy}. Otherwise, for mean-payoff games, or for discount factors $\gamma$ exponentially close to $1$, no upper-bound is known on the number of iterations, significantly better than the number of policies, which is exponential in the number $n$ of states. More precisely, the best upper-bound on the number of iterations is $2^{\tilde O(\sqrt n)}$ \cite{halman,hansen_zwick_random_facet}.

\subsection{Policy iteration versus the simplex method}
When one first studies policy iteration algorithms, one gets a sense of familiarity, as if policy iteration algorithms are analogous to the simplex method for linear programming. The intuitive sense is that the choice of policy plays the same role as the choice of basic feasible solution in the simplex method, with a change in policy being analogous to a pivot operation.

In fact, in some cases, this analogy can be formally established. It is possible to express a MDP by a particular linear program, and in this particular case the connection is perfect: a simplex pivoting rule gives us a policy iteration algorithms for MDPs, and any policy iteration algorithm that switches a single node at a time gives us a pivoting rule for applying simplex on this particular program.

As a result, many known counter-examples for the simplex method, showing that certain pivoting rules require an exponential number of pivots, were devised by first finding examples of MDPs for which certain policy-iteration algorithms need an exponential number of iterations, and then \emph{translating} the counter-example to work for the simplex algorithm, by the above connection \cite{friedmann2011exponential,disser2023unified}.

More broadly, it turns out that deterministic two-player mean-payoff games are exactly equivalent to tropical linear programming, i.e., solving systems of ``linear'' inequalities over the tropical $(\min,+)$ semiring.
This was first explicitly shown by Akian, Gaubert, and Guterman~\cite{polyhedra_equiv_mean_payoff}, strengthening earlier connections between these problems that were made in the literature on tropical algebra (such as \cite{gaubert1998duality,dhingra2006how,katz2007max}) and in works on scheduling problems \cite{moehring2004scheduling}.\footnote{Stochastic two-player mean-payoff games, on the other hand, are equivalent to tropical semidefinite programming~\cite{issac2016jsc}.}

Furthermore, tropical linear programs can be reduced to linear programs over the non-Archimedean field of convergent generalized power series \cite{develin_yu,tropical_simplex}.\footnote{A linear program over such a field can be thought of as a parametric family of linear programs over $\R$. It follows from the above reduction that deterministic mean-payoff games can be encoded as linear programs with coefficients of exponential bit-length. Such an encoding was first derived by Schewe \cite{schewe2009parity}, without reference to non-Archimedean fields.} This characterization has been exploited to show that, if there exists a strongly-polynomial-time pivoting rule for the simplex algorithm, where the choice of basis element to pivot is semialgebraic in a certain technical sense (and this is the case for many pivoting rules), then the entire algorithm can be tropicalized, to get a polynomial-time algorithm for deterministic two-player mean-payoff games~\cite{combinatorial_mean_payoff}.

The analogy between policy iteration and the simplex method is also seen in practice. The aforementioned counter-examples show that policy-iteration algorithms run in exponential time in the worst-case. 
And yet, various benchmarks have shown that policy-iteration algorithms are very efficient at solving real-world instances, both for single-player \cite{georgiadis2009experimental,cohen1998numerical,kretinsky2017efficient}  and two-player games \cite{dhingra2006how,chaloupka2009parallel}. This difference between worst-case and real-life performance is also what happens with the simplex method. And in both cases it begs the question: \emph{why?}

\subsection{Smoothed analysis} 
In the case of the simplex method, the generally accepted explanation was proposed by Spielman and Teng \cite{spielman2004smoothed}. They have shown that, if one takes any linear programming instance $\max \{ c\cdot x \mid A \cdot x \ge b\}$ of dimension $n$, and perturbs each entry of $A,b$ and $c$ by a Gaussian with mean $0$ and standard deviation $\frac{1}{\phi}$, then the simplex method, with a particular choice of pivoting rule, will solve the resulting perturbed system in time $\poly(\phi \cdot n)$ \cite{spielman2004smoothed,dadush2018friendly}. It is then reasonable to expect the simplex method to work efficiently on real-world instances, since they incorporate real-world data which is prone to such perturbations. It was this result of Spielman and Teng that founded the area of \emph{smoothed analysis}, where one studies the efficiency of algorithms on such perturbed inputs.

The question then naturally follows: are policy-iteration algorithms efficient in the sense of smoothed analysis?

\medskip
Recent evidence seems to indicate that no, they are not. The first policy-iteration algorithm for single-player games (MDPs), by Howard \cite{howard1960dynamic}, determines for the current policy $\sigma:\mathsf{States} \to \mathsf{Actions}$, and for each state $i$ of the game, if a local improvement is possible: \emph{would a different choice of action at $i$ improve the value of the game when starting at $i$, if the game were to be played according to $\sigma$ at every other state?} The algorithm then changes the action $\sigma(i)$ at every state $i$ where such a local improvement is possible, to the best possible local improvement. This is sometimes called \emph{Howard's policy iteration}, or the \emph{greedy all-switches} rule.

In a paper published in STOC last year, Christ and Yannakakis \cite{christ2023smoothed} showed a remarkable lower-bound. They showed that $2^{\Omega(n^c)}$ iterations are necessary on a certain family of stochastic MDPs (single-player games), even when the payoffs are perturbed. In fact, the lower-bound holds not only probabilistically, where each payoff is independently perturbed by a Gaussian with standard deviation $\frac{1}{\poly(n)}$, but even adversarily, where each payoff is perturbed by any value within $\pm\frac{1}{\poly(n)}$.

It is surprising that such a bound can be proven at all. However, their result only holds for \emph{stochastic} games, and does not necessarily apply to deterministic games, where it has been previously conjectured that Howard's rule is efficient \cite{hansen2010lower}. Also, this result shows that a particular way of improving the policy, the greedy all-switches rule, does not give us an efficient algorithm (in the sense of smoothed analysis). This is analogous to saying that a particular pivoting rule in the simplex algorithm is not efficient, and does not exclude the possibility that other ways of improving the policy might work.

Exploiting any one of these caveats could in principle allow for a smoothed analysis for policy iteration. Our main result exploits both: we show that for \emph{deterministic} two-player games, a \emph{slightly different} policy-improvement method will be efficient, in the sense of smoothed analysis.

\begin{theorem}[our main theorem]\label{thm:main1}
  There exists a policy-iteration algorithm for solving $n$-state deterministic two-player (discounted or mean-payoff) games, which runs in time $\poly(\phi \cdot n)$ with high probability, on an input where normalized payoffs in $[-1,1]$ have been independently perturbed by a Gaussian with mean $0$ and standard deviation $\frac{1}{\phi}$.
\end{theorem}

It should be emphasized that the lower-bound of Christ and Yannakakis holds even for the single-player case, and our result could be contrasted with theirs, even if we only had proved it for the single-player case. However, our policy-iteration algorithm (our upper-bound) works even for two-player games, which are much harder. Our policy-improvement rule is similar to the greedy all-switches rule, except that the choice of switches is allowed to depend on an additional parameter (a discount factor) which evolves over time.

\subsection{Computational complexity} Our result should be contrasted with the case of the simplex algorithm for linear programming. Recall that we \emph{do} know polynomial-time algorithms for linear programming, it is only the simplex algorithm which fails to run efficiently in the worst case. However, it should be emphasized that we do \emph{not} know any polynomial-time algorithms for solving two-player discounted or mean-payoff games.

Indeed, solving one-player discounted and mean-payoff games reduces to linear programming, so we have polynomial-time algorithms. But the complexity of solving \emph{two-player} discounted and mean-payoff games is one of the great unsolved problems in computational complexity, first posed by Gurvich, Karzanov, and Khachiyan \cite{gurvich}. 
For any of the two-player variants, the respective decision problem (what is the $i$-th bit of the value) is in $\mathsf{UP} \cap \mathsf{coUP}$ \cite{jurdzinski_parity_up}, and the search problem (find an optimal strategy) is in the computational complexity class $\mathsf{UEOPL}$ (\emph{Unique End-of-Potential-Line}) \cite{hubacek2020hardness,fearnley2020unique}. In fact it is in the promise version of this class, which we denote $\mathsf{pUEOPL}$. Even within the class $\mathsf{pUEOPL}$, the problem of solving two-player discounted and mean-payoff games seems to be only a special, simple case: the sought optimal policies can be obtained from the unique fixed point of a simple monotone operator. This places the search problem in the class $\mathsf{Tarski}$. The two complexity classes $\mathsf{pUEOPL}$ and $\mathsf{Tarski}$ sit at the bottom of a large hierarchy of complexity classes \cite{fearnley2022complexity,goeoes2022further}. This hierarchy stratifies the broad class $\mathsf{TFNP}$ of \emph{$\mathsf{NP}$ search problems} \cite{johnson1988how,megiddo1991total,papadimitriou1994complexity}. See \Cref{fig:complexity-classes}.

    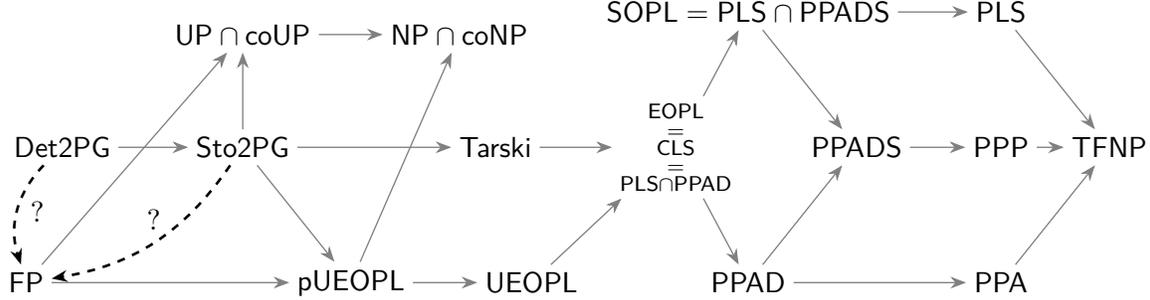
\begin{figure}
      \newmathsfs{FP,UEOPL,pUEOPL,EOPL,SOPL,CLS,PLS,PPAD,PPADS,PPP,PPA,TFNP,DetG:Det2PG,StG:Sto2PG,NP,coNP,UP,coUP,Tarski}
      \centering
      \begin{tikzpicture}[scale=0.6]
    \tikzset{inner sep=0,outer sep=3}


    \begin{scope}[xscale=1.6]
      \node (FP) at (-7.5,-3) {$\FP$};
      \node (DetG) at (-7,0) {$\DetG$};
      \node (StG) at (-4.5,0) {$\StG$};
      \node (pUEOPL) at (-3,-3) {$\pUEOPL$};
      \node (UPcoUP) at (-4.5,2.5) {$\UP\cap\coUP$};
      \node (NPcoNP) at (-1.5,2.5) {$\NP\cap\coNP$};
      \node (Tarski) at (-1,0) {$\Tarski$};
      \node (UEOPL) at (-.5,-3) {$\UEOPL$};
      \node (PLS-PPAD) at (1.5,0) {$\substack{\EOPL\\=\\\CLS\\=\\\PLS\cap\PPAD}$};
      \node (PLS-PPADS) at (2.5,3) {$\SOPL=\PLS\cap\PPADS$};
      \node (PPAD) at (2.5,-3) {$\PPAD$};
      \node (PPADS) at (4,0) {$\PPADS$};
      \node (PLS) at (6,3) {$\PLS$};
      \node (PPP) at (6,0) {$\PPP$};
      \node (PPA) at (6,-3) {$\PPA$};
      \node (TFNP) at (7.5,0) {$\TFNP$};
    \end{scope}

    \path[-{Stealth[length=6pt]},line width=.5pt,gray]
    (FP) edge (pUEOPL)
    (FP) edge (UPcoUP)
    (DetG) edge (StG)
    (StG) edge (pUEOPL)
    (StG) edge (UPcoUP)
    (StG) edge (Tarski)
    (Tarski) edge (PLS-PPAD)
    (UPcoUP) edge (NPcoNP)
    (pUEOPL) edge (NPcoNP)
    (pUEOPL) edge (UEOPL)
    (UEOPL) edge (PLS-PPAD)
    (PLS-PPAD) edge (PLS-PPADS)
    (PLS-PPAD) edge (PPAD)
    (PLS-PPADS) edge (PPADS)
    (PLS-PPADS) edge (PLS)
    (PPAD) edge (PPADS)
    (PPAD) edge (PPA)
    (PPADS) edge (PPP)
    (PLS) edge (TFNP)
    (PPP) edge (TFNP)
    (PPA) edge (TFNP);

    \hypersetup{hidelinks}
    \tikzset{new/.style={-{Stealth[length=6pt]},line width=1pt}}
    
\draw[new]
(DetG) edge[dashed,bend right=31]
node[midway,right,inner sep=2pt] {?}
(FP);
\draw[new]
(StG) edge[dashed,bend left=23]
node[midway,above,inner sep=2pt] {?}
(FP);

  \end{tikzpicture}
  
  \caption{A hierarchy of $\NP$ search problems. $\DetG$ and $\StG$ refer to deterministic, respectively stochastic, two-player games. Arrows denote inclusion or containment. By inclusion of a search problem in the classes $\UP \cap \coUP$ and $\NP \cap \coUP$ of decision problems, we mean that the problem of deciding each bit of the unique answer can be computed in these classes. }
  \label{fig:complexity-classes}
\end{figure}


In this sense, the problem of solving two-player discounted and mean-payoff games is the simplest known problem in $\mathsf{NP}$, which is not yet known to be solvable in polynomial time (or even in time $2^{n^{o(1)}}$). For any problem which is not known to be in $\mathsf{P}$, one may ask the question: \emph{is the problem still hard on a random instance?} Many hard problems have been conjectured to have this property, of being hard to solve even for a random instance. This is the case for $\mathsf{SAT}$ \cite{selman1996generating,feige2002relations} and subset sum \cite{rudich1997super}, but also for other, not necessarily $\mathsf{NP}$-hard, problems in $\mathsf{NP}$, such as lattice problems \cite{ajtai1999generating}. There is a broad belief that natural problems which are hard, remain hard on random inputs.\footnote{The distributions which are considered hard must be chosen carefully to avoid trivial cases, e.g. CNF formulas with too many or too few clauses, but they are often simple and natural distributions.} In contrast, our main theorem generalizes to show that solving deterministic two-player games is easy, for any sufficiently random input distribution.

\begin{theorem}[generalization]
  Consider distributions on $n$-state deterministic two-player discounted or mean-payoff games, where each payoff is chosen independently according to a (not necessarily identical) distribution with mean in $[-1,1]$ and standard deviation $\le \frac{1}{\phi}$ and with probability density functions satisfying $f(y) \le \phi$ for all $y \in \mathbb R$. (This includes, for example, payoffs perturbed by a Gaussian, or payoffs sampled from a uniform interval of length $\ge \frac{1}{\phi}$.)
  
  There exists a policy-iteration algorithm which runs in time $\poly(n, \phi)$, with high probability, on games sampled according to any such distribution.
\end{theorem}

\subsection{A previous approach and our approach} 

A first, naive attempt at proving \Cref{thm:main1} could proceed along the same lines as the Mulmuley, Vazirani and Vazirani's isolation lemma \cite{mulmuley1987matching}. We think of what happens to the game when all of the payoffs of all the actions are fixed, except for one, which is sampled independently according to some distribution as above.
  Let us suppose that the free payoff is for an action of some state $i$. As it turns out, when every other payoff is fixed, the value of the game at state $i$ is a piecewise linear function of the free payoff, and one can then hope that it has few break points.
If this were indeed the case, that there were only $\poly(n)$
break points, one could then argue similarly to the MVV isolation lemma, to show that approximating the random payoffs using $O(\log(n))$ bits of precision is enough to isolate the linear piece (in the piecewise linear function).
From this it would follow that optimal policies for the truncated payoffs are also optimal for the untruncated payoffs.
One could then invoke a pseudo-polynomial-time value-iteration algorithm \cite{zwick_paterson} on the truncated payoffs, and this would run in polynomial time.

The conference version of the paper of Boros, Elbassioni, Fouz, Gurvich, Makino and Manthey \cite{BorosElbassioniFouzGurvich:2011} outlines a similar proof strategy. Among other results, a result similar to our \Cref{thm:main1} was claimed. Their proof works for the one-player case, and the authors claimed, without a careful proof, that the two-player case also follows. This claim was later retracted in the journal version \cite{boros2018approximation}.\footnote{In their paper, the breakpoints are chosen so that between any two breakpoints the value of the optimal strategy and the value of the second-best strategy are sufficiently far apart. This works for the single player-case, but the argument we just presented, where we only keep track of the number of break points in the value function, is simpler and also works.} Indeed, it turns out that the two-player case is significantly more subtle.

In the one-player case it can be shown that for every action there exist at most $n$
break points, and hence an isolation lemma can be proven. One can then conjecture, for the two-player case, a $\poly(n)$ upper-bound on the number of break points. As it turns out, this conjecture is wrong. 
An exponential example (\cref{ex:exponential}) can be created using the construction of \cite{bezem2008exponential} that was also used in \cite{log_barrier} to prove that interior point methods for linear programming are not strongly polynomial. This construction gives us a deterministic two-player game with $n$ states such that, leaving the payoffs of all but one of actions fixed, as the free payoff varies between $-1$ and $1$, the value of the game is a piecewise-linear function with $2^{\Omega(n)}$ break points.

\medskip
So what do we do instead? One natural thing to try is to show that the number of break points \emph{is} $\poly(n)$, with very high probability, for randomly-chosen payoffs. This could well be true, and an argument in the style of the MVV isolation lemma would then follow. But we were unable to show it.

Instead, our results depend on a deeper analysis of deterministic two-player games. We also prove an isolation lemma, but using an approach different to MVV. Instead of attempting to isolate an optimal policy among all possible policies, we show that sufficiently random payoffs, with high probability, isolate a Blackwell-optimal policy. Blackwell-optimal policies are policies that arise in discounted games with discount factor $\gamma$ close to $1$. Blackwell-optimal policies are part of a family of policies which are induced by an object called a \emph{bias}.
Not all optimal policies are Blackwell-optimal, or even bias-induced. Nonetheless, every two-player discounted game has a Blackwell-optimal policy \cite{liggett_lippman}. Furthermore, there exist policy-iteration algorithms for finding a Blackwell-optimal policy \cite{kallenberg2002handbook,hordijk2002handbook} (that are inefficient in the worst case). 

We are then able to show (\Cref{threshold_two_players}) that, if the payoffs are sufficiently random, then with high probability it will happen that there is a unique bias-induced policy, which must then be the Blackwell-optimal policy. Furthermore, from the proof of this theorem we devise a \emph{condition number} $\Cond(r)$, associating a number in $[0, +\infty]$ to any given choice $r$ of payoffs (\Cref{def:condition-number}). The uniqueness proof generalizes to show that a deterministic two-player game with sufficiently random payoffs will have a small condition number ($\Delta(r) \le \poly(n)$) with high probability (\Cref{cond_estimate}).

This is a condition number in the same sense as the known condition numbers that govern the complexity of algorithms for solving linear equations, semidefinite feasibility, \emph{etc}, and broadly measure the \emph{inverse-distance to ill-posedness} (see \cite{buergisser2013condition}). Although, strictly speaking, our condition number does not measure inverse-distance to some set, it does have the property that $\Delta(r)$ is finite if and only if the game with payoffs $r$ has a unique bias-induced policy, and that every payoff $\tilde r \in B_\infty(r,\delta)$, within a ball of size $\delta \le \frac{1}{\poly(\Delta(r))}$ around $r$, will also have a unique bias-induced policy (\Cref{cond_estimate_approx}). So, at least intuitively, we can think of $\Delta(r)$ as an inverse-distance between $r$ and the set of games with multiple bias-induced policies. 

Finally, it can be shown that, taking a discount rate $\gamma \ge 1 - \frac{1}{\poly(n, \Cond)}$  (i.e., sufficiently close to $1$), the only optimal policy is the Blackwell-optimal policy (\Cref{conditioned_discount}). We can then use the results of \cite{ye2011simplex,hansen2013strategy} to obtain a policy iteration algorithm for finding the Blackwell optimal policy in time  $\poly(n,\Cond)$ (\Cref{main_algo,main_algo_disc}). This algorithm runs in polynomial smoothed time, because, as mentioned above, sufficiently random payoffs have small condition number. This includes any fixed choice of payoffs that has been randomly perturbed by a Gaussian.

\subsection{Related work}

There is not a lot of work on the complexity of discounted and mean-payoff games with random payoffs. Besides the papers of Boros et al. \cite{BorosElbassioniFouzGurvich:2011,boros2018approximation} and Christ and Yannakakis \cite{christ2023smoothed}, which we mentioned above, we only know of a paper by Mathieu and Wilson \cite{MathieuWilson:2013}. They do not provide an algorithm, but they analyze the distribution of the value of a deterministic single-player mean-payoff game (deterministic MDP) played on a complete graph with i.i.d.~exponentially distributed payoffs. Our algorithm will also work on such a payoff distribution.

A paper of Allamigeon, Benchimol, and Gaubert \cite{tropical_shadow_vertex} analyzes efficiency in a different random model. The shadow vertex rule is known to be efficient on average under any distribution over linear feasibility problems, which is symmetric up to changing of the sign in each linear inequality \cite{adler1987simplex}. Allamigeon et al tropicalize this result, to show that a certain tropical analogue of the shadow vertex rule will solve deterministic two-player mean-payoff games in a bipartite graph in expected polynomial time, if the distribution over the payoff matrix is invariant under transposition (this is the tropical analogue of the above symmetry). In particular, if the payoffs of some given fixed input obey the same symmetry, their algorithm runs in polynomial time. Of course, in general, perturbed payoffs need not be symmetric in this way.  

\medskip

It was a paper of Frieze and Sorkin \cite{frieze2007probabilistic} that gave us the first idea of how to approach the problem. Frieze and Sorkin analyse the gap between optimal and second-optimal assignment in the assignment problem, using a bound on the reduced costs of the associated linear program at the optimum solution \cite[Theorem 3]{frieze2007probabilistic}. In the simplex algorithm, the reduced cost works as a gradient, telling us the improvement in the objective function obtained by changing the current basic solution in a given direction. Frieze and Sorkin show that, at an optimum basic solution of a random assignment problem, every reduced cost is large, which implies that there is a large difference between the optimum and second-best solution. This large difference implies that the optimum solution is stable under perturbations. In the case of deterministic two-player games, biases will play the role of dual variables, which allows us define an analogous notion of reduced costs at the optimal solution. We then show, analogously, that, with high probability for a random instance, every reduced cost is large at the optimum policy, which also implies stability. Our condition number is the (normalized) inverse of the smallest reduced cost.

Our analysis of discounted and mean-payoff games is based on the operator approach to study these games. Using this approach, Akian, Gaubert, and Hochart \cite{generic_uniqueness} have previously shown that a generic two-player mean-payoff game has a unique bias (which must then equal the Blackwell bias). More precisely, they show that for any stochastic or deterministic two-player mean-payoff game, the set of payoffs where the bias is not unique has measure zero. We give a more precise version of this result, for deterministic two-player mean-payoff games, by showing that the policies induced by the unique bias are also unique. This further allows us the measure ``how far from having multiple bias-induced policies'' is a given choice of payoffs.

%

The operator approach was also used by Allamigeon, Gaubert, Katz and Skomra \cite{AllamigeonGaubertKatzSkomra:2022}, who define a condition number for the value iteration algorithm. Computer experiments indicate that value iteration  converges quickly for random games \cite{issac2016jsc}, which strongly suggests that the condition number of \cite{AllamigeonGaubertKatzSkomra:2022} is small for random games, but there is currently no formal proof of this claim. Even though our condition number and the one from \cite{AllamigeonGaubertKatzSkomra:2022} are based on the bias vector, it is not clear how these two conditions numbers compare to each other. In particular, we do not know if the value iteration algorithm has polynomial smoothed complexity and we leave this problem as an open question.

\subsection{Outline of the paper}

The paper is organized as follows.  In \cref{sec:one_player}, we analyze the single-player case. We give two simple proofs that a random choice of payoffs, with high probability, has a unique optimal solution which is significantly better than the second-best solution. In \cref{sec:det_mpg} we introduce two-player mean-payoff games, and give several examples explaining why the two-player case is much more subtle. In \cref{sec:ergodic}, we study discounted games, biases, and bias-induced policies for mean-payoff games on ergodic graphs. This culminates in \cref{polyhedral_complex}, showing that the set of games with non-unique bias-induced policies has Lebesgue measure zero.

Based on this result, in \cref{sec:of-threshold_two_players} we define a condition number $\Cond(r)$ and show that random games are robustly well conditioned, meaning, with high probability, a deterministic two-player mean-payoff game with random payoffs $r$ has small condition number $\Cond(r)$, and this holds for any other set of payoffs $\tilde r$ sufficiently close to $r$.

Finally, in \cref{sec:algorithms}, we present an algorithm for two-player games that runs in time $\poly(n, \Cond(r))$.

\subsection{Technical summary}
For the sake of simplicity, we restrict our attention to deterministic mean-payoff games played on an \emph{ergodic} weighted directed graph $\dgraph = ([n], \edges, \payoff)$, where $\card{\edges} = m$ and $[n]$ is split intro vertices controlled by Max and Min, $[n] = \Maxvertices \dunion \Minvertices$. The ergodicity is taken in the sense of \cite{gurvich_lebedev,ergodicity_conditions,hochartdominion}.\footnote{This is a notion similar to strong connectivity, but for two-player games. Intuitively, a graph is ergodic if no player can play in such a way as to force the game to get stuck on a sub-graph.} A typical example of such a graph is a complete bipartite graph, in which the bipartition is formed by $\Maxvertices,\Minvertices$. Ergodic graphs are representative for the difficulty of mean-payoff games, because solving games on general graphs reduces to solving games on complete bipartite graphs~\cite{ChatterjeeHenzingerKrinningerNanongkai:2014}.

The weights $\payoff_{ij}$ of the edges in our model are chosen randomly: we suppose that $(\payoff_{ij})$ are independent absolutely continuous variables with densities $\dens_{ij}$. We further suppose that weights are normalized --- $\E(\payoff_{ij}) \in [-1,1]$ --- and that there exists a number $\phi > 0$ such that $\dens_{ij}(y) \le \phi$ for all $i,j,y$ and $\Var(\payoff_{ij}) \le 1/\phi^2$. As an example, if the weights $\payoff$ are taken by perturbing some fixed weights $\bar{\payoff}_{ij} \in [-1,1]$ by Gaussian noise, so that $\payoff_{ij} \sim \Normal(\bar{\payoff}_{ij}, \rho^2)$, then we can take $\phi \coloneqq 1/\rho$.

Under the ergodicity assumption, it is known~\cite{hochartdominion} that the following \emph{ergodic equation} has a solution $(\gameval,\bias) \in \R^{n+1}$ for all choices of weights:
\begin{equation*}
\begin{cases}
\forall i \in \Maxvertices, \ \gameval + \bias_i = \max_{(i,j) \in \edges}\{\payoff_{ij} + \bias_j\} \, , \\
\forall i \in \Minvertices, \ \gameval + \bias_i = \min_{(i,j) \in \edges}\{\payoff_{ij} + \bias_j\} \, .
\end{cases}
\end{equation*}

Furthermore, the number $\gameval$ is unique and it is the value of the game (which does not depend on the initial state because of ergodicity). The vector $\bias \in \R^n$, called a \emph{bias}, is never unique because we can always add a constant to all its coordinates. In general, the bias is not unique even up to adding a constant. We say that a pair of policies $\sigma \colon \Maxvertices \to \vertices$ (of Max) and $\tau \colon \Minvertices \to \vertices$ (of Min) is \emph{bias-induced} if there is a bias $\bias$ such that the edges used by $\sigma,\tau$ achieve the maxima and minima in the ergodic equation. Bias-induced policies are optimal, but not every optimal policy is bias-induced. To study the behavior of random games, we introduce the sets 
\begin{align*}
\polyh^{\sigma,\tau} \coloneqq \{\payoff \in \R^m \colon &\text{$(\sigma,\tau)$ is the only pair of bias-induced policies}\\
&\text{in the MPG with weights $\payoff$}\} ,
\end{align*}
for any pair $(\sigma,\tau)$ such that the resulting graph $\dgraph^{\sigma,\tau}$ has a single directed cycle. We denote by $\Xi$ the set of all such pairs of policies. We also put $\Generic \coloneqq \cup \polyh^{\sigma,\tau}$. Using the techniques from \cite{generic_uniqueness} we are then able to show the following proposition. This proposition strengthens \cite[Theorem~3.2]{generic_uniqueness} for deterministic games by showing that each maximum and minimum in the ergodic equation is generically achieved by a single edge.

\begin{proposition}[{cf.~\cite[Theorem~3.2]{generic_uniqueness}}]\label{intro_polyhedral_complex}
The sets $\polyh^{\sigma,\tau}$ are open polyhedral cones. Moreover, these cones are disjoint and $\R^m \setminus \Generic$ is included in a finite union of hyperplanes. In particular, this set has Lebesgue measure zero. Furthermore, if $\payoff \in \Generic$, then the ergodic equation has a single solution (up to adding a constant to the bias), and each maximum and minimum in the ergodic equation is achieved by a single edge.
\end{proposition}

This motivates the introduction of the following \emph{condition number} $\Cond$, which measures the difference between the edge that achieves a maximum or minimum in the ergodic equation and the ``second best'' edge, relatively to the spread of the weights around the value:

\begin{definition}
Given $\payoff \in \Generic$, we put
\[
\Cond(\payoff) \coloneqq \frac{\max\{\abs{\payoff_{ij} - \gameval} \colon (i,j) \in \edges\}}{\min\{\abs{\payoff_{ij} - \gameval + \bias_j - \bias_i} \colon (i,j) \in \edges, \payoff_{ij} - \gameval + \bias_j - \bias_i \neq 0\} } \, .
\]
\end{definition}

When defined in such way, the condition number does not change when the weights are multiplied by a positive constant, or when the same constant is added to all the weights.

To analyze the behavior of the random games, we introduce the following random variables. If $i$ is a vertex controlled by Min, then for any edge $(i,j) \in \edges$ we put
\begin{align*}
Z_{ij} = \inf\{x \in \R \colon &\text{$(x,\payoff_{-ij}) \in \Generic$ and the MPG with weights $(x,\payoff_{-ij})$ has a pair }\\
&\text{of bias-induced policies $(\sigma,\tau) \in \Xi$ such that $\tau(i) \neq j $}\} \, .
\end{align*}
Here, $(x,\payoff_{-ij})$ is the vector obtained from $\payoff$ by replacing the $ij$th coordinate with $x$. We analogously define the variables $Z_{ij}$ for vertices controlled by Max, changing $\inf$ to $\sup$. Since the variable $Z_{ij}$ does not depend on $\payoff_{ij}$, we get the estimate

\begin{lemma}
For any $\alpha > 0$, $\Prob(\exists ij, \abs{\payoff_{ij} - Z_{ij}} \le \alpha) \le 2\alpha m \phi$.
\end{lemma}

Furthermore, the variables $Z_{ij}$ are related to the ergodic equation in the following way.

\begin{lemma}
Suppose that $\payoff \in \polyh^{\sigma,\tau}$ for some $(\sigma,\tau) \in \Xi$. Then, for every $(i,j) \in \edges$ that is not used in $(\sigma,\tau)$ we have $Z_{ij} = \gameval + \bias_i - \bias_j$.
\end{lemma}

The two lemmas above improve the conclusion of \cref{intro_polyhedral_complex}: not only each maximum and minimum in the ergodic equation is achieved by a single edge, but with high probability the difference between the best edge and the second best edge is large. In particular, this shows that bias-induced policies do not change when the random weights are truncated, and it gives an estimate of the condition number.

\begin{theorem}\label{intro_robust_bias_policy}
Let $\delta \coloneqq 1/(4n(2n+1)m\phi)$. Then, with probability at least $1 - 1/n$, the whole $\ell_\infty$ ball $B_{\infty}(\payoff, \delta)$ is included in a single polyhedron $\polyh^{\sigma,\tau}$.
\end{theorem}

\begin{theorem}\label{intro_cond_estimate}
Random mean payoff games are well conditioned with high probability. More formally, for every $\varepsilon > 0$ we have $\Prob\bigl(\Cond \ge \frac{8m}{\varepsilon}(\phi + \sqrt{\frac{2m}{\varepsilon}})\bigr) \le \varepsilon$.
\end{theorem}

To propose an algorithm that exploits the condition number, we use the fact that every mean-payoff game has a pair of \emph{Blackwell-optimal} policies, i.e., policies that are optimal for all discount factors $\disc$ close to $1$. Such policies are induced by the \emph{Blackwell bias}, which is defined as $\bias^* \coloneqq \lim_{\disc \to 1}(\gameval^{(\disc)} - \gameval)/(1 - \disc)$, where $\gameval^{(\disc)}$ is the value of the discounted game with discount factor $\disc$. We then show that, for well-conditioned games, the Blackwell-optimal policies can be already found when the discount factor is low.

\begin{theorem}\label{intro_conditioned_discount}
Suppose that $\payoff \in \polyh^{\sigma,\tau}$ and fix $1 > \disc > 1 - \frac{1}{6n^2\Cond(\payoff)}$. Then, $(\sigma,\tau)$ is the unique pair of optimal policies in the discounted game with discount factor $\disc$. 
\end{theorem}

Combining \cref{intro_cond_estimate,intro_conditioned_discount} with the results of \cite{ye2011simplex,hansen2013strategy} showing that policy iteration has polynomial complexity for any fixed discount rate, we get our final result.

\begin{theorem}
The greedy-all switches policy iteration rule combined with increasing discount factor solves random instances of deterministic discounted or mean-payoff games in polynomial smoothed complexity.
\end{theorem}

In the theorem above, ``polynomial smoothed complexity'' is defined as in \cite{beier2004typical,roglin2007smoothed}: there exists a polynomial $\poly(x_1,x_2,x_3,x_4)$ such that for all $\varepsilon \in \interval[open left]{0}{1}$ the probability that number of iterations of our algorithm exceeds $\poly(n,m,\phi,\frac{1}{\varepsilon})$ is at most $\varepsilon$.

\newpage


\section{Toy example: deterministic Markov Decision Processes}\label{sec:one_player}

A deterministic Markov Decision Process (DMDP) is a deterministic single-player game, and here we will consider its mean-payoff variant. As such, the game is played on a weighted digraph $\dgraph = ([n],\edges, \payoff)$ with $m \coloneqq |\edges|$ edges. The \emph{mean weight} of a cycle of $\dgraph$ is the sum of the weights of the cycle's edges, divided by the number of edges in the cycle. Given an initial state $i \in [n]$, the objective of the player is to find a cycle of minimal mean weight that is reachable from $i$ and a path from $i$ that leads to such a cycle. Let us here focus on the case of graphs $\dgraph$ with one strongly connected component (and such that each vertex has at least one outgoing edge). In this case, the weight of an optimal cycle does not depend on the initial state $i \in [n]$, since the player can reach any cycle from any starting position. We may then speak of \emph{the} mean weight of an optimal (minimum mean weight) cycle, without reference to $i$. This is the \emph{value} of the game $\dgraph$.

A \emph{policy} of a player is a function $\tau \colon [n] \to [n]$ such that $(i,\tau(i)) \in \edges$ for all $i$.  A policy $\tau$ is \emph{optimal} if an optimal cycle is reached from every starting position if the player follows $\tau$.

The aim of this introductory section is to prove a simplified version of our proof that is adapted to the one-player case. Our main claim is as follows: when the weights of the graph are picked at random, then the optimal cycle is robust with respect to any perturbation of the weights within a ``large'' radius. This claim can also be deduced from \cite[Lemma~5]{BorosElbassioniFouzGurvich:2011}. We give a full proof of the claim to illustrate our approach. In this section, we assume that the weights $(\payoff_{ij})$ are absolutely continuous independent random variables and we denote by $\dens_{ij}$ the density function of $\payoff_{ij}$. We further assume that there exists $\phi > 0$ such that $\dens_{ij}(y) \le \phi$ for all $i,j$ and $y$.

\begin{definition}\label{def:polyhedral_cells_one_player}
  Given a cycle $\cycle \subseteq \edges$ we denote
  \[
    \polyh^{\cycle} \coloneqq \{\payoff \in \R^{m} \colon \text{$\cycle$ is the unique optimal cycle in a DMDP with weights $\payoff$} \} \, .
  \]
  We also set $\Generic \coloneqq \cup_{\cycle} \polyh^{\cycle}$.
\end{definition}

\begin{lemma}\label{polyhedral_cells_one_player}
The sets $\polyh^{\cycle}$ are disjoint and open polyhedral cones. Furthermore, $\R^{m} \setminus \Generic$ is included in a finite union of hyperplanes. In particular, it is a set of Lebesgue measure zero. 
\end{lemma}
\begin{proof}
The fact that the sets $\polyh^{\cycle}$ are disjoint follows from the definition. They are open polyhedral cones because they are defined by strict affine inequalities of the form $\frac{1}{\card{\cycle}}\sum_{ij \in \cycle} \payoff_{ij} < \frac{1}{\card{\cycle'}}\sum_{ij \in \cycle'} \payoff_{ij}$, where $\cycle' \neq \cycle$ is a cycle of $\dgraph$. Moreover, if $\payoff \in \R^{m} \setminus \Generic$, then the DMPD with weights $\payoff$ has at least two optimal cycles, so $\payoff$ lies on a nontrivial hyperplane of the form $\{ x \in \R^{m} \colon \frac{1}{\card{\cycle}}\sum_{ij \in \cycle} x_{ij} = \frac{1}{\card{\cycle'}}\sum_{ij \in \cycle'} x_{ij}\}$.
\end{proof}

In particular, \cref{polyhedral_cells_one_player} implies that in our random model, the probability that two different cycles are optimal at the same time is zero. 

It is useful to use the following notation. Given a vector of weights $\payoff \in \R^{m}$ and $(i,j) \in \edges$, we denote by $\payoff_{-ij} \in \R^{m - 1}$ the vector obtained from $\payoff$ by removing the $ij$th coordinate. We also denote by $(x,\payoff_{-ij}) \in \R^{m}$ the vector obtained from $\payoff$ by replacing the $ij$th coordinate with $x$. We further denote by $\gameval(\payoff)$ the value of the DMDP with weights $\payoff$ and by $\gameval^{\cycle}(\payoff)$ the mean weight of the cycle $\cycle$ in the DMDP with weights $\payoff$. We now give the most important lemma for our analysis.

\begin{lemma}\label{one_player_concavity}
For any $\payoff_{-ij}$, the function $\R \ni x \mapsto \gameval(x,\payoff_{-ij}) \in \R$ is continuous, piecewise affine, and has at most $n$ break points.
\end{lemma}
\begin{proof}
We have $\gameval(x,\payoff_{-ij}) = \min_{\cycle} \gameval^{\cycle}(x,\payoff_{-ij})$. Therefore, the function $\gameval(\cdot,\payoff_{-ij})$ is the minimum of some collection of univariate affine functions. In particular, $\gameval(\cdot,\payoff_{-ij})$ is continuous, piecewise affine, and concave. Concavity implies that the slopes of the different affine pieces of this function are decreasing. Since these slopes are in the set $\{1,1/2,\dots,1/n,0\}$, the function has at most $n+1$ affine pieces and at most $n$ break points.
\end{proof}

For $k \in \{1,2,\dots,n\}$, let $X_{ij}^{(k)}(\payoff) \in \R$ denote the $k$th break point of the function $x \mapsto \gameval(x,\payoff_{-ij})$, with the convention that $X_{ij}^{(k)}(\payoff) = 0$ if this function has less than $k$ break points. By a reasoning similar to the MVV isolation lemma, we can now prove that it is unlikely that the chosen weights are close to one of the breakpoints.

\begin{lemma}\label{prob_estimate_one_player}
For any $\alpha > 0$ we have 
\[
\Prob(\exists (i,j) \in \edges, \, \exists k \in [n], \abs{\payoff_{ij} - X_{ij}^{(k)}} \le \alpha) \le 2\alpha nm\phi \, .
\]
\end{lemma}
\begin{proof}
Note that the random variable $X_{ij}^{(k)}$ depends only on $\payoff_{-ij}$, not on $\payoff_{ij}$. Let $\hat{\dens}(\payoff_{-ij}) = \prod_{(i',j') \in \edges \setminus \{(i,j)\}} \dens_{i'j'}(\payoff_{i',j'})$. By Fubini's theorem we have
\begin{align*}
\Prob(X_{ij}^{(k)} - \alpha \le \payoff_{ij} \le X_{ij}^{(k)} + \alpha) &= \int_{\R^{m-1}}\Bigl( \int_{X_{ij}^{(k)} - \alpha}^{X_{ij}^{(k)} + \alpha} \dens_{ij}(\payoff_{ij}) d \payoff_{ij} \Bigr) \hat{\dens}(\payoff_{-ij}) d \payoff_{-ij} \\
&\le 2\alpha \phi \int_{\R^{m-1}}\hat{\dens}(\payoff_{-ij}) d \payoff_{-ij} = 2\alpha \phi \, .
\end{align*}
The claim follows from the union bound.
\end{proof}

Let $B_1(\payoff, \delta) \coloneqq \{\tilde{\payoff} \in \R^{m} \colon \onenorm{\payoff - \tilde{\payoff}} \le \delta \}$ denote the closed ball around $\payoff$ with radius $\delta$ in the $\ell_1$ norm. We recall that such ball is a polytope. More precisely, $B_1(\payoff, \delta)$ is the convex hull of the points $\{\payoff \pm \delta e_{ij}\}_{ij}$, where $e_{ij}$ denotes the $ij$th vector of the standard basis in $\R^m$. We are now ready to prove the main proposition of this section. 

\begin{proposition}[Robustness of optimal cycle, proof I]\label{robust_cycle_one_player}
Let $\delta \coloneqq 1/(2n^2m\phi)$. If the weights $\payoff$ are picked at random, then with high probability the whole ball $B_1(\payoff, \delta)$ is included in a single polyhedron $\polyh^{\cycle}$. More formally,
\[
\Prob(\exists \cycle, B_1(\payoff, \delta) \subseteq \polyh^{\cycle}) \ge 1 - \frac{1}{n} \, .
\]
\end{proposition}
\begin{proof}
By applying \cref{prob_estimate_one_player} with $\alpha \coloneqq \delta$ we get
\[
\Prob(\forall (i,j) \in \edges, \, \forall k \in [n], \abs{\payoff_{ij} - X_{ij}^{(k)}} > \delta) \ge 1 - \frac{1}{n} \, .
\]
Let $\payoff \in \Generic$ be such that $\abs{\payoff_{ij} - X_{ij}^{(k)}} > \delta$ for all $(i,j)$ and all $k$. Let $\cycle$ denote the cycle such that $\payoff \in \polyh^{\cycle}$ and let $I = I_{ij} \coloneqq \interval[open]{a}{b}$ be the open interval of maximal length such that $\payoff_{ij} \in I$ and $\gameval(\cdot, \payoff_{-ij})$ is affine on $I$. Note that such interval $I$ exists because $\payoff_{ij}$ is not one of the break points of this function. Furthermore, we have $\interval{\payoff_{ij} - \delta}{\payoff_{ij} + \delta} \subset I$. Moreover, by definition, $\gameval(\payoff_{ij}, \payoff_{-ij}) = \gameval^{\cycle}(\payoff_{ij},\payoff_{-ij}) < \gameval^{\cycle'}(\payoff_{ij},\payoff_{-ij})$ for every cycle $\cycle'$ of $\dgraph$ that is different from $\cycle$. By continuity, the equality $\gameval(z, \payoff_{-ij}) = \gameval^{\cycle}(z,\payoff_{-ij})$ holds on some neighborhood of $\payoff_{ij}$. Since $\gameval(\cdot,\payoff_{-ij})$ is piecewise affine and $\gameval^C(\cdot,\payoff_{-ij})$ is affine, this implies that the equality $\gameval(z, \payoff_{-ij}) = \gameval^{\cycle}(z,\payoff_{-ij})$ holds for every $z \in I$. Furthermore, we have $\gameval^{\cycle}(z,\payoff_{-ij}) < \gameval^{\cycle'}(z,\payoff_{-ij})$ for every $\cycle' \neq \cycle$ and $z \in I$. Indeed, if $\gameval^{\cycle}(z,\payoff_{-ij}) = \gameval^{\cycle'}(z,\payoff_{-ij})$ for some $z \in I$, then the two lines $\{\bigl(x, \gameval^{\cycle}(x,\payoff_{-ij})\bigr) \colon x \in \R\}$, $\{\bigl(x, \gameval^{\cycle'}(x,\payoff_{-ij})\bigr) \colon x \in \R\}$ either intersect transversally at $z$ or are identical. The former possibility cannot hold because it would imply that $\gameval^{\cycle'}(x,\payoff_{-ij}) < \gameval^{\cycle}(x,\payoff_{-ij})= \gameval(x, \payoff_{-ij})$ for some $x \in I$ close to $z$. The latter possibility cannot hold because $\gameval^{\cycle'}(\payoff_{ij},\payoff_{-ij}) \neq \gameval^{\cycle}(\payoff_{ij},\payoff_{-ij})$. Hence, we get that $(y,\payoff_{-ij}) \in \polyh^{\cycle}$ for every $y \in \interval{\payoff_{ij} - \delta}{\payoff_{ij} + \delta}$. Since $(i,j)$ was arbitrary, every point of the form $\payoff \pm \delta e_{ij}$ belongs to $\polyh^{\cycle}$. By convexity of $\polyh^{\cycle}$, the convex hull of these points is included in $\polyh^{\cycle}$. This convex hull is exactly $B_1(\payoff, \delta)$, which finishes the proof.
\end{proof}

\Cref{robust_cycle_one_player} shows the robustness property of the optimal cycle in the randomized setting: if we pick the weights at random, then we can perturb them to $\frac{1}{\poly(n,\phi)}$ and the optimal cycle does not change. The main technical result of this paper shows that a similar property is still true in the case of two-player games. One could try to naively generalize the proof above to the two-player case by replacing the ``optimal cycle'' $\cycle$ with an ``equilibrium cycle'' $\cycle$, where the equilibrium cycle is the cycle that is reached by the players when they use their optimal policies. The first (and less important) issue with this approach is that the sets $\polyh^{\cycle}$ are no longer polyhedral cones. The second (and major) issue is that the value function $x \mapsto \gameval(x,\payoff_{-ij})$ is no longer convex or concave, and so we cannot prove an analogue of \cref{one_player_concavity}. Indeed, in the case of two-player mean-payoff games, \cref{ex:non_convex} gives a two-player game where the function $x \mapsto \gameval(x,\payoff_{-ij})$ is non-convex, and \cref{ex:exponential} shows that the number of break points of this function can be exponential.

\medskip
To overcome this issue let us present a second proof of the robustness result in the one-player case. This proof is may be less intuitive, but its advantage is that it does not involve \cref{one_player_concavity}. The proof relies on the following quantity. Given an edge $(i,j) \in \edges$, we denote 
\begin{align*}
Y_{ij}(\payoff) \coloneqq \inf\{x \in \R \colon &\text{there exists an optimal cycle in the DMPD with weights } \\ &\text{$(x,\payoff_{-ij})$  that does not contain the edge $(i,j)$} \} \, .
\end{align*}

We note that $Y_{ij} = -\infty$ if and only if the edge $(i,j)$ does not belong to the strongly connected component of $\dgraph$. Indeed, an edge that does not belong to this component can never be a part of the optimal cycle, and an edge $(i,j)$ that belongs to this component is a part of at least one cycle. For a sufficiently small $x$, the mean weight of this cycle is smaller that the mean weight of any cycle that does not contain $(i,j)$, so all optimal cycles contain $(i,j)$ and therefore $x < Y_{ij}$.  Analogously, for a large $x$ the mean weight of any cycle that contains $(i,j)$ is higher than the mean weight of any cycle that does not contain $(i,j)$. Therefore, $x > Y_{ij}$ if at least one of the cycles does not contain $(i,j)$. Hence, we have $Y_{ij} = +\infty$ if and only if all cycles of the graph contain the edge $(i,j)$. If $Y_{ij} \in \R$, then the number $Y_{ij}$ splits the real line into two parts: for all $x > Y_{ij}$, there is an optimal cycle of DMDP with payoff $(x,\payoff_{-ij})$ that does not involve the edge $(i,j)$. In particular, a cycle that has the minimal mean weight among the cycles that do not contain $(i,j)$ is optimal for all $x > Y_{ij}$. For all $x < Y_{ij}$, every optimal cycle contains the edge $(i,j)$. The optimal cycle can change as $x$ decreases. In essence, \cref{one_player_concavity} shows that the optimal cycle can change at most $n$ times, but the proof below does not use this property.

\begin{lemma}\label{prob_estimate_one_player_second}
For any $\alpha > 0$ we have
\[
\Prob(\exists ij, Y_{ij} \le \payoff_{ij} \le Y_{ij} + \alpha) \le \alpha m \phi \, .
\]
\end{lemma}
\begin{proof}
Fix an edge $(i,j)$ and note that $Y_{ij}$ depends only on $\payoff_{-ij}$ and not on $\payoff_{ij}$. If $(i,j)$ is not in the strongly connected component, then $Y_{ij} = -\infty$ and $\Prob(Y_{ij} \le \payoff_{ij} \le Y_{ij} + \alpha) = 0$. Likewise, if all cycles of the graph contain the edge $(i,j)$, then $Y_{ij} = +\infty$ and $\Prob(Y_{ij} \le \payoff_{ij} \le Y_{ij} + \alpha) = 0$. Otherwise, as discussed above, we have $Y_{ij} \in \R$ and the Fubini theorem gives $\Prob(Y_{ij} \le \payoff_{ij} \le Y_{ij} + \alpha) \le \alpha \phi$ just as in the proof of \cref{prob_estimate_one_player}. The union bound finishes the proof.
\end{proof}

In the second proof of robustness, we replace the ball in the $\ell_{1}$ norm by a ball in the $\ell_{\infty}$ norm, $B_{\infty}(\payoff, \delta) \coloneqq \{\tilde{\payoff} \in \R^{m} \colon \supnorm{\payoff - \tilde{\payoff}} \le \delta \}$.

\begin{proposition}[Robustness of optimal cycle, proof II]\label{robust_cycle_one_player_second}
Let $\delta \coloneqq 1/(3n^2m\phi)$. If the weights $\payoff$ are picked at random, then with high probability the whole ball $B_{\infty}(\payoff, \delta)$ is included in a single polyhedron $\polyh^{\cycle}$. More formally,
\[
\Prob(\exists \cycle, B_{\infty}(\payoff, \delta) \subseteq \polyh^{\cycle}) \ge 1 - \frac{1}{n} \, .
\]
\end{proposition}
\begin{proof}
By applying \cref{prob_estimate_one_player_second} with $\alpha \coloneqq 1/(nm\phi)$ we get
\[
\Prob(\forall (i,j) \in \edges, \payoff_{ij} < Y_{ij} \lor \payoff_{ij} > Y_{ij} + \alpha) \ge 1 - \frac{1}{n} \, .
\]
Let $\payoff \in \Generic$ be such that this event holds. Let $\cycle$ be the cycle such that $\payoff \in \polyh^{\cycle}$ and $\cycle' \neq \cycle$ be any other cycle. We will prove that 
\[
(\text{mean weight of $\cycle$}) \le (\text{mean weight of $\cycle'$}) - \frac{\alpha}{n} \, .
\]
To do so, let $(i,j)$ be an edge that is in $\cycle'$ but not in $\cycle$. Then, by definition we have $\payoff_{ij} \ge Y_{ij}$, so that $\payoff_{ij} > Y_{ij} + \alpha$. In particular, if we replace $\payoff_{ij}$ by $\tilde{\payoff}_{ij} \coloneqq \payoff_{ij} - \alpha$, then $\cycle$ is still an optimal cycle. Indeed, we have $\tilde{\payoff}_{ij} > Y_{ij}$, so there exists an optimal cycle that does not contain the edge $(i,j)$ in the graph with modified weight. In particular, the mean weight of this cycle is the same as in the graph with original weights. Therefore, the same cycle $C$ is optimal with the modified weights. The optimality of $\cycle$ gives
\begin{align*}
  (\text{mean weight of $\cycle$ with weights $r$})
  & = (\text{mean weight of $\cycle$ with weights $\tilde r$}) \\
  & \le (\text{mean weight of $\cycle'$ with weights $\tilde r$})\\
  & = (\text{mean weight of $\cycle'$ with weights $r$}) - \frac{\alpha}{|C'|}\\
  & \le (\text{mean weight of $\cycle'$ with weights $r$}) - \frac{\alpha}{n} \, . %
\end{align*}
To finish the proof, note that if $\tilde{\payoff} \in \R^m$ is such that $\supnorm{\tilde{\payoff} - \payoff} \le \delta = \frac{\alpha}{3n}$, then the mean weight of any cycle in the graph with weights $\tilde{\payoff}$ differs by at most $\frac{\alpha}{3n}$ from its mean weight in the graph with weights $\payoff$. Hence $\tilde{\payoff} \in \polyh^{\cycle}$.
\end{proof}


\section{Deterministic mean-payoff games}\label{sec:det_mpg}

\subsection{Definition of mean-payoff games}
A \emph{(deterministic, two-player) mean-payoff game (MPG)} generalizes the notion of the deterministic Markov Decision Process discussed in \cref{sec:one_player}. It is a game played on directed weighted graph $\dgraph = ([n], \edges, \payoff)$ in which the vertices are split into two disjoint sets, $[n] = \Maxvertices \dunion \Minvertices$. We say that $\Maxvertices$ is the set of vertices controlled by \emph{player Max} and $\Minvertices$ is the set of vertices controlled by \emph{player Min}. One of these sets can be empty (in this case, we get a one-player game, i.e., a DMDP). Furthermore, we denote by $m \coloneqq \card{\edges}$ the number of edges of the graph. The game is played as follows. At the beginning of the game, a token is placed on one of the vertices $i_0 \in [n]$. Subsequently, the player controlling that vertex chooses an outgoing edge $(i_0, i_1) \in \edges$ and the token is moved to $i_1$. Then, the player controlling $i_1$ chooses the next move of the token, and this procedure is repeated ad infinitum. The final payoff received by player Max is defined as
\[
\liminf_{N \to \infty} \frac{1}{N}(\payoff_{i_0i_1} + \payoff_{i_1i_2} + \dots + \payoff_{i_{N-1}i_N}) \, .
\]
The objective of player Max is to maximize the payoff, while the objective of player Min is to minimize the payoff. This means that the game is of zero sum. Moreover, the players try to optimize the average payoff of the game in the long run, which justifies the name ``mean-payoff game''. The game is well-defined provided that each player can always make the next move, i.e., that the graph satisfies the following assumption.

\begin{assumption}
Every vertex of $\dgraph$ has at least one outgoing edge.
\end{assumption}

In the remaining part of the paper, we always assume that $\dgraph$ satisfies the assumption above. We further assume that $\dgraph$ has no multiple edges, but it can have loops. The weights $\payoff \in \R^m$ of the graph are often called \emph{payoffs} because they correspond to immediate payoffs received by player Max in the game and the final payoff of player Max is often called the \emph{reward}. However, since in this work we study what happens when $\payoff$ changes or is picked at random, for the sake of clarity we prefer to use a more distinct terminology: we refer to $\payoff$ as the \emph{weights} of the underlying graph, and we reserve the word ``payoff'' for the final outcome of the game. 

There are many ways to play a mean-payoff game, but the most relevant to us is to consider the situation in which players use policies. A \emph{policy of player Max} is a function $\sigma \colon  \Maxvertices \to [n]$ that associates to a given vertex controlled by Max a possible transition, i.e., satisfies $(i,\sigma(i)) \in \edges$ for all $i \in \Maxvertices$. Analogously, a \emph{policy of player Min} is a function $\tau \colon \Minvertices \to [n]$ such that $(i,\tau(i)) \in \edges$ for all $i \in \Minvertices$. Policies are \emph{memoryless} strategies because they do not depend on the history of the game and are \emph{stationary} because they only depend on the current state of the game. If both players play according to policies $(\sigma,\tau)$, then the token reaches some cycle $(i_k, i_{k+1}, \dots, i_{k+l-1})$ of the graph and stays there forever. In particular, the resulting payoff of Max is given by
\[
g_{i_{0}}(\sigma,\tau) = \frac{1}{l}(\payoff_{i_ki_{k+1}} + \dots + \payoff_{i_{k+l-1}i_k}) \, .
\]
In other words, the payoff is the mean weight of the final cycle reached by the token. It turns out that mean-payoff games have a value and a pair of optimal policies.

\begin{theorem}[\cite{liggett_lippman,ehrenfeucht_mycielski}]\label{th:value_mpg}
For any mean-payoff game there exists a \emph{unique} vector $\gameval \in \R^n$ and a (not necessarily unique) pair of policies $(\sigma^*,\tau^*)$ such that for every pair of policies $(\sigma,\tau)$ and every initial state $i \in [n]$ we have
\begin{equation}\label{eq:optimal_policies}
g_i(\sigma,\tau^*) \le \gameval_i \le g_i(\sigma^*,\tau) \, .
\end{equation}
\end{theorem}

\begin{definition}
  The unique vector $\gameval$ given above is called the \emph{value} of the game. The $i$th coordinate of $\gameval$ is called the \emph{value of state $i$}.
  Any policies $(\sigma^*,\tau^*)$ satisfying \cref{eq:optimal_policies} for all $i,\sigma,\tau$ are called \emph{optimal}. We will call \emph{equilibrium cycle} any cycle which is reached by a pair of optimal policies.
\end{definition}

We should emphasize that a mean-payoff game can have more than one pair of optimal policies, with different equilibrium cycles.

An intuitive interpretation of optimality is as follows: if player Max uses an optimal policy $\sigma^*$, then they are guaranteed to receive at least $\gameval_i$, no matter what Min does. Analogously, if Min uses an optimal policy $\tau^*$, then they are guaranteed that the final payoff is at most $\gameval_i$, no matter what Max does. The theorem above can actually be stated in a stronger way: optimal policies are optimal against \emph{any} possible strategy of the opponent. In other words, by using $\sigma^*$, Max is guaranteed a payoff at least $\gameval_i$ even if Min uses an arbitrarily complicated strategy that depends on the history of the game. See \cite{payoffchapter} for a general introduction to mean-payoff games. 

\begin{remark}
In this paper, we often work with games in which all states have the same value, so that $\gameval = (\gameval_1, \dots, \gameval_1)$ for some $\gameval_1 \in \R$. In this case, it is convenient to refer to $\gameval_1$ as the value of the whole game, even though the value is formally defined as a vector and not as a real number. We follow this convention throughout the paper.
\end{remark}

We can now try to generalize our approach from \cref{sec:one_player} to two-player games. \cref{ex:non_convex,ex:exponential} below will show that \cref{one_player_concavity} does not hold in the two-player setting, and hence the proof of \cref{robust_cycle_one_player} cannot be repeated. So we would like to generalize the proof of \cref{robust_cycle_one_player_second}. However, in order to do so, we will need to keep in mind the numerous differences between the one-player DMDPs and two-player MPGs, which we illustrate in \cref{ex:emerging_policies,ex:unstable_policies_good_approx,ex:non_convex_again,ex:unstable_policies_again}.

\subsection{A two-player mean-payoff game with exponential number of breakpoints}

\begin{figure}[t]
\centering
\begin{tikzpicture}[scale=0.75,>=stealth',max/.style={draw,rectangle,minimum size=0.5cm},min/.style={draw,circle,minimum size=0.5cm},av/.style={draw, circle,fill, inner sep = 0pt,minimum size = 0.2cm}]

\node[min] (min1) at (-3, 0) {$1$};
\node[min] (min2) at (3, 0) {$2$};
\draw[->] (min1) to node[above, font=\small]{$x$} (min2);

\node[min] (min3) at (2,2) {$3$};
\node[min] (min4) at (-2,2) {$4$};

\draw[->] (min2) to (min3);
\draw[->] (min3) to (min4);
\draw[->] (min4) to (min1);

\node[max] (max1) at (0,-2) {$5$};
\draw[->] (min2) to (max1);
\draw[->] (max1)[out=-20,in=-110] to node[right=1ex, font=\small]{$2$} (min2);
\draw[->] (max1) to node[below left, font=\small]{$-3$} (min1);

\end{tikzpicture}
\caption{Example of a mean payoff game with non-convex value function. Nodes controlled by Min are depicted by circles and the node controlled by Max is depicted by a square. Edges without numbers have weights $0$.} \label{fig:nonconvex_value}
\end{figure}
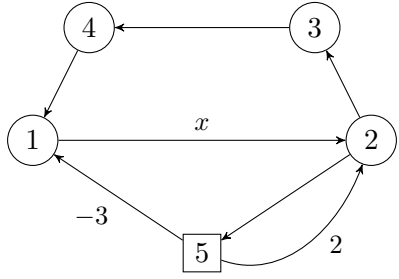

\begin{example}\label{ex:non_convex}
Consider the two-player game depicted in \cref{fig:nonconvex_value}. One can check that the value function of this game is given by
\[
\gameval(x) = \min\bigl\{\frac{x}{4}, \max\{1, \frac{x}{3} - 1\} \bigr\} \, .
\]
In particular, this function is neither concave nor convex. Furthermore, the cycle going through nodes $\{1,2,3,4\}$ is the unique equilibrium cycle for $x < 4$ and $x > 12$, but not for the intermediate values. In particular, the set of all $x$ such that this cycle is an equilibrium cycle is not convex.
\end{example}

\begin{figure}[t]
\centering 
\begin{tikzpicture}[scale=0.75,>=stealth',max/.style={draw,rectangle,minimum size=0.6cm},min/.style={draw,circle,minimum size=0.6cm},av/.style={draw, circle,fill, inner sep = 0pt,minimum size = 0.2cm}]

\node[min] (leftmost) at (-7, 1){$2n$};

\node[min] (leftmin) at (-5, 2){};
\node[max] (leftmax) at (-5, 0){};

\node[min] (min8) at (-3, 2){$8$};
\node[max] (max9) at (-3, 0){$9$};

\node[min] (min6) at (-1, 2){$6$};
\node[max] (max7) at (-1, 0){$7$};

\node[min] (min4) at (1, 2){$4$};
\node[max] (max5) at (1, 0){$5$};

\node[min] (min2) at (3, 2){$2$};
\node[max] (max3) at (3, 0){$3$};

\node[max] (rightmost) at (5, 1){$1$};

\draw[->] (rightmost) to[out =-40, in = 220] (leftmost);
\draw[->] (max3) to (rightmost);
\draw[->] (min2) to (rightmost);

\draw[->] (leftmost) to node[below left,font=\tiny]{1} (leftmax);
\draw[->] (leftmost) to node[above left,font=\tiny]{1} (leftmin);

\draw[->] (max5) to node[below,font=\tiny]{1} (max3);
\draw[->] (max5) to node[above,at start,font=\tiny]{1} (min2);

\draw[->] (min4) to node[above,font=\tiny]{$x$} (min2);
\draw[->] (min4) to node[below,at start,font=\tiny]{2} (max3);

\draw[->] (max7) to node[below,font=\tiny]{$1-\frac{1}{2}$} (max5);
\draw[->] (max7) to node[above,at start,xshift=-7pt,font=\tiny]{$1-\frac{1}{2}$} (min4);

\draw[->] (min6) to node[above,font=\tiny]{$1$} (min4);
\draw[->] (min6) to node[below,at start,font=\tiny]{1} (max5);

\draw[->] (max9) to node[below,font=\tiny]{$1-\frac{1}{4}$} (max7);
\draw[->] (max9) to node[above,at start,xshift=-7pt,font=\tiny]{$1-\frac{1}{4}$} (min6);

\draw[->] (min8) to node[above,font=\tiny]{$1$} (min6);
\draw[->] (min8) to node[below,at start,font=\tiny]{1} (max7);

\node at (-4,2) {\ldots};
\node at (-4,0) {\ldots};

\end{tikzpicture}
\caption{Example of a mean payoff game in which the optimal policies change exponentially many times.} \label{fig:exponential}
\end{figure}
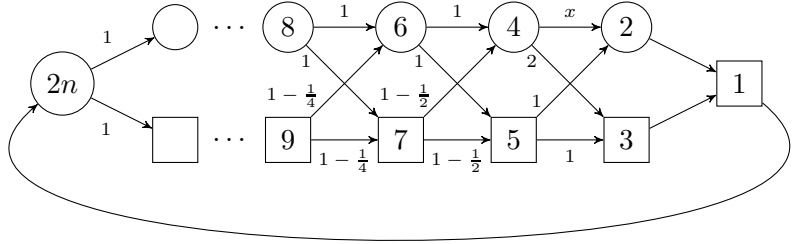

We thank Matthew Maat for sharing with us a variant of the following exponential example.
\begin{example}\label{ex:exponential}
Consider the game depicted in \cref{fig:exponential} with $2n$ states, and $\payoff_{1,2n} = 0$. This game repeats every $n+1$ turns, so the value of the mean-payoff game is equal to $\chi/(n+1)$, where $\chi$ is the value of the game that starts at state $2n$ and in which the aim of Max is to maximize the total payoff before the state $1$ is reached. Hence, the value and optimal policies are found by solving the recurrence
\begin{align*}
y_1 &= 0,  \ y_2 = \payoff_{21}, \ y_{3} = \payoff_{31}, \\
y_{2k} &= \min\{\payoff_{2k,2k-1} + y_{2k-1},\payoff_{2k,2k-2}+y_{2k-2} \} \quad \text{for $2 \le k \le n$},\\
y_{2k+1} &= \max\{\payoff_{2k+1,2k-1} + y_{2k-1},\payoff_{2k+1,2k-2}+y_{2k-2} \} \quad \text{for $2 \le k \le n-1$}.
\end{align*}
We then have $\chi = y_{2n}$. By putting $\payoff_{21} = \payoff_{31} = 0$, $\payoff_{53} = \payoff_{52} = 1$, $\payoff_{43} = 2$, parametrizing the edge $(4,2)$ by setting $\payoff_{42} = x$, and putting the remaining weights as $\payoff_{2k,2k-1} = \payoff_{2k,2k-2} = 1$ and $\payoff_{2k+1,2k-1} = \payoff_{2k+1,2k-1} = 1 - 1/2^{k-2}$, we get exactly the equations of the tropical central path from \cite[Proposition~20]{log_barrier}. In particular, the curve $y_{2n}(x)$  has exponentially many breakpoints when $x \in \interval{0}{2}$.  Furthermore, the $(4,2)$ edge alternates an exponential number of times between being included and excluded from the equilibrium cycle.
\end{example}

\subsection{Further differences between one-player and two-player games}

We now point out three natural two-player analogues of the properties proven in \cref{robust_cycle_one_player,robust_cycle_one_player_second}:
\begin{enumerate}[label=(P-\arabic*),ref=(P-\arabic*)]
\item\label{enum:first}  The game $\dgraph$ has a unique equilibrium cycle, which does not change under any perturbation of the weights $\tilde{\payoff} \in B_{\infty}(\payoff,\delta)$.
\item\label{enum:second} The game $\dgraph$ has a unique equilibrium cycle, and any pair $(\sigma,\tau)$ of optimal policies for $\dgraph$ remains optimal under any perturbation of the weights $\tilde{\payoff} \in B_{\infty}(\payoff,\delta)$.
\item\label{enum:third} The game $\dgraph$ has a unique equilibrium cycle, and if two policies $(\sigma,\tau)$ are optimal for $\dgraph$ with perturbed weights $\tilde{\payoff} \in B_{\infty}(\payoff,\delta)$, they are also optimal for $\dgraph$ with the original weights $\payoff$.
\end{enumerate}
For any single-player game $\dgraph$ where the optimal cycle is unique, a policy is optimal if and only if it leads every vertex into that cycle. So it is not hard to see that the above properties are all equivalent in the single-player case (with \emph{equilibrium cycle} replaced by \emph{optimal cycle}).  In particular, they are all satisfied with high probability thanks to \cref{robust_cycle_one_player,robust_cycle_one_player_second}. The first subtle difficulty in the two-player case arises from the fact that these properties are no longer equivalent.

%
%
%
%
%

\begin{figure}[t]
\centering
\begin{tikzpicture}[scale=0.75,>=stealth',max/.style={draw,rectangle,minimum size=0.5cm},min/.style={draw,circle,minimum size=0.5cm},av/.style={draw, circle,fill, inner sep = 0pt,minimum size = 0.2cm}]

\node[min] (min1) at (-3, 0) {$1$};
\node[max] (max1) at (3, 0) {$2$};

\node[min] (min2) at (0,2) {$3$};

\draw[->] (min1) to node[above, font=\small]{$-10$}  (max1);
\draw[->] (min2) to (max1);
\draw[->] (min1) to (min2);
\draw[->] (max1) to[out=200,in=-20] node[below, font=\small]{$\varepsilon$} (min1);
\draw[->] (max1) to[out=-30,in=30,looseness=8] (max1);

\end{tikzpicture}
\caption{Example of a mean payoff game in which new optimal policies appear after a small perturbation of weights.} \label{fig:emerging_policies}
\end{figure}
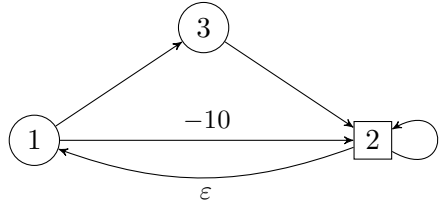

\begin{example}\label{ex:emerging_policies}
Consider the game depicted in \cref{fig:emerging_policies} and let $\payoff$ denote its weights for $0 < \varepsilon < 1$. Note that the policies that use the edges $(1,2)$ and $(2,2)$ are the unique optimal policies. Furthermore, these policies are still optimal for every $\tilde{\payoff}$ such that $\tilde{\payoff} \in B_{\infty}(\payoff,1)$, so the property \labelcref{enum:second} is satisfied. The property \labelcref{enum:first} is also satisfied, because the cycle $(2,2)$ is the only equilibrium cycle for every $\tilde{\payoff} \in B_{\infty}(\payoff,1)$. However, if we replace the weight $\varepsilon$ by $0$, then the policy $(1,3)$ becomes optimal for Min, so the property \labelcref{enum:third} is not satisfied even under a very small perturbation of the weights. Furthermore, note that if $\varepsilon = 0$, then property \labelcref{enum:first} is still satisfied, but the policy using $(1,3)$ stops being optimal for any $\varepsilon >0$, so the property \labelcref{enum:second} is not satisfied.
\end{example}

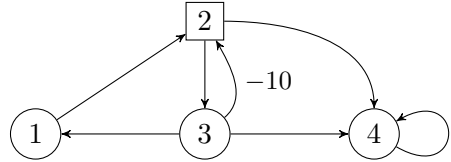
\begin{figure}[t]
\centering
\begin{tikzpicture}[scale=0.75,>=stealth',max/.style={draw,rectangle,minimum size=0.5cm},min/.style={draw,circle,minimum size=0.5cm},av/.style={draw, circle,fill, inner sep = 0pt,minimum size = 0.2cm}]

\node[min] (min1) at (-3, 0) {$1$};
\node[max] (max1) at (0,2) {$2$};

\node[min] (min2) at (0, 0) {$3$};

\node[min] (min3) at (3, 0) {$4$};

\draw[->] (min1) to (max1);
\draw[->] (min2) to (min1);
\draw[->] (max1) to (min2);
\draw[->] (max1) to[out=0,in=90] (min3);
\draw[->] (min3) to[out=-30,in=30,looseness=8] (min3);
\draw[->] (min2) to[out=40,in=-60] node[right, font=\small]{$-10$} (max1);
\draw[->] (min2) to (min3);

\end{tikzpicture}
\caption{Example of a mean payoff game in which optimal policies are not stable under perturbations.} \label{fig:unstable_policies_good_approx}
\end{figure}

\begin{example}\label{ex:unstable_policies_good_approx}
Consider the game depicted in \cref{fig:unstable_policies_good_approx}. Note that the equilibrium cycle is unique. Moreover, the policy of Min that uses the edge $(3,1)$ is optimal, but it stops being optimal if we replace the weight of the edge $(1,2)$ by $\varepsilon > 0$, so the property \labelcref{enum:second} does not hold. However, if we take any $\tilde{\payoff} \in B_{\infty}(\payoff,1)$ and any pair of optimal policies $(\sigma,\tau)$ for $\tilde{\payoff}$, then they are optimal for $\payoff$. Indeed, the unique optimal policy of Max is the one that uses $(2,4)$ and both possible policies of Min are optimal in the original game. Therefore, the property \labelcref{enum:third} is satisfied but the property \labelcref{enum:second} is not.
\end{example}

Another difficulty arises from the fact that in the two-case we lose the ``stability'' of the equilibrium cycle/optimal policy that was used in the proof of \cref{robust_cycle_one_player_second}. Indeed, this proof relies on the fact that if $\cycle$ is optimal for \emph{some} weights $(x,\payoff_{-ij})$ such that $x > Y_{ij}$, then it is optimal for \emph{all} weights $(x,\payoff_{-ij})$ such that $x > Y_{ij}$. This property is no longer satisfied in the two-player case, as shown by the next two examples.

\begin{example}\label{ex:non_convex_again}
Consider once again the game given in \cref{ex:non_convex}. For $x \in \interval[open]{4}{6}$, the equilibrium cycle does not contain the edge with weight $x$. In particular, the infimum of the values of $x$ such that the edge $(1,2)$ does not belong to the equilibrium cycle is $x = 4$. However, the equilibrium cycle changes twice for higher values of $x$, namely for $x = 6$ and $x = 12$.
\end{example}

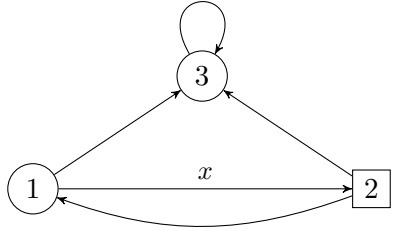
\begin{figure}[t]
\centering
\begin{tikzpicture}[scale=0.75,>=stealth',max/.style={draw,rectangle,minimum size=0.5cm},min/.style={draw,circle,minimum size=0.5cm},av/.style={draw, circle,fill, inner sep = 0pt,minimum size = 0.2cm}]

\node[min] (min1) at (-3, 0) {$1$};
\node[max] (max1) at (3, 0) {$2$};

\node[min] (min2) at (0,2) {$3$};

\draw[->] (min1) to node[above, font=\small]{$x$} (max1);
\draw[->] (max1) to (min2);
\draw[->] (min1) to (min2);
\draw[->] (max1) to[out=200,in=-20] (min1);
\draw[->] (min2) to[out=120,in=60,looseness=8] (min2);

\end{tikzpicture}
\caption{Example of a mean payoff game in which optimal policies change above the natural threshold for $x$.} \label{fig:unstable_policies}
\end{figure}

\begin{example}\label{ex:unstable_policies_again}
Consider the game depicted in \cref{fig:unstable_policies}. The value of this game is equal to $0$ for all $x \in \R$ and the policies that use the edges $(1,3)$ and $(2,3)$ are always optimal. In particular, the infimum of the values of $x$ such that there exists an optimal policy of Min which does not use $(1,2)$ is $x = -\infty$. However, if we look at a pair of policies that use the edges $(1,3)$ and $(2,1)$, then these policies are optimal for some $x > -\infty$ (namely $x \ge 0$), but not for all $x > -\infty$, because they stop being optimal for $x < 0$. Note that, in this example, we vary the payoff of an edge controlled by Min, we do not modify the policy of Min, this policy does not use the edge whose payoff we modify, but an optimal policy of Max may still change depending on $x$.
\end{example}

\subsection{Conclusion}

We do not know which of the properties \labelcref{enum:first}--\labelcref{enum:third} are satisfied for random games. Indeed, in this paper we no not prove \emph{any} of them. This is in contrast to the retracted work~\cite{BorosElbassioniFouzGurvich:2011}, which claimed that the properties \labelcref{enum:second} and \labelcref{enum:third} are satisfied for random games. Instead, we restrict our attention to a special class of optimal policies, namely the \emph{bias-induced} policies. The bias-induced policies are very suitable to the techniques that we used in the proof of \cref{robust_cycle_one_player_second} because they admit an algebraic characterization using the so-called ergodic equation. In particular, this equation allows us to restore the ``stability'' property discussed above. A further advantage of bias-induced policies is that under genericity conditions, the pair of bias-induced policies is unique and coincides with the unique pair of \emph{Blackwell-optimal} policies. Moreover, the bias-induced policies split the weight space $\R^m$ into a union of disjoint open polyhedral cones, just like the optimal cycles split the weight space in the case of DMDPs. This allows us to prove a combination of properties \labelcref{enum:second} and \labelcref{enum:third} restricted to bias-induced policies. In other words, we prove that the following holds with high probability:

\begin{enumerate}
\item[(P-4)]\label{enum:fourth} The bias-induced policies are unique and identical for all $\tilde{\payoff} \in B_{\infty}(\payoff,\delta)$.
\end{enumerate}


\section{Blackwell-optimal and bias-induced policies}\label{sec:ergodic}

\subsection{Discounted games}
A \emph{(deterministic) discounted game} is defined in a very similar way to a mean-payoff game. The only difference is that the payoff received by Max is defined as\footnote{The payoff is normalized by a $1 - \disc$ factor so that for any $\gamma \in \interval[open]{0}{1}$, it is always bounded by the maximum weight, and we can meaningfully consider the limit as $\disc$ goes to $1$.}
\[
(1 - \disc)(\payoff_{i_0i_1} + \disc\payoff_{i_1i_2} +  \disc^2\payoff_{i_2i_3} + \dots) \, ,
\]
where $(i_0,i_1,\dots)$ is the path of the token and $\disc \in \interval[open]{0}{1}$ is the \emph{discount factor}.

\medskip\noindent
Shapley proved that players in discounted games have optimal policies.

\begin{theorem}[\cite{shapley_stochastic}]\label{th:shapley}
Both players in a discounted game have optimal policies. More precisely, the value $\gameval^{(\disc)} \in \R^n$ of the discounted game with discount factor $\disc$ is the unique solution of the equations
\begin{align*}
\gameval^{(\disc)}_i \coloneqq \begin{cases}
\max_{(i,j) \in \edges}\{(1 - \disc)\payoff_{ij} + \disc \gameval^{(\disc)}_j\} &\text{if $i \in \Maxvertices$},\\[0.1cm]
\min_{(i,j) \in \edges}\{(1 - \disc)\payoff_{ij} + \disc \gameval^{(\disc)}_j\} &\text{if $i \in \Minvertices$}.
\end{cases}
\end{align*}
A policy of Max (resp. Min) is optimal if and only if it uses edges that achieve the maxima (resp. minima) in these equalities.
\end{theorem}

\begin{definition}
We say that a policy of a given player is \emph{Blackwell optimal} if it is optimal for all discount factors $\disc$ that are sufficiently close to $1$.
\end{definition}

Liggett and Lippman showed that Blackwell optimal policies always exist and are optimal in the mean payoff game.

\begin{theorem}[\cite{liggett_lippman}]\label{th:liggett_lippman}
  Both players have Blackwell-optimal policies. More precisely, there exists $\disc^* \in \interval[open]{0}{1}$ such that if a policy is optimal for some $\disc^* < \disc < 1$, then it is optimal for all $\disc^* < \disc < 1$. Moreover, any Blackwell-optimal policy is optimal in the mean payoff game on the same weighted graph.
    The smallest number $\disc^*$ with this property is called the \emph{Blackwell threshold.}
\end{theorem}

A more precise result about the asymptotic behavior of discounted games as $\disc$ goes to $1$ follows by applying a theorem of Kohlberg~\cite{kohlberg} to the Shapley operator $\shapley \colon \R^{n} \to \R^{n}$ defined as
\begin{align*}
\shapley(x)_i \coloneqq 
\begin{cases}
\max_{(i,j) \in \edges}\{ \payoff_{ij} + x_j \} &\text{if $i \in \Maxvertices$}, \\
\min_{(i,j) \in \edges}\{ \payoff_{ij} + x_j \} &\text{if $i \in \Minvertices$}.
\end{cases}
\end{align*}

\begin{theorem}[Corollary of \cite{kohlberg}]\label{th:kohlberg}
For $\disc$ sufficiently close to $1$, the value of the discounted game can be developed into a convergent power series. More precisely, there exist $c_{0}, c_1, c_{2}, \ldots \in \R^n$ and $\disc_0 \in \interval[open]{0}{1}$ such that for all $\disc_0 < \disc < 1$ the value of the discounted game is equal to
\[
\gameval^{(\disc)} = c_{0} + (1 - \disc)c_1 + (1 - \disc)^2c_2 + (1 - \disc)^3c_3 + \dots \, .
\]
Furthermore, $c_{0}$ is the value of the mean payoff game, and for all sufficiently large $t > 0$ we have $\shapley(t c_{0} + c_1) = (t+1)c_{0} + c_1$. We call the vector $c_1$ the \emph{Blackwell bias} of the mean-payoff game.
\end{theorem}

 We refer the reader to \cite{polyhedra_equiv_mean_payoff,AllamigeonGaubertKatzSkomra:2022}, or \cite[Chapter~6]{skomra_phd} for more information about Kohlberg's theorem and its applications in mean payoff games.

\subsection{Games with constant value}
We now introduce the basic framework of the ``operator approach'' to mean-payoff games. In this work, we focus our attention on mean-payoff games that have the same value for all initial states of the game. Under this assumption, the operator approach studies the following equation.

\begin{definition}\label{def:ergodic}
The \emph{ergodic equation} of a mean-payoff game is the following equation in variables $(\gameval, \bias) \in \R^{n+1}$:
\begin{equation}\label{eq:ergodic}
\begin{cases}
\forall i \in \Maxvertices, \ \gameval + \bias_i = \max_{(i,j) \in \edges}\{\payoff_{ij} + \bias_j\} \, , \\
\forall i \in \Minvertices, \ \gameval + \bias_i = \min_{(i,j) \in \edges}\{\payoff_{ij} + \bias_j\} \, .
\end{cases}
\end{equation}
\end{definition}

For games with constant value, the following theorem can be obtained as a consequence of \cref{th:kohlberg}, see also \cite{gurvich,boros2013canonical} for a different approach.

\begin{theorem}[{e.g., \cite[Section~2.2]{generic_uniqueness}}]\label{ergodic_solvability}
The ergodic equation has a solution $(\gameval,\bias)$ if and only if the value of the MPG does not depend on the initial state. Furthermore, if $(\gameval,\bias)$ is a solution of the ergodic equation, then $\gameval$ is the value of the game. In particular, $\gameval$ is unique.
\end{theorem}

\begin{proof}
Suppose that the mean-payoff game has the same value for all initial states and denote it by $\gameval$. Let $\bias$ be the Blackwell bias of the game. Then, \cref{th:kohlberg} shows that $\shapley(t\gameval(1,\dots,1) + \bias) = (t + 1)\gameval(1,\dots,1) + \bias$. On the other hand, by definition of $\shapley$ we have $\shapley(t\gameval(1,\dots,1) + \bias) = t\gameval(1,\dots,1) + \shapley(\bias)$, so $\shapley(\bias) = \gameval(1,\dots,1) + \bias$ and $(\gameval,\bias)$ solves the ergodic equation. Conversely, suppose that the ergodic equation has a solution $(\gameval,\bias)$. Let $\tau$ be a policy of Min that at each $i \in \Minvertices$ chooses an edge that achieves the minimum in $\gameval + \bias_i = \min_{(i,j) \in \edges}\{\payoff_{ij} + \bias_j\}$. Let $\sigma$ be any policy of Max, and $\cycle$ be any cycle in the graph obtained from using policies $\sigma,\tau$. Then, we have $\gameval + \bias_i \ge \payoff_{ij} + \bias_j$ for every edge $(i,j)$ of $\cycle$. By summing these inequalities, we get that the mean weight of $\cycle$ is not greater than $\gameval$. Hence, $\tau$ guarantees that the payoff of the mean-payoff game is at most $\gameval$. Analogously, we can construct a policy of Max that guarantees that the payoff of the mean-payoff game is at least $\gameval$. Hence, $\gameval$ is the value of the game and this value does not depend on the initial state.
\end{proof}

We refer the reader to~\cite{hochartdominion} for more information about the solvability of the ergodic equation for more general classes of games and to \cite{AllamigeonGaubertKatzSkomra:2022,AkianGaubertNaepelsTerver:2023} for recent algorithmic applications of the operator approach to mean-payoff games and entropy games.

\begin{definition}\label{def:bias}
We say that a vector $\bias$ is a \emph{bias} of the mean-payoff game if $(\gameval,\bias)$ is a solution of the ergodic equation.
\end{definition}

\Cref{ergodic_solvability} shows that a bias exists, because the Blackwell bias is a bias. We note that bias is never unique, because the vector $\alpha(1,\dots,1) + \bias$ is also a bias. In general, a bias is not unique even up to adding a constant (we give an example in \cref{subsec:non-uniqueness}). However, the set of biases is connected. 

\begin{proposition}[{\cite[Theorem~3.10]{generic_uniqueness}}]\label{biases_connected}
The set of biases of a mean payoff game is path connected.
\end{proposition}

We refer to~\cite{ambitropical_convexity} for more information about the geometry of the set of biases of a mean payoff game. The proof of \cref{ergodic_solvability} also suggests the concept of bias-induced policies. 

\begin{definition}
We say that a policy $\tau$ of Min is \emph{bias induced} if there exists a bias $\bias$ such that $\payoff_{i\tau(i)} + \bias_{\tau(i)} = \min_{(i,j) \in \edges}\{\payoff_{ij} + \bias_j\}$ for all $i \in \Minvertices$. We analogously define a bias-induced policy of Max. 
\end{definition}

We note that since biases are not unique, the bias-induced policies are also not unique in general. \Cref{ergodic_solvability} gives the following corollary.

\begin{corollary}\label{bias_induced_optimality}
All bias-induced policies are optimal. Furthermore, any Blackwell-optimal policies are bias induced.
\end{corollary}
\begin{proof}
  The first part follows directly from the proof of \cref{ergodic_solvability}. To prove the second part, let $\tau$ be a Blackwell-optimal policy of Min, $\bias$ be the Blackwell bias of the game, $\gameval$ be the value of the mean-payoff game, and $\gameval^{(\disc)}$ be the value of the discounted game for discount factor $\disc$. We then find that, for any $(i,j) \in \edges$ such that $i \in \Minvertices$ and any $\disc$ close to $1$,
\begin{align*}
  \gameval^{(\disc)}_i %
  & = (1-\disc)\payoff_{i \tau(i)} + \disc \gameval_{\tau(i)}^{(\disc)}\tag*{(\cref{th:shapley})}\\
  & = (1 - \disc)\payoff_{i\tau(i)} + \disc \gameval + \disc(1 - \disc)\bias_{\tau(i)} + \disc(1-\disc)^2 c_{2\tau(i)} + \dots  \tag*{(\cref{th:kohlberg})}  \, .
\end{align*}
But also,
\begin{align*}
  \gameval^{(\disc)}_i %
  & \le (1-\disc)\payoff_{i j} + \disc \gameval_{j}^{(\disc)}\tag*{(\cref{th:shapley})}\\
  & = (1 - \disc)\payoff_{ij} + \disc \gameval + \disc(1 - \disc)\bias_j + \disc(1-\disc)^2c_{2j} + \dots \, .  \tag*{(\cref{th:kohlberg})}
\end{align*}
It follows that
\begin{align*}
  (1 - \disc)\payoff_{i\tau(i)} + \disc \gameval + & \disc(1 - \disc)\bias_{\tau(i)} + \disc(1-\disc)^2 c_{2\tau(i)} + \dots\\
  & \le (1 - \disc)\payoff_{ij} + \disc \gameval + \disc(1 - \disc)\bias_j + \disc(1-\disc)^2c_{2j} + \dots \, . 
\end{align*}
By dividing this inequality by $(1 - \disc)$ and going to the limit as $\gamma$ goes to $1$, we get $\payoff_{i\tau(i)} + \bias_{\tau(i)} \le \payoff_{ij} + \bias_j$. Hence $\payoff_{i\tau(i)} + \bias_{\tau(i)} = \min_{(i,j) \in \edges}\{\payoff_{ij} + \bias_j\}$ and $\tau$ is induced by the Blackwell bias. The proof for player Max is analogous.
\end{proof}

We refer the reader to \cite{ambitropical_convexity} for more information about bias-induced policies and, in particular, to \cite[Proposition~7.4]{ambitropical_convexity} for a proof of the fact that bias-induced policies are ``$u$-calibrated''.

\subsection{Examples of uniqueness and non-uniqueness}\label{subsec:non-uniqueness} We here show several examples illustrating that biases and bias-induced policies are not unique in general, but are sometimes. Similar examples can be found in \cite{ambitropical_convexity}. The Blackwell bias is unique, even though the Blackwell-optimal policy is not unique in general (think of a game with multiple zero edges). Nonetheless, \cref{ex:non-unique-bias-induced-policies} shows simultaneously that (1) the Blackwell bias, which induces the Blackwell-optimal policy, can also induce other, non-Blackwell-optimal policies; this can happen even if the Blackwell-optimal policy is itself unique; (2) infinitesimal changes in a bias can keep it a bias but change the set of policies induced by it; also, a game can have two biases that are not the same up to adding a constant to every coordinate; (3) a game can have two different biases, where each of these two biases induces a single policy, but these two bias-induced policies are different; and (4) a game can have optimal policies that are not bias-induced. In contrast, \cref{ex:emerging_policies_again,ex:unstable_policies_good_approx_again} revisit \cref{ex:emerging_policies,ex:unstable_policies_good_approx} to show that bias-induced policies are sometimes unique, and stable under small perturbations.

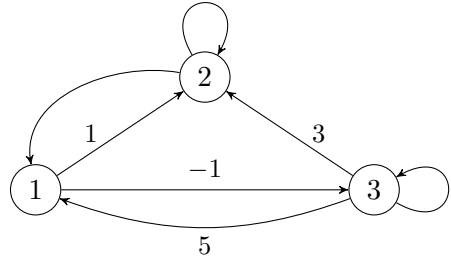
\begin{figure}[t]
\centering
\begin{tikzpicture}[scale=0.75,>=stealth',max/.style={draw,rectangle,minimum size=0.5cm},min/.style={draw,circle,minimum size=0.5cm},av/.style={draw, circle,fill, inner sep = 0pt,minimum size = 0.2cm}]

\node[min] (min1) at (-3, 0) {$1$};
\node[min] (min3) at (3, 0) {$3$};

\node[min] (min2) at (0,2) {$2$};

\draw[->] (min1) to node[above, font=\small]{$-1$} (min3);
\draw[->] (min3) to node[right=1ex, font=\small]{$3$} (min2);
\draw[->] (min1) to node[left=1ex, font=\small]{$1$} (min2);
\draw[->] (min3) to[out=200,in=-20] node[below, font=\small]{$5$} (min1);
\draw[->] (min2) to[out=120,in=60,looseness=8] (min2);
\draw[->] (min3) to[out=-30,in=30,looseness=8] (min3);
\draw[->] (min2) to[out=170,in=100] (min1);

\end{tikzpicture}
\caption{An example of a one-player game with a single Blackwell-optimal policy but multiple bias-induced policies.} \label{fig:blackwell_optimal}
\end{figure}

\begin{example}\label{ex:non-unique-bias-induced-policies}
Consider the one-player game presented in \cref{fig:blackwell_optimal}. The minimum mean weight of any cycle in this game is $\gameval = 0$ and there are two cycles that achieve this value, namely $(2,2)$ and $(3,3)$. The only Blackwell-optimal policy of Min is the one that uses the edges $(1,3),(2,1),(3,3)$. To verify this, it is enough to note that the equations from \cref{th:shapley} are given by
\begin{align*}
\gameval^{(\disc)}_1 &= \min\{(1 - \disc) + \disc\gameval^{(\disc)}_2,  -(1 - \disc) + \disc\gameval^{(\disc)}_3 \},\\
\gameval^{(\disc)}_2 &= \min\{\disc\gameval^{(\disc)}_1,  \disc\gameval^{(\disc)}_2 \},\\
\gameval^{(\disc)}_3 &= \min\{5(1 - \disc) + \disc\gameval^{(\disc)}_1,  3(1 - \disc) + \disc\gameval^{(\disc)}_2, \disc\gameval^{(\disc)}_3\}
\end{align*}
and the (unique) solution is given by
\[
\gameval^{(\disc)}_1 = -(1-\disc), \quad \gameval^{(\disc)}_2 = -(1-\disc)\disc, \quad \gameval^{(\disc)}_3 = 0
\]
for any $\disc > 0$. So the only policy induced by this bias for all $\gamma$ close to $1$ is $(1,3),(2,1),(3,3)$. Moreover, the Blackwell bias of this game is $\bias^* = \lim_{\disc \to 1^-}\frac{\gameval^{(\disc)} - \gameval}{1 - \disc} = (-1,-1,0)$ (by \cref{th:kohlberg}) and the ergodic equation is given by
\begin{align*}
\bias_1 &= \min\{1 + \bias_2,  -1 + \bias_3 \},\\
\bias_2 &= \min\{\bias_1,  \bias_2 \},\\
\bias_3 &= \min\{5 + \bias_1,  3 + \bias_2, \bias_3\}.
\end{align*}
\begin{enumerate}
\item Note that the Blackwell bias induces the unique Blackwell optimal policy $(1,3),(2,1),(3,3)$, but it also induces the policy given by the edges $(1,3),(2,2),(3,3)$, which is not Blackwell optimal.
\item Furthermore, any vector of the form $(-1,y,0)$ for $y \in \interval{-2}{-1}$ is a bias. This bias for $y = -2$ induces the policy $(1,2),(2,2),(3,3)$. Likewise, the vector $(-1,-2,z)$ is a bias for $z \in \interval{0}{1}$. For $z = 1$ it induces the policy $(1,2),(2,2),(3,2)$. 
\item The bias $(-1,-\frac{3}{2},0)$ uniquely induces the policy $(1,3),(2,2),(3,3)$, and the bias $(-1,-2,\frac{1}{2})$ uniquely induces the different policy $(1,2),(2,2),(3,3)$.
\item Finally, note that the policy that uses the edges $(1,2),(2,2),(3,1)$ is optimal but it is not induced by any bias. Indeed, a bias inducing this policy would satisfy $5 + \bias_1 \le 3 + \bias_2$ and $\bias_1 = 1 + \bias_2$ which is impossible. In the same way, the policy that uses the edges $(1,3),(2,2),(3,2)$ is optimal but it is not induced by any bias, because this bias would satisfy $-1 + \bias_3 \le 1 + \bias_2$ and $\bias_3 = 3 + \bias_2$, which is impossible.
\end{enumerate}
\end{example}

The crucial feature of the example above is that the optimal cycle is not unique. We now revisit the examples from \cref{sec:det_mpg} in which the equilibrium cycle is unique.

\begin{example}\label{ex:emerging_policies_again}
Consider once again the example of a game from \cref{ex:emerging_policies} and $\varepsilon = 0$. We claim that the policies using the edges $(1,2),(2,2)$ are the unique bias-induced policies even under a ``large'' perturbation of the weights. Indeed, if $\tilde{\payoff} \in B_{\infty}(\payoff,1)$, then the value of the game is $\gameval = \tilde{\payoff}_{22}$ and the only optimal policy of Max is to use $(2,2)$. Let $\bias$ be any bias. By adding a constant to this bias, we may suppose that $\bias_2 = 0$. Then, the ergodic equation reads
\begin{align*}
\tilde{\payoff}_{22} + \bias_1 &= \min\{\tilde{\payoff}_{12}, \tilde{\payoff}_{13} + \bias_3 \},\\
\tilde{\payoff}_{22} &= \min\{\tilde{\payoff}_{21} + \bias_1, \tilde{\payoff}_{22} \},\\
\tilde{\payoff}_{22} + \bias_3 &= \tilde{\payoff}_{32}.
\end{align*}
Therefore, in any bias-induced policy, at node $1$ player Min chooses an edge that achieves the minimum in $\min\{\tilde{\payoff}_{12}, \tilde{\payoff}_{13} + \tilde{\payoff}_{32} - \tilde{\payoff}_{22} \}$. Since $\tilde{\payoff}_{12} \le -9$ and $\tilde{\payoff}_{13} + \tilde{\payoff}_{32} - \tilde{\payoff}_{22} \ge -3$, we get that Min chooses $(1,2)$ in any bias-induced policy. Hence, contrary to the discussion in \cref{ex:emerging_policies}, the bias-induced policies are not only unique for $\varepsilon = 0$ but they are also stable under perturbation.
\end{example}

\begin{example}\label{ex:unstable_policies_good_approx_again}
Consider once again the example of a game from \cref{ex:unstable_policies_good_approx} and let $\tilde{\payoff} \in B_{\infty}(\payoff,1)$ be a perturbation of the weights. The only optimal policy of Max is to use the edge $(2,4)$. Let $\bias$ be any bias of this game and suppose, by adding a constant, that $\bias_4 = 0$. Then, the ergodic equation reads
\begin{align*}
\tilde{\payoff}_{44} + \bias_1 &= \tilde{\payoff}_{12} + \bias_2,\\
\tilde{\payoff}_{44} + \bias_2 &= \max\{\tilde{\payoff}_{23} + \bias_3,\tilde{\payoff}_{24}\},\\
\tilde{\payoff}_{44} + \bias_3 &= \min\{\tilde{\payoff}_{31} + \bias_1, \tilde{\payoff}_{32}+\bias_2,\tilde{\payoff}_{34}\}.
\end{align*}
Since $\tilde{\payoff}_{31} + \tilde{\payoff}_{12} - \tilde{\payoff}_{44} \ge -3 > -9 \ge \tilde{\payoff}_{32}$, player Min cannot choose the edge $(3,1)$ in any bias-induced policy. Moreover, we have $\bias_2 \ge \bias_3 + \tilde{\payoff}_{44} - \tilde{\payoff}_{32} \ge \bias_3 + 8$ from the third equation. Since $\tilde{\payoff}_{23} - \tilde{\payoff}_{44} + \bias_3 \le 2 + \bias_3$, we get that $\bias_2 = \tilde{\payoff}_{24} - \tilde{\payoff}_{44}$. Since $\tilde{\payoff}_{32} + \tilde{\payoff}_{24} - \tilde{\payoff}_{44} \le -7 < -1 \le \tilde{\payoff}_{34}$, we get that the unique bias-induced policy is the one in which Min uses the edge $(3,2)$. As above, this implies that in this example the bias-induced policies are unique and stable under perturbation.
\end{example}

\subsection{Games on ergodic graphs}

In \cref{sec:one_player} we performed an analysis of random DMDPs. To simplify the situation, we considered only DMDPs on graphs with a single strongly-connected component. For the sake of simplicity, our analysis of MPGs will also restrict the class of graphs under consideration. To do so, we need a condition similar to strong connectivity for MPGs. Since our approach is based on the ergodic equation, we use the following definition.

\begin{definition}
We say that a digraph $\dgraph$ is \emph{ergodic} if the ergodic equation has a solution for every choice of weights $\payoff$. Equivalently, by \cref{ergodic_solvability}, a digraph is ergodic if and only if every MPG on this graph has a value that does not depend on the initial state.
\end{definition}

This notion of ergodicity was studied in \cite{gurvich_lebedev} and extended to wider classes of games in \cite{hochartdominion,boros_gurvich_makino}.\footnote{We also note that the notion of an ``ergodic'' game is defined in many different ways by different authors, depending on the application.} In particular, these papers prove that the ergodic graphs are fully characterized using the notion of dominions. We give this characterization for the sake of completeness.
 
\begin{definition}
A \emph{dominion of Max} is a set of states $\dominion \subseteq [n]$ such that if the game starts in $\dominion$, then Max can force it to stay in $\dominion$. More precisely, $\dominion$ is a dominion if every vertex $\vertex \in \dominion \cap \Maxvertices$ has at least one outgoing edge that goes to $\dominion$ and all outgoing edges of all vertices $\vertex \in \dominion \cap \Minvertices$ go to $\dominion$. We analogously define the notion of a dominion of Min.
\end{definition}

We note that the notion of the dominion depends only on the underlying digraph and not on the particular choice of the weights.

\begin{theorem}[{\cite{gurvich_lebedev}, see also \cite{hochartdominion,boros_gurvich_makino}}]\label{ergodic_graph}
A digraph $\dgraph$ is ergodic if and only if there is no disjoint nonempty sets $\dominion_1, \dominion_2 \subseteq [n]$ such that $\dominion_1$ is a dominion of Max and $\dominion_2$ is a dominion of Min.
\end{theorem}

A standard example of an ergodic graph is a complete bipartite graph in which $\Maxvertices, \Minvertices$ form the bipartition. We note that ergodic graphs are very representative for MPGs, because solving MPGs on complete bipartite graphs is as hard as solving MPGs on general graphs.

\begin{theorem}[{\cite[Corollary~19]{ChatterjeeHenzingerKrinningerNanongkai:2014}}]
The problem of solving MPGs is polynomial-time (Turing) reducible to the problem of solving MPGs on complete bipartite graphs.
\end{theorem}

\subsection{Generic uniqueness of bias-induced policies} From now on, we only consider games played on ergodic graphs. In this case, \cite[Theorem~3.2]{generic_uniqueness} shows that the bias becomes unique up to a constant if the weights are generic. In our analysis, we need to establish slightly more precise results than the ones stated in \cite{generic_uniqueness}, so we give complete proofs. Nevertheless, our proofs are based on the same techniques as the ones used in \cite{generic_uniqueness}.

We start by analyzing the biases in very simple graphs that contain only one cycle. The following fundamental lemma characterizes the bias of ``zero-player'' games. (Generalizations appear, e.g., in \cite[Corollary 3.8]{generic_uniqueness} or \cite[Section 8.2.3]{puterman}).

%
%

\begin{lemma}\label{one_cycle_graph}
Suppose that every vertex of $\dgraph$ has exactly one outgoing edge and that $\dgraph$ has only one directed cycle. Then, the bias is unique up to adding a constant to every coordinate. More precisely, if $(\gameval,\bias)$ is a solution of the ergodic equation and $k \in [n]$ is a vertex on the directed cycle, then for every $l \in [n]$ the difference $\bias_l - \bias_k$ is the distance from $l$ to $k$ in the digraph $\dgraph$ with weights $\payoff_{ij} - \gameval$.
\end{lemma}

\begin{proof}
To start, note that the digraph $\dgraph$ is ergodic, because any dominion of any player contains the cycle of the graph. In particular, by \cref{ergodic_graph}, the ergodic equation on $\dgraph$ always has a solution. Let $k$ be any vertex in the cycle. Let $(k,k')$ be the edge of $\dgraph$ that goes out of $k$ and consider the digraph $\dgraph'$ obtained by removing this edge. Note that the digraph $\dgraph'$ is a tree rooted at $k$. Let $\bias$ be any bias. By adding a constant to every coordinate of $\bias$, we can suppose that $\bias_k = 0$. By definition, for every edge $(i,j)$ of the graph we have $\bias_i = \payoff_{ij} - \lambda + \bias_j$. Hence, we can find all coordinates of $\bias$ by a backward induction starting at $k$, and $\bias_l$ is precisely the distance from $l$ to $k$ in the graph with weights $\payoff_{ij} - \lambda$. This implies that the bias is unique up to adding a constant.
\end{proof}

\begin{definition}[$\Xi$, $\gameval^{\sigma,\tau}$, $\bias^{\sigma,\tau}$]
  We denote by $\Xi$ the set of pairs $(\sigma,\tau)$ of policies such that the induced subgraph $\dgraph^{\sigma,\tau}$ has only one cycle. Given a pair $(\sigma, \tau) \in \Xi$, we denote by $\gameval^{\sigma,\tau} \colon \R^{m} \to \R$ the function that takes as an argument the vector of weights and outputs the value of the MPG played on $\dgraph^{\sigma,\tau}$. Analogously, we denote by $\bias^{\sigma,\tau} \colon \R^{m} \to \R^{n}$ the function that takes as an argument the weight vector and outputs the bias of the MPG played on $\dgraph^{\sigma,\tau}$ that satisfies $(\bias^{\sigma,\tau}(\payoff))_k = 0$, where $k$ is the vertex with the minimal index on the cycle of $\dgraph^{\sigma,\tau}$.
\end{definition}

\begin{lemma}\label{linear_functions}
For every pair $(\sigma,\tau) \in \Xi$, the functions $\gameval^{\sigma,\tau}, \bias^{\sigma,\tau}$ are linear. Furthermore,  for all $\payoff, \tilde{\payoff} \in \R^{m}$ we have $\supnorm{\gameval^{\sigma,\tau}(\payoff) - \gameval^{\sigma,\tau}(\tilde{\payoff})} \le \supnorm{\payoff - \tilde{\payoff}}$ and $\supnorm{\bias^{\sigma,\tau}(\payoff) - \bias^{\sigma,\tau}(\tilde{\payoff})} \le 2n\supnorm{\payoff - \tilde{\payoff}}$.
\end{lemma}

\begin{proof}
The function $\gameval^{\sigma,\tau}$ is the mean value of the cycle of $\dgraph^{\sigma,\tau}$, so it is linear and satisfies $\supnorm{\gameval^{\sigma,\tau}(\payoff) - \gameval^{\sigma,\tau}(\tilde{\payoff})} \le \supnorm{\payoff - \tilde{\payoff}}$. Furthermore, by \cref{one_cycle_graph}, $(\bias^{\sigma,\tau}(\payoff))_l$ is the length of some path in the graph with weights $\payoff_{ij} - \lambda^{\sigma,\tau}(\payoff)$, so it is also linear and satisfies $\supnorm{\bias^{\sigma,\tau}(\payoff) - \bias^{\sigma,\tau}(\tilde{\payoff})} \le 2n\supnorm{\payoff - \tilde{\payoff}}$.
\end{proof}

\begin{lemma}\label{solution_from_subgraph}
Suppose that $(\gameval, \bias)$ solve the ergodic equation and that the bias $\bias$ induces a pair of policies $(\sigma,\tau) \in \Xi$. Then, $\gameval = \gameval^{\sigma,\tau}(\payoff)$ and $\bias = \bias^{\sigma,\tau}(\payoff) + \alpha(1,\dots,1)$ for some $\alpha \in \R$.
\end{lemma}
\begin{proof}
The pair $(\gameval,\bias)$ solves the ergodic equation on $\dgraph^{\sigma,\tau}$, so the claim follows from \cref{one_cycle_graph}.
\end{proof}


\medskip\noindent
The next lemma uses the same perturbation argument as \cite[Proposition 4.5]{generic_uniqueness}.

\begin{lemma}\label{biased_induced_one_cycle}
For every setting of weights there exists a pair $(\sigma, \tau)$ of bias-induced policies such that $(\sigma,\tau) \in \Xi$.
\end{lemma}
\begin{proof}
To start, note that by perturbing the weights slightly we can ensure that all the cycles of the digraph $\dgraph$ have different mean weights. To make this precise, let $w_{ij} \in \N$ be a collection of distinct natural numbers, one for every edge $(i,j)$ of $\dgraph$. Consider the perturbed weights $\payoff_{ij}(\varepsilon) \coloneqq \payoff_{ij} + \varepsilon^{w_{ij}}$ for $\varepsilon > 0$. Consider the digraph $\dgraph$ with weights $\payoff_{ij}(\varepsilon)$. Given two cycles of $\dgraph$, the set of $\varepsilon$ such that these cycles have the same mean is a solution to a nontrivial polynomial equation. In particular, this set is finite. Hence, there exists $\varepsilon_0 > 0$ such that all cycles of $\dgraph$ have different mean weights for all $0 < \varepsilon < \varepsilon_0$. For each $0 < \varepsilon < \varepsilon_0$, let $(\sigma^{(\varepsilon)},\tau^{(\varepsilon)})$ be any pair of bias-induced policies in the game with weights $\payoff_{ij}(\varepsilon)$. Note that $(\sigma^{(\varepsilon)},\tau^{(\varepsilon)}) \in \Xi$. Indeed, since $(\sigma^{(\varepsilon)},\tau^{(\varepsilon)})$ are optimal and $\dgraph$ is ergodic, every cycle of the graph $\dgraph^{\sigma^{(\varepsilon)},\tau^{(\varepsilon)}}$ has the same mean weight. Since all cycles have different mean weights, we get that the cycle is unique. Furthermore, since the set of policies is finite, we can find a sequence $\varepsilon_{\ell} \to 0^{+}$ such that $(\sigma^{(\varepsilon_\ell)},\tau^{(\varepsilon_\ell)})$ are identical for all $\ell$. Let $(\sigma,\tau) \coloneqq (\sigma^{(\varepsilon_1)},\tau^{(\varepsilon_1)}) \in \Xi$. We want to show that $(\sigma,\tau)$ are bias induced in the game with weights $\payoff$. To do so, let $k$ be the smallest index of a vertex in the cycle of $\dgraph^{\sigma,\tau}$ and let $\gameval^{(\ell)}, \bias^{(\ell)}$ be the value of the game with weights $\payoff_{ij}(\varepsilon_\ell)$ and the bias of this game that induces $(\sigma^{(\varepsilon_\ell)},\tau^{(\varepsilon_\ell)})$ and satisfies $\bias^{(\ell)}_k = 0$. By \cref{solution_from_subgraph}, we have $\gameval^{(\ell)} = \gameval^{\sigma,\tau}\bigl(\payoff(\varepsilon_{\ell})\bigr)$ and $\bias^{(\ell)} = \bias^{\sigma,\tau}\bigl(\payoff(\varepsilon_{\ell})\bigr)$ for all $\ell$. Therefore, the continuity of $\gameval^{\sigma,\tau},\bias^{\sigma,\tau}$ proven in \cref{linear_functions} implies that the pair $\bigl(\gameval^{\sigma,\tau}(\payoff),\bias^{\sigma,\tau}(\payoff)\bigr)$ solves the ergodic equation of the game with weights $\payoff$ and that $(\sigma,\tau)$ are induced by $\bias^{\sigma,\tau}(\payoff)$.
\end{proof}

By the above Lemma, if a pair of policies $(\sigma, \tau)$ is the unique pair of bias-induced policies, then we must have $(\sigma,\tau)\in\Xi$. This justifies the following.

\begin{definition}[$\polyh^{\sigma,\tau}$, $\Generic$]\label{def:polyhedral_cells}
  Given a pair of policies $(\sigma,\tau) \in \Xi$ we consider the set
  \begin{align*}
    \polyh^{\sigma,\tau} \coloneqq \{\payoff \in \R^m \colon &\text{$(\sigma,\tau)$ is the only pair of bias-induced policies}\\
                                                             &\text{in the MPG with weights $\payoff$}\}.
  \end{align*}
  We also let $\Generic \coloneqq \bigcup_{(\sigma,\tau)\in\Xi} \polyh^{\sigma,\tau}$.
\end{definition}

We now state the main theorem of this section, which extends \cite[Theorem~3.2]{generic_uniqueness} for deterministic games by showing that each maximum and minimum in the ergodic equation is achieved by a single edge for a generic choice of weights. Furthermore, \Cref{eq:polyhedral_cone} in the proof gives an explicit description of $\polyh^{\sigma,\tau}$ using affine inequalities.

\begin{theorem}[{cf. \cite[Theorem~3.2]{generic_uniqueness}}]\label{polyhedral_complex}
The sets $\polyh^{\sigma,\tau}$ are open polyhedral cones. Moreover, these cones are disjoint and $\R^m \setminus \Generic$ is included in a finite union of hyperplanes. In particular, this set has Lebesgue measure zero. Furthermore, if $\payoff \in \Generic$, then the ergodic equation has a single solution (up to adding a constant to the bias), and each maximum and minimum in the ergodic equation is achieved by a single edge.
\end{theorem}

Before giving a proof, let us note that this proposition implies that bias-induced policies are generically unique. In particular, by \cref{bias_induced_optimality}, this means that Blackwell-optimal policies are generically unique and they coincide with the unique pair of bias-induced policies.

\begin{proof}
To show the first part, it is enough to prove that $\polyh^{\sigma,\tau}$ is the polyhedral cone given by
\begin{equation}\label{eq:polyhedral_cone}
\begin{aligned}
\{\payoff \in \R^m \colon &\forall (i,j) \in \edges, i \in \Maxvertices, j \neq \sigma(i), \lambda^{\sigma,\tau}(\payoff) + \bias_i^{\sigma,\tau}(\payoff) - \bias_j^{\sigma,\tau}(\payoff) > \payoff_{ij} \, , \\
&\forall (i,j) \in \edges, i \in \Minvertices, j \neq \tau(i), \lambda^{\sigma,\tau}(\payoff) + \bias_i^{\sigma,\tau}(\payoff) - \bias_j^{\sigma,\tau}(\payoff) < \payoff_{ij} 
\} \, .
\end{aligned}
\end{equation}
Indeed, since $\lambda^{\sigma,\tau}$ and $\bias^{\sigma,\tau}$ are linear by \cref{linear_functions}, the set above is an open polyhedral cone. Suppose that $\payoff \in \polyh^{\sigma,\tau}$ and let $\gameval$ be the value of the game with weights $\payoff$. Let $\bias$ be a bias that induces $(\sigma,\tau)$ and such that $\bias_k = 0$ where $k$ is the smallest index of a vertex in the cycle of $\dgraph^{\sigma,\tau}$. By \cref{solution_from_subgraph} we have $\gameval = \gameval^{\sigma,\tau}(\payoff)$ and $\bias = \bias^{\sigma,\tau}(\payoff)$. Since $(\sigma,\tau)$ are the only policies induced by $\bias$, we have $\gameval + \bias_i > \payoff_{ij} + \bias_j$ for $(i,j) \in \edges$ such that $i \in \Maxvertices$ and $j \neq \sigma(i)$ and analogously $\gameval + \bias_i < \payoff_{ij} + \bias_j$ for $(i,j) \in \edges$ such that $i \in \Minvertices$ and $j \neq \tau(i)$. Hence $\payoff$ satisfies the inequalities of the polyhedral cone given in \cref{eq:polyhedral_cone}.

Conversely, suppose that $\payoff$ is inside the polyhedral cone from \cref{eq:polyhedral_cone} and let $\gameval \coloneqq \gameval^{\sigma,\tau}(\payoff)$, $\bias \coloneqq \bias^{\sigma,\tau}(\payoff)$. By the definition of $\gameval^{\sigma,\tau},\bias^{\sigma,\tau}$, and of this polyhedral cone, the pair $(\lambda,\bias)$ solves the ergodic equation on $\dgraph$ and so $\gameval$ is the value and $\bias$ is a bias of the game on $\dgraph$. Moreover, owing to the strict inequalities, the bias $\bias$ induces $(\sigma,\tau)$ and these are the only policies induced by this bias. Therefore, to show that $(\sigma,\tau)$ is the only pair of bias-induced policies, it is enough to show that $\bias$ is the only bias up to adding a constant. To do so, let $\varepsilon > 0$ be such that
\begin{align*}
&\forall (i,j) \in \edges, i \in \Maxvertices, j \neq \sigma(i), \lambda + \bias_i - \bias_j > \payoff_{ij} + \varepsilon \, , \\
&\forall (i,j) \in \edges, i \in \Minvertices, j \neq \tau(i), \lambda + \bias_i - \bias_j < \payoff_{ij} - \varepsilon \, .
\end{align*}
Let $k$ be the smallest index of the cycle on $\dgraph^{\sigma,\tau}$ so that $\bias_k = 0$. Consider the set $B \subseteq \R^n$ defined as
\[
B \coloneqq \{z \in \R^{n} \colon \forall i, \abs{z_i - z_k - u_i} \le \frac{\varepsilon}{4} \} \, .
\]
Note that $B$ contains the line $\{\bias + \alpha(1,1,\dots,1) \colon \alpha \in \R\}$. Suppose that $v \in B$ is a bias of the game on $\dgraph$. Then, for every $i \in \Maxvertices, j \neq \sigma(i)$ we have $\lambda + v_i - v_j = \lambda + u_i - u_j + (v_i - v_k -u_i) - (v_j - v_k - u_j) > \payoff_{ij}$. Analogously, for every $i \in \Minvertices, j \neq \tau(i)$ we have $\lambda + v_i - v_j  < \payoff_{ij}$. Hence, $v$ is a bias of the game on $\dgraph^{\sigma,\tau}$ and \cref{one_cycle_graph} implies that $v$ is on the line $\{\bias + \alpha(1,1,\dots,1) \colon \alpha \in \R\}$. Therefore, the set of biases of the MPG with weights $\payoff$ is included in the union of two disjoint open sets
\[
(\R^{n} \setminus B) \cup \{z \in \R^{n} \colon \forall i, \abs{z_i - z_k - u_i} < \frac{\varepsilon}{4} \}  \, .
\] 
Since the set of biases is connected by \cref{biases_connected}, it can have a nonempty intersection with only one of these sets. Since $u$ itself is in $B$, all biases are therefore included in $B$ and so they all lie on the line $\{\bias + \alpha(1,1,\dots,1) \colon \alpha \in \R\}$. Therefore, we have shown that the bias is unique up to adding a constant and that $\payoff \in \polyh^{\sigma,\tau}$. This also shows the last part of the claim, that each maximum and minimum in the ergodic equation is achieved by a single edge.

Note that the sets $\polyh^{\sigma,\tau}$ are disjoint by definition, so it suffices to prove that $\R^m \setminus \Generic$ is included in a finite union of hyperplanes. Let $\payoff \in \R^m \setminus \Generic$ be any weight vector and $\lambda$ be the value of the resulting game. Then, \cref{biased_induced_one_cycle} implies that there exists a pair of policies $(\sigma,\tau) \in \Xi$ that are induced by some bias $\bias$. By \cref{solution_from_subgraph}, we have $\lambda = \lambda^{\sigma,\tau}(\payoff)$, $\bias = \bias^{\sigma,\tau}(\payoff) + \alpha(1,\dots,1)$ for some $\alpha \in \R$. In particular, $\payoff$ satisfies the inequalities
\begin{align*}
&\forall (i,j) \in \edges, i \in \Maxvertices, j \neq \sigma(i), \lambda^{\sigma,\tau}(\payoff) + \bias_i^{\sigma,\tau}(\payoff) - \bias_j^{\sigma,\tau}(\payoff) \ge \payoff_{ij} \, , \\
&\forall (i,j) \in \edges, i \in \Minvertices, j \neq \tau(i), \lambda^{\sigma,\tau}(\payoff) + \bias_i^{\sigma,\tau}(\payoff) - \bias_j^{\sigma,\tau}(\payoff) \le \payoff_{ij} \, .
\end{align*}
Since $\payoff \notin \polyh^{\sigma,\tau}$, at least one of these inequalities is satisfied as an equality. Note that each set of the form $\{ \payoff \in \R^{m} \colon \lambda^{\sigma,\tau}(\payoff) + \bias_i^{\sigma,\tau}(\payoff) - \bias_j^{\sigma,\tau}(\payoff) = \payoff_{ij}\}$ is a nontrivial hyperplane because the left-hand side of this equality does not depend of $\payoff_{ij}$ when $j \notin \{\sigma(i),\tau(i)\}$. Hence, the union of these hyperplanes satisfies the claim.
\end{proof}

\medskip\noindent
The inequalities given in \cref{eq:polyhedral_cone} mean that $\payoff \in \polyh^{\sigma,\tau}$ if and only if every maximum and minimum in the ergodic equation is achieved only once, by the edge used in $(\sigma,\tau)$. In particular, this property does not change if we take an edge that is not used in $(\sigma,\tau)$ and decrease or increase its weight, depending on which player controls this edge. This is formalized in the next observation, which gives us a two-player analogue of the ``stability'' property needed in the proof of \cref{robust_cycle_one_player_second}.

\begin{lemma}\label{two_player_stability}
Suppose that $\payoff \in \polyh^{\sigma,\tau}$ for some $(\sigma, \tau) \in \Xi$. If $(i,j) \in \edges$ is such that $i \in \Minvertices$, $j \neq \tau(i)$, then we have $(x,\payoff_{-ij}) \in \polyh^{\sigma,\tau}$ for all $x > \lambda^{\sigma,\tau}(\payoff) + \bias_i^{\sigma,\tau}(\payoff) - \bias_j^{\sigma,\tau}(\payoff)$. Analogously, if $(i,j) \in \edges$ is such that $i \in \Maxvertices$, $j \neq \sigma(i)$, then we have $(x,\payoff_{-ij}) \in \polyh^{\sigma,\tau}$ for all $x < \lambda^{\sigma,\tau}(\payoff) + \bias_i^{\sigma,\tau}(\payoff) - \bias_j^{\sigma,\tau}(\payoff)$.
\end{lemma}
\begin{proof}
Note that the functions $\gameval^{\sigma,\tau}, \bias^{\sigma,\tau}$ do not depend on the variable $\payoff_{ij}$. Therefore, the claim is an immediate consequence of the description of $\polyh^{\sigma,\tau}$ given in \cref{eq:polyhedral_cone}.
\end{proof}

%


\subsection{$\polyh^C$ versus $\polyh^{\sigma,\tau}$}
The following two examples illustrate the difference between the partition of $\R^m = \bigcup_C \polyh^C$ into polyhedral cells of \cref{def:polyhedral_cells_one_player} and the partition $\R^m = \bigcup_{\sigma,\tau} \polyh^{\sigma,\tau}$ of \cref{def:polyhedral_cells}. As it turns out, even in the special case (where the comparison makes sense) of single-player games this is not the same partition.

\begin{figure}[t]
\centering
\begin{tikzpicture}[scale=0.75,>=stealth',max/.style={draw,rectangle,minimum size=0.5cm},min/.style={draw,circle,minimum size=0.5cm},av/.style={draw, circle,fill, inner sep = 0pt,minimum size = 0.2cm}]

\node[min] (min1) at (-2, 0) {$1$};
\node[min] (min3) at (2, 0) {$2$};

\draw[->] (min1) to[in=160,out=20] (min3);

\draw[->] (min3) to[out=200,in=-20] (min1);

\draw[->] (min3) to[out=-30,in=30,looseness=8] (min3);
\draw[->] (min1) to[in=150,out=210,looseness=8] (min1);

\end{tikzpicture}
\caption{An example of a one-player game with $3$ policies in $\Xi$.} \label{fig:polyh_complex}
\end{figure}
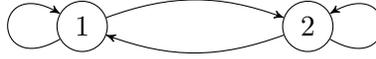

\begin{example}
Consider the one-player DMDP presented in \cref{fig:polyh_complex}. This game has $3$ policies in $\Xi$. Consider the policy given by the edges $(1,2),(2,1)$. The mean payoff of this policy is $\gameval = (\payoff_{12} + \payoff_{21})/2$ and the bias in the graph obtained by fixing these policies is $\bias_1 = 0, \bias_2 = \payoff_{21} - \gameval = (\payoff_{21} - \payoff_{12})/2$. Hence, the polyhedral cone on which these policies are the unique bias-induced policies is given by
\[
\frac{\payoff_{12} + \payoff_{21}}{2} < \payoff_{11}, \quad \frac{\payoff_{12} + \payoff_{21}}{2} < \payoff_{22}.
\]
For the policy given by the edges $(1,1),(2,1)$, the mean payoff is $\gameval = \payoff_{11}$ and the bias is $\bias_1 = 0, \bias_2 = \payoff_{21} - \gameval = \payoff_{21} - \payoff_{11}$. The corresponding polyhedral cone is given by
\[
2\payoff_{11} - \payoff_{21} < \payoff_{12}, \quad \payoff_{11} < \payoff_{22}.
\]
Finally, for the policy given by the edges $(1,2),(2,2)$, the mean payoff is $\gameval = \payoff_{22}$ and the bias is $\bias_1 = \payoff_{12} - \payoff_{22}, \bias_2 = 0$. The corresponding polyhedral cone is given by
\[
\payoff_{22} < \payoff_{11}, \quad 2\payoff_{22} - \payoff_{12} < \payoff_{21}.
\]
In this simple example, the inequalities that we obtained correspond precisely to the condition ``the final cycle of the policy is the only optimal cycle in the DMDP'' and so we got the same partition of $\R^4$ into polyhedral cells as in \cref{polyhedral_cells_one_player}.
\end{example}

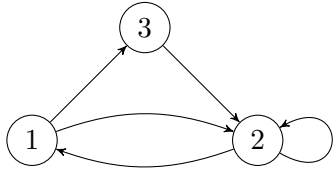
\begin{figure}[t]
\centering
\begin{tikzpicture}[scale=0.75,>=stealth',max/.style={draw,rectangle,minimum size=0.5cm},min/.style={draw,circle,minimum size=0.5cm},av/.style={draw, circle,fill, inner sep = 0pt,minimum size = 0.2cm}]

\node[min] (min1) at (-2, 0) {$1$};
\node[min] (min3) at (2, 0) {$2$};

\node[min] (min2) at (0,2) {$3$};

\draw[->] (min1) to[in=160,out=20] (min3);

\draw[->] (min3) to[out=200,in=-20] (min1);

\draw[->] (min3) to[out=-30,in=30,looseness=8] (min3);

\draw[->] (min1) to (min2);
\draw[->] (min2) to (min3);

\end{tikzpicture}
\caption{An example of a one-player game with $4$ policies in $\Xi$.} \label{fig:polyh_complex_two_ways}
\end{figure}

\begin{example}
Consider the one-player DMDP presented in \cref{fig:polyh_complex_two_ways}. This game has $4$ policies in $\Xi$. Consider the policy given by the edges $(1,2),(2,1)$. The mean payoff of this policy is given by $\gameval = (\payoff_{12} + \payoff_{21})/2$ and the bias for this policy is given by $\bias_1 = 0, \bias_2 = (\payoff_{21} - \payoff_{12})/2, \bias_3 = \payoff_{32} - \payoff_{12}$. Therefore, the corresponding polyhedral cone is given by
\[
\frac{3}{2}\payoff_{12} + \frac{1}{2}\payoff_{21} - \payoff_{32} < \payoff_{13}, \quad \frac{\payoff_{12} + \payoff_{21}}{2} < \payoff_{22} \, .
\]
Note that the second inequality means that the cycle $\{1,2\}$ is better than the cycle $\{2,2\}$, while the first one can be rewritten as $(\payoff_{13} + \payoff_{32} + \payoff_{21})/3 > (\payoff_{12} + \payoff_{21}) / 2$, meaning that the cycle $\{1,2\}$ is better than the cycle $\{1,3,2\}$, similarly to the previous example. However, note that in this example the cycle $\{2,2\}$ can be reached in two different ways. If we consider the policy given by $(1,2),(2,2)$, then the mean payoff is equal to $\gameval = \payoff_{22}$ and the bias for this policy is $\bias_1 = \payoff_{12} - \payoff_{22}, \bias_2 = 0, \bias_{3} = \payoff_{32} - \payoff_{22}$. Hence, the corresponding polyhedral cone $\polyh_1$ is given by
\[
\payoff_{12} + \payoff_{22} - \payoff_{32} < \payoff_{13}, \quad 2\payoff_{22} - \payoff_{12} < \payoff_{21} \, .  
\]
If we consider the policy $(1,3),(2,2)$, which reaches the same cycle, then $\gameval = \payoff_{22}$ and the bias for this policy is $\bias_1 = \payoff_{13} + \payoff_{32} - 2\payoff_{22}, \bias_2 = 0, \bias_{3} = \payoff_{32} - \payoff_{22}$. Hence, the corresponding polyhedral cone $\polyh_2$ is given by
\[
\payoff_{13} + \payoff_{32} - \payoff_{22}  < \payoff_{12}, \quad 3\payoff_{22} - \payoff_{13} - \payoff_{32} < \payoff_{21} \, .  
\]
Note that the cones $\polyh_1$ and $\polyh_2$ are disjoint. Moreover, if $\payoff \in \polyh_1 \cup \polyh_2$, then the cycle $\{2,2\}$ is better for Min that any other cycle. Indeed, for $\polyh_1$ the second inequality means that the cycle $\{2,2\}$ is better than the cycle $\{1,2\}$. Furthermore, the two inequalities together imply that the cycle $\{2,2\}$ is better than $\{1,3,2\}$ because $\payoff_{13} + \payoff_{32} + \payoff_{21} > \payoff_{12} + \payoff_{21} + \payoff_{22} > 3\payoff_{22}$. For $\polyh_2$, the second inequality implies that $\{2,2\}$ is better than $\{1,3,2\}$. Furthermore, the two inequalities together give $\payoff_{12} + \payoff_{21} > 2\payoff_{22}$, meaning that $\{2,2\}$ is better than $\{1,2\}$. If we consider the weights $\payoff_{12} = \payoff_{21} = 0$, $\payoff_{13} = \payoff_{32} = 2$, and $\payoff_{22} = -1$, then $\payoff \in \polyh_1 \setminus \polyh_2$, and if we consider the weights $\payoff_{13} = \payoff_{32} = \payoff_{21} = 0$, $\payoff_{22} = -1$, and $\payoff_{12} = 10$, then $\payoff \in \polyh_2 \setminus \polyh_1$. Hence, if $\payoff \in \polyh_1 \cup \polyh_2$, then $\{2,2\}$ is the only optimal cycle and we have two optimal policies, but only one bias-induced policy. In particular, we get a more precise partition of space into polyhedral cells that the partition from \cref{polyhedral_cells_one_player}.
\end{example}

\section{Smoothed analysis and a condition number}\label{sec:of-threshold_two_players}

We are now ready to perform the smoothed analysis of mean payoff games. We start by introducing our random model. We suppose that the mean payoff game is played on an ergodic graph $\dgraph$. We consider $(\payoff_{ij})$ to be absolutely continuous independent random variables and we denote by $\dens_{ij}$ the density function of $\payoff_{ij}$. We make two assumptions on $\dens_{ij}$. First, we assume that the game is normalized, i.e., $\E(\payoff_{ij}) \in [-1,1]$ for all $i,j$. Second, we assume there exists $\phi > 0$ giving us a bound on the variance: $$\Var(\payoff_{ij}) \le \frac{1}{\phi^2}\, ,$$ and a bound on the density functions: $$\dens_{ij}(y) \le \phi$$ for all $i,j$, and $y$.

\medskip \noindent
This simple random model is general enough to cover the three cases of smoothed analysis, uniform distribution over an interval, and exponentially-distributed weights:

\begin{example}
A standard setting for smoothed analysis is to assume a normal distribution $\payoff_{ij} \sim \Normal(\bar{\payoff}_{ij}, \rho^2)$, where $\bar{\payoff}_{ij}$ are some numbers in $[-1,1]$ and $\rho$ is the common standard deviation. In this case, $\Var(\payoff_{ij}) = \rho^2$ and $\dens_{ij}(y) \le \frac{1}{\rho\sqrt{2\pi}}$, so we can take $\phi \coloneqq 1/\rho$. Similarly, if weights $\payoff_{ij}$ are taken from a uniform distribution, $\payoff_{ij} \in \interval{\bar{\payoff}_{ij} - \frac{\rho}{2}}{\bar{\payoff}_{ij} + \frac{\rho}{2}}$, where $\bar{\payoff}_{ij} \in [-1,1]$, then $\Var(\payoff_{ij}) = \frac{1}{12}\rho^2$ and $\dens_{ij}(y) \le 1/\rho$, so we can take $\phi \coloneqq 1/\rho$. Another model used in the literature~\cite{MathieuWilson:2013} is to suppose that the random weights come from the exponential distribution with mean $1$. In this case, we have $\dens_{ij}(y) \le 1$ and $\Var(\payoff_{ij}) = 1$, so we can take $\phi \coloneqq 1$.
\end{example}

\bigskip\noindent
To study the behavior of the smoothed games, we introduce random variables similar to the ones we have seen in \cref{sec:one_player}. For every $(i,j) \in \edges$ such that $i \in \Minvertices$ we put
\begin{align*}
Z_{ij} = \inf\{x \colon &\text{$(x,\payoff_{-ij}) \in \Generic$ and the MPG with weights $(x,\payoff_{-ij})$ has a pair }\\
&\text{of bias-induced policies $(\sigma,\tau) \in \Xi$ such that $\tau(i) \neq j $}\} \, .
\end{align*}
Analogously, for every $(i,j) \in \edges$ such that $i \in \Maxvertices$ we put
\begin{align*}
Z_{ij} = \sup\{x \colon &\text{$(x,\payoff_{-ij}) \in \Generic$ and the MPG with weights $(x,\payoff_{-ij})$ has a pair }\\
&\text{of bias-induced policies $(\sigma,\tau) \in \Xi$ such that $\sigma(i) \neq j $}\} \, .
\end{align*}
We note that the assumption $(x,\payoff_{-ij}) \in \Generic$ implies that the MPG with weights $(x,\payoff_{-ij})$ has exactly one pair of bias-induced policies. Furthermore, we note that the variables $Z_{ij}$ may have values in $\{-\infty,+\infty\}$. This does not change the fact that we have an analogue of \cref{prob_estimate_one_player_second}.

\begin{lemma}\label{prob_estimate_two_players}
For any $\alpha > 0$ we have
\[
\Prob(\exists ij, \abs{\payoff_{ij} - Z_{ij}} \le \alpha) \le 2\alpha m \phi \, .
\]
\end{lemma}
\begin{proof}
Fix an edge $(i,j)$ and note that $Z_{ij}$ depends only on $\payoff_{-ij}$ and not on $\payoff_{ij}$. Let $\hat{\dens}(\payoff_{-ij}) = \prod_{(i',j') \in \edges \setminus \{(i,j)\}} \dens_{i'j'}(\payoff_{i',j'})$. By Fubini's theorem we have
\begin{align*}
\Prob(Z_{ij} - \alpha \le \payoff_{ij} \le Z_{ij} + \alpha) &= \Prob(Z_{ij} \in \R \land Z_{ij} - \alpha \le \payoff_{ij} \le Z_{ij} + \alpha) \\
&= \int_{\{Z_{ij} \in \R\}}\Bigl( \int_{Z_{ij} - \alpha}^{Z_{ij} + \alpha} \dens_{ij}(\payoff_{ij}) d \payoff_{ij} \Bigr) \hat{\dens}(\payoff_{-ij}) d \payoff_{-ij} \\
&\le 2\alpha \phi \Prob(Z_{ij} \in \R) \le 2\alpha \phi \, .
\end{align*}
The union bound finishes the proof.
\end{proof}

\noindent
Furthermore, we can express the values of $Z_{ij}$ in terms of the functions $\gameval^{\sigma,\tau}$,$\bias^{\sigma,\tau}$.

\begin{lemma}\label{threshold_two_players}
Suppose that $\payoff \in \polyh^{\sigma,\tau}$ for some $(\sigma,\tau) \in \Xi$. Then for every $(i,j) \in \edges$ that is not used in $(\sigma,\tau)$ we have $Z_{ij} = \gameval^{\sigma,\tau}(\payoff) + \bias_i^{\sigma,\tau}(\payoff) - \bias_j^{\sigma,\tau}(\payoff)$.
\end{lemma}
\begin{proof}
Fix $(i,j) \in \edges$ such that $i \in \Minvertices$ and $j \neq \tau(i)$. By \cref{two_player_stability}, we have $(x,\payoff_{-ij}) \in \polyh^{\sigma,\tau}$ for all $x > \gameval^{\sigma,\tau}(\payoff) + \bias_i^{\sigma,\tau}(\payoff) - \bias_j^{\sigma,\tau}(\payoff)$. This proves that $Z_{ij} \le \gameval^{\sigma,\tau}(\payoff) + \bias_i^{\sigma,\tau}(\payoff) - \bias_j^{\sigma,\tau}(\payoff)$. To prove the opposite inequality, fix $x \in \R$, suppose that $(x, \payoff_{-ij}) \in \Generic$, and that $\tilde{\sigma},\tilde{\tau}$ are such that $(x,\payoff_{-ij}) \in \polyh^{\tilde{\sigma},\tilde{\tau}}$ and $\tilde{\tau}(i) \neq j$. Then, by applying \cref{two_player_stability} to $\polyh^{\sigma,\tau}$ and $\polyh^{\tilde{\sigma},\tilde{\tau}}$ we get that $\bigl(\max\{x,\payoff_{ij}\},\payoff_{-ij}\bigr) \in \polyh^{\sigma,\tau} \cap \polyh^{\tilde{\sigma},\tilde{\tau}}$. Therefore, we have $(\tilde{\sigma},\tilde{\tau}) = (\sigma,\tau)$ and $x > \gameval^{\sigma,\tau}(\payoff) + \bias_i^{\sigma,\tau}(\payoff) - \bias_j^{\sigma,\tau}(\payoff)$ by the description of $\polyh^{\sigma,\tau}$ from \cref{eq:polyhedral_cone}. In particular, the infimum of all such $x$ also has $Z_{ij} \ge \gameval^{\sigma,\tau}(\payoff) + \bias_i^{\sigma,\tau}(\payoff) - \bias_j^{\sigma,\tau}(\payoff)$. The proof is analogous if $i \in \Maxvertices$.
\end{proof}

Before continuing, let us give an intuitive interpretation of \cref{prob_estimate_two_players,threshold_two_players}. In \cref{polyhedral_complex}, we proved that, for random games, the ergodic equation has a single solution (up to adding a constant to a bias) and that each maximum/minimum of this equation is achieved by a single edge. \Cref{prob_estimate_two_players,threshold_two_players} strengthen this statement: with high probability, the difference between the edge that achieves the maximum/minimum and the ``second best'' edge is large. This leads to the following analogue of \cref{robust_cycle_one_player_second}.

\begin{theorem}\label{robust_bias_policy}
Let $\delta \coloneqq 1/(4n(2n+1)m\phi)$. If the weights $\payoff$ are picked at random, then with high probability the whole ball $B_{\infty}(\payoff, \delta)$ is included in a single polyhedron $\polyh^{\sigma,\tau}$. More formally,
\[
\Prob(\exists (\sigma,\tau), B_{\infty}(\payoff, \delta) \subseteq \polyh^{\sigma,\tau}) \ge 1 - \frac{1}{n} \, .
\]
\end{theorem}
\begin{proof}
By applying \cref{prob_estimate_two_players} with $\alpha \coloneqq 1/(2nm\phi)$ we get that
\[
\Prob(\forall (i,j) \in \edges, \payoff_{ij} < Z_{ij} - \alpha \lor \payoff_{ij} > Z_{ij} + \alpha) \ge 1 - \frac{1}{n} \, .
\]
Let $\payoff \in \Generic$ be such that this event holds and let $(\sigma,\tau)$ be such that $\payoff \in \polyh^{\sigma,\tau}$. Fix an edge $(i,j) \in \edges$ that is not used in $(\sigma,\tau)$ and suppose that $i \in \Minvertices$. Then, by \cref{threshold_two_players} we have
\[
\alpha + \gameval^{\sigma,\tau}(\payoff) + \bias_i^{\sigma,\tau}(\payoff) - \bias_j^{\sigma,\tau}(\payoff) < \payoff_{ij} \, .
\]
Hence, by \cref{linear_functions} and by our choice of $\alpha$ and $\delta$, for every $\tilde{\payoff} \in B_{\infty}(\payoff, \delta)$ we have $\gameval^{\sigma,\tau}(\tilde{\payoff}) + \bias_i^{\sigma,\tau}(\tilde{\payoff}) - \bias_j^{\sigma,\tau}(\tilde{\payoff}) < \tilde{\payoff}_{ij}$. Analogously, if $i \in \Maxvertices$ then $\gameval^{\sigma,\tau}(\tilde{\payoff}) + \bias_i^{\sigma,\tau}(\tilde{\payoff}) - \bias_j^{\sigma,\tau}(\tilde{\payoff}) > \tilde{\payoff}_{ij}$ for all $\tilde{\payoff} \in B_{\infty}(\payoff, \delta)$. Therefore, $B_{\infty}(\payoff, \delta) \subseteq \polyh^{\sigma,\tau}$ by \cref{eq:polyhedral_cone}.
\end{proof}

As we explain in more detail in \cref{sec:algorithms}, this theorem alone does not directly lead to a fully general and efficient algorithm for random mean payoff games. In order to obtain such an algorithm, we introduce a new condition number for mean payoff games and we show that this condition number is small for random games. We note that a condition number for the value iteration algorithm for MPGs was proposed in \cite{AllamigeonGaubertKatzSkomra:2022}, but this condition number differs from the one below. We leave the problem of comparing these two condition numbers for future work. 

\begin{definition}\label{def:condition-number}
Given $\payoff \in \Generic$, we define a \emph{condition number} of the associated MPG as
\[
\Cond(\payoff) \coloneqq \frac{\max\{\abs{\payoff_{ij} - \gameval} \colon (i,j) \in \edges\}}{\min\{\abs{\payoff_{ij} - \gameval + \bias_j - \bias_i} \colon (i,j) \in \edges, \payoff_{ij} - \gameval + \bias_j - \bias_i \neq 0\} } \, ,
\]
where $\gameval$ is the value of the MPG and $\bias$ is the bias (since $\payoff \in \Generic$, the bias is unique up to a constant).
\end{definition}
We note that this condition number does not change when the weights of the game are multiplied by the same positive constant. Likewise, the condition number does not change if the same constant is added to all weights. Numerous algorithms for MPGs, such as policy iteration algorithms, are also invariant by these two operations.

\begin{remark}
We note that the condition number above is not well defined if the game has only one pair of policies. Such zero-player games are trivial to solve for policy iteration algorithms, so in this case we may simply put $\Cond(\payoff) \coloneqq 1$.
\end{remark}

\begin{theorem}\label{cond_estimate}
Random mean payoff games are well conditioned with high probability. More formally, for every $\varepsilon > 0$ we have $\Prob\bigl(\Cond \ge \frac{8m}{\varepsilon}(\phi + \sqrt{\frac{2m}{\varepsilon}})\bigr) \le \varepsilon$. In particular, $\Prob\bigl(\Cond \ge 8nm(\phi + \sqrt{2nm})\bigr) \le 1/n$.
\end{theorem}
\begin{proof}
By applying \cref{prob_estimate_two_players} with $\alpha \coloneqq \varepsilon/(4m\phi)$ we get
\[
\Prob(\exists ij, \abs{\payoff_{ij} - Z_{ij}} \le \alpha) \le \frac{\varepsilon}{2} \, .
\]
Furthermore, since we assume that $\E(\payoff_{ij}) \in [-1,1]$ for all $(i,j) \in \edges$, the Chebyshev inequality gives
\[
\Prob(\abs{\payoff_{ij}} \ge 1 +a) \le \Prob(\abs{\payoff_{ij} - \E(\payoff_{ij})} \ge a) \le \frac{\Var(\payoff_{ij})}{a^2} \le \frac{1}{a^2\phi^2}
\]
for all $a > 0$. Hence, by taking $a = \frac{1}{\phi}\sqrt{\frac{2m}{\varepsilon}}$ and applying the union bound we get
\[
\Prob(\supnorm{\payoff} \ge 1 + \frac{1}{\phi}\sqrt{\frac{2m}{\varepsilon}}) \le \frac{\varepsilon}{2} \, .
\]
In particular,
\begin{equation}\label{eq:event_for_cond}
\Prob(\supnorm{\payoff} < 1 + \frac{1}{\phi}\sqrt{\frac{2m}{\varepsilon}} \land \forall (i,j) \in \edges, \abs{\payoff_{ij} - Z_{ij}} > \alpha) \ge 1 - \varepsilon \, .
\end{equation}
Let $\payoff \in \Generic$ be such that this events holds and $(\sigma,\tau)$ be such that $\payoff \in \polyh^{\sigma,\tau}$. By \cref{solution_from_subgraph,threshold_two_players}, we have $\payoff_{ij} - Z_{ij} = \payoff_{ij} - \gameval + \bias_j - \bias_i$ whenever the right-hand side is not equal to zero, meaning, whenever the edge $(i,j)$ is not in the unique bias-induced policies $(\sigma,\tau)$. Furthermore, $\gameval$ is the mean value of the cycle in $\dgraph^{\sigma,\tau}$, so we have $\abs{\gameval} \le \supnorm{\payoff}$ and $\max\{\abs{\payoff_{ij} - \gameval} \colon (i,j) \in \edges\} \le 2\supnorm{\payoff}$. Hence, we get
\begin{align*}
  \Cond(\payoff) \; %
  & = \frac{\max\{\abs{\payoff_{ij} - \gameval}\}}{\min\{\abs{\payoff_{ij} - \gameval + \bias_j - \bias_i} \colon \text{$(i,j)$ not in $(\sigma,\tau)$}\} } \\
  & \le \frac{2\supnorm{\payoff}}{\min\{\abs{\payoff_{ij} - Z_{ij}} \colon \text{$(i,j)$ not in $(\sigma,\tau)$}\} }\\
  &< \frac{2(1 + \frac{1}{\phi}\sqrt{\frac{2m}{\varepsilon}})}{\alpha} = \frac{8m\phi}{\varepsilon} + 8\sqrt{2}(\frac{m}{\varepsilon})^{3/2} \, . \qedhere
\end{align*}
\end{proof}

\begin{remark}
We note that if $\E(\payoff_{ij}) = 0$ for all $i,j$, then the proof above gives $\Prob(\Cond \ge 8\sqrt{2}(\frac{m}{\varepsilon})^{3/2} ) \le \varepsilon$, independently of $\phi$.
\end{remark}

We will also need an analogue of \cref{cond_estimate} for random weights which are not known exactly. This is given in the next lemma.

\begin{lemma}\label{cond_estimate_approx}
Fix $\varepsilon > 0$ and let $\delta \coloneqq \frac{\varepsilon}{16(n+1)m\phi}$. Then, with probability at least $1 - \varepsilon$, the following three things happen simultaneously: the supremum norm of $\payoff$ is smaller than $1 + \frac{1}{\phi}\sqrt{\frac{2m}{\varepsilon}}$, the ball $B_{\infty}(\payoff, \delta)$ is contained in a single polyhedron $\polyh^{\sigma,\tau}$, and all the games with weights $\tilde{\payoff}$ in this ball are well-conditioned, $\Cond(\tilde{\payoff}) < \frac{16m}{\varepsilon}(\phi + \sqrt{\frac{2m}{\varepsilon}}) + o(1)$. In symbols,
\begin{align*}
\Prob\Bigl(&\supnorm{\payoff} < 1 + \frac{1}{\phi}\sqrt{\frac{2m}{\varepsilon}} \ \land \ \exists (\sigma,\tau), B_{\infty}(\payoff, \delta) \subseteq \polyh^{\sigma,\tau} \\
&\land \forall \tilde{\payoff} \in B_{\infty}(\payoff, \delta), \Cond(\tilde{\payoff}) < \frac{16m}{\varepsilon}(\phi + \sqrt{\frac{2m}{\varepsilon}}) + \frac{1}{n+1} \Bigr) \ge 1 - \varepsilon \, .
\end{align*}
\end{lemma}

\begin{proof}
As in the proof of \cref{cond_estimate} we fix $\alpha \coloneqq \varepsilon/(4m\phi)$, so $\delta = \frac{\alpha}{4(n+1)}$, and we let $\payoff \in \Generic$ be such that the event \cref{eq:event_for_cond} holds. Let $(\sigma,\tau)$ be such that $\payoff \in \polyh^{\sigma,\tau}$. Furthermore, let $\tilde{\payoff} \in B_{\infty}(\payoff,\delta)$. As in the proof of \cref{robust_bias_policy}, we have $\tilde{\payoff} \in \polyh^{\sigma,\tau}$. Even more, we have chosen the parameter $\delta$ in such a way that for every edge $(i,j)$ that is not used in $(\sigma,\tau)$, by \cref{linear_functions}, we get 
\begin{align*}
  \abs{\tilde{\payoff}_{ij} - \gameval^{\sigma,\tau}(\tilde{\payoff}) + \bias_i^{\sigma,\tau}(\tilde{\payoff}) - \bias_j^{\sigma,\tau}(\tilde{\payoff})} %
  & \ge \abs{{\payoff}_{ij} - \gameval^{\sigma,\tau}(\payoff) + \bias_i^{\sigma,\tau}(\payoff) - \bias_j^{\sigma,\tau}(\payoff)} - (2n + 2)\supnorm{\payoff - \tilde{\payoff}}\\
  & = \abs{{\payoff}_{ij} - Z_{ij}} - 2(n + 1)\supnorm{\payoff - \tilde{\payoff}}\\
  & \ge \abs{{\payoff}_{ij} - Z_{ij}} - 2(n + 1) \delta\\
  & = \abs{{\payoff}_{ij} - Z_{ij}} - \frac{\alpha}{2}\\
  & > \frac{\alpha}{2} \, .
\end{align*}
Moreover, since $\supnorm{\payoff} < 1 + \frac{1}{\phi}\sqrt{\frac{2m}{\varepsilon}}$, we get $\supnorm{\tilde{\payoff}} < 1 + \frac{1}{\phi}\sqrt{\frac{2m}{\varepsilon}} + \delta$. Hence, as in the proof of \cref{cond_estimate} we get
\[
\Cond(\tilde{\payoff}) < \frac{2(1 + \frac{1}{\phi}\sqrt{\frac{2m}{\varepsilon}} + \delta)}{\alpha/2} = \frac{16m\phi}{\varepsilon} + 16\sqrt{2}(\frac{m}{\varepsilon})^{3/2} + \frac{1}{n+1} \, . \qedhere
\]
\end{proof}


\section{An efficient algorithm for random MPGs}\label{sec:algorithms}

We can now present our algorithms to solve MPGs in smoothed polynomial time. To start, let us note that smoothed analysis can be done in at least two computational models. In the real model of computation, we deal with random weights that are given exactly (as real numbers). The task is to solve such random games in polynomial time, using an algorithm that can be implemented in the real model of computation.\footnote{Here we can choose any reasonable model, e.g. Blum-Shub-Smale, or real arithmetic RAMs. One needs to be able to do arithmetic operations (including division) and branch on comparisons.} In the oracle model of computation, the user has access to an oracle that outputs the random weights bit-by-bit. The task of the user is then to solve such random games in polynomial time on a standard Turing machine. We note that it is impossible to solve arbitrary MPGs with real weights in the oracle model. Indeed, it is easy to construct examples of MPGs with real weights in which the optimal policies are revealed only after acquiring an arbitrarily high number of bits. Therefore, our aim in the oracle model is to propose an algorithm that stops with probability $1$ and has a polynomial smoothed complexity.

The oracle model was also studied in the retracted work \cite{BorosElbassioniFouzGurvich:2011}. In this work, the authors propose to solve random MPGs using the following strategy: we use the oracle to get $O(\log(n\phi))$ bits of the weights, find optimal policies of the resulting ``truncated'' MPG using any pseudopolynomial-time algorithm, and we claim that these policies are also optimal in the original game with high probability. The correctness of this approach relies on the property \labelcref{enum:third}. However, as we discussed earlier, we were not able to show that \labelcref{enum:third} is true for random games, because our analysis only applies to \emph{bias-induced} policies, and not to all policies in general. This naturally suggests the following modification of the strategy of \cite{BorosElbassioniFouzGurvich:2011}: we use the oracle to get $O(\log(n\phi))$ bits of the weights, find bias-induced (or Blackwell-optimal) policies of the resulting ``truncated'' MPG using some pseudopolynomial-time algorithm, and we claim that these policies are also optimal in the original game with high probability. In order to execute this strategy, one needs a pseudopolynomial-time algorithm that finds bias-induced policies. We are aware of one such algorithm, namely the Gurvich-Karzanov-Khachiyan (GKK) algorithm \cite{gurvich}, which was adapted to be pseudopolynomial-time by Pisaruk \cite{pisaruk1999mean} (and further adapted to the stochastic case in \cite{boros_gurvich_makino}). Hence, using the GKK algorithm, we can obtain the smoothed complexity result in the oracle model. However, in the real model of computation we work with exact real weights, and we could adapt the GKK algorithm to the real model, but it is not clear how to upper bound its running time. Indeed, the pseudopolynomial time bounds of this algorithm depend crucially on the bitlength of the input weights\footnote{There exists a combinatorial upper-bound on the running time of the GKK algorithm \cite{ohlmann2022gkk}, and it does not depend on the bitlength of the input weights, but this bound is exponential in $n$.} Also, the GKK algorithm does not handle discounted games. So, in order to give an algorithm that works in all of these cases, we proceed in the opposite direction: we propose an algorithm that runs in polynomial smoothed time in the real model of computation, and then we show that the same algorithm still gives a correct answer if we run it on a ``truncated'' MPG with $O(\log(n\phi))$ bits of precision.


\subsection{Upper bound on the Blackwell threshold}
In order to find bias-induced policies in the real model of computation, one can use policy iteration algorithms. For instance, the two-player analogue of Howard's algorithm developed in~\cite{CochetGaubertGunawardena:1999,gaubert1998duality,dhingra2006how} outputs a pair bias-induced policies and the lexicographic policy iteration~\cite{kallenberg2002handbook,hordijk2002handbook} outputs Blackwell-optimal policies. The issue with these algorithms is that we do not currently have any upper bound on their running time that is significantly better than the naive bound obtained by enumerating all policies. To overcome this issue, we propose a different variant of the policy iteration method that we present in \cref{alg:increasing_discount_mpg}. In this variant, the algorithm keeps track of a discount factor that is allowed to evolve over time. In particular, two policies can occur multiple times during the execution of this algorithm, but for different values of the discount factor. The advantage of our approach is that for each discount rate we can use the strongly-polynomial upper bounds for the running time of the policy iteration method in discounted games \cite{ye2011simplex,hansen2013strategy}. In order to solve mean-payoff games, our algorithm relies on the existence of the \emph{Blackwell threshold} discussed in \cref{th:liggett_lippman}: for any MPG there exists a number $\disc^* \in \interval[open]{0}{1}$ such that if a policy is optimal for any $\disc^* < \disc < 1$, then it is optimal for all $\disc^* < \disc < 1$. It turns out that the condition number $\Cond$ which we introduced in the previous section controls the Blackwell threshold, in the sense that well-conditioned MPGs have low Blackwell threshold.

\begin{theorem}\label{conditioned_discount}
Suppose that $\payoff \in \polyh^{\sigma,\tau}$ and fix $1 > \disc > 1 - \frac{1}{6n^2\Cond(\payoff)}$. Then, $(\sigma,\tau)$ is the unique pair of optimal policies in the discounted game with discount factor $\disc$. In other words, the Blackwell threshold is bounded by $\gamma^* \le 1 - \frac{1}{6n^2\Cond(\payoff)}$.
\end{theorem}

The proof of \cref{conditioned_discount} relies on the following estimate.

\begin{proposition}[{\cite[Theorem~5.2]{zwick_paterson}}]\label{discount_to_mean}
If $\gameval^{(\disc)}$ denotes the value of a discounted game with discount factor $\disc$ and $\gameval$ denotes the value of the mean payoff game with the same weights, then $\supnorm{\gameval^{(\disc)} -  \gameval} \le 2n(1 - \disc)\supnorm{\payoff}$.
\end{proposition}

\begin{proof}[Proof of \cref{conditioned_discount}]
Let $\gameval$ be the value of the mean payoff game. Consider the modified game with payoffs $\hat{\payoff}_{ij} \coloneqq \payoff_{ij} - \gameval$. Note that this modification does not change the optimal policies, neither in the mean payoff game nor in the discounted game. Furthermore, the value of the modified mean payoff game is equal to $0$ and we have $\Cond(\payoff) = \frac{\supnorm{\hat{\payoff}}}{\beta}$, where $\beta \coloneqq \min\{\abs{\hat{\payoff}_{ij} + \bias_j - \bias_i} \colon (i,j) \in \edges, \hat{\payoff}_{ij} + \bias_j - \bias_i \neq 0\}$ and $\bias$ is the bias such that $u_k = 0$ where $k$ is the smallest index of a vertex in the cycle of $\dgraph^{\sigma,\tau}$. Suppose that player Max plays according to $\sigma$ in the discounted game. Let $\tau'$ be a policy of Min that is an optimal response to $\sigma$. We will show that $\tau' = \tau$. Suppose that this is not the case and let $i \in \Minvertices$ be such that $\tau'(i) \neq \tau(i)$. Consider two cases.

First case: the vertex $i$ is on a cycle in the graph $\dgraph^{\sigma,\tau'}$. By \cref{discount_to_mean}, the value of the discounted game at $i$ is at most $2n(1-\disc)\supnorm{\hat{\payoff}}$. Since $\tau'$ is an optimal response to $\sigma$, the value of $i$ in the discounted game on $\dgraph^{\sigma,\tau'}$ is at most $2n(1-\disc)\supnorm{\hat{\payoff}}$. On the other hand, let $i = i_0, i_1, \dots, i_l$ denote the cycle in $\dgraph^{\sigma,\tau'}$ that contains $i$. Since $i_1 \neq \tau(i)$,
we have $\hat{\payoff}_{i_0i_1} + \bias_{i_1} - \bias_{i_0} \ge \beta$.
Moreover, for all $p \ge 1$ we have $\hat{\payoff}_{i_{p}i_{p+1}} + \bias_{i_{p+1}} - \bias_{i_p} \ge 0$: if $i_p\in\Maxvertices$, the left-hand side is $0$ because $\sigma$ is induced by $u$, and if $i_p\in\Minvertices$ it is either $0$ if $i_{p+1} = \tau(i_p)$ or $\ge \beta$ if $i_{p+1} \neq \tau(i_p)$.
%
%
%
Hence, the mean weight of this cycle is at least $\beta/n$ and \cref{discount_to_mean} implies that the value of $i$ in the discounted game on $\dgraph^{\sigma,\tau'}$ is at least $\frac{\beta}{n} - 2n(1 - \disc)\supnorm{\hat{\payoff}}$. But by assumption $\Cond(r) < \frac{1}{6 n^2 (1-\disc)}$, hence
  \begin{align*}
    \frac{\beta}{n} - 2n(1 - \disc)\supnorm{\hat{\payoff}} %
    & = \frac{\supnorm{\hat{\payoff}}}{\Cond(r) n} - 2n(1 - \disc)\supnorm{\hat{\payoff}}\\
    & > 6 n (1-\disc) \supnorm{\hat{\payoff}} - 2n(1 - \disc)\supnorm{\hat{\payoff}}\\
    & > 2n(1-\disc)\supnorm{\hat{\payoff}} \, ,
  \end{align*}
and we get a contradiction. Thus $\tau'(i) = \tau(i)$.

Second case: the vertex $i$ is not on a cycle in the graph $\dgraph^{\sigma,\tau'}$. By the first case, every cycle in the graph $\dgraph^{\sigma,\tau'}$ appears in $\dgraph^{\sigma,\tau}$, and so $\dgraph^{\sigma,\tau'}$ (which must contain at least one cycle) contains exactly one cycle, namely the cycle of $\dgraph^{\sigma,\tau}$. Let $k$ be the vertex on this cycle with the smallest index. Furthermore, let $i= i_0, i_1, \dots, i_{l} = k$ be the path from $i$ to $k$ in $\dgraph^{\sigma,\tau}$ and $i= i'_0, i'_1, \dots, i'_{l'} = k$ be the path from $i$ to $k$ in $\dgraph^{\sigma,\tau'}$. By \cref{th:shapley} and since $\hat{\payoff}_{ij} = \payoff_{ij} - \lambda = u_i - u_j$ (\cref{def:bias,def:ergodic}) the value of $i$ in the discounted game on $\dgraph^{\sigma,\tau}$ is equal to
\begin{align*}
&(1 - \disc)\bigl( (\bias_{i_0} - \bias_{i_1}) + \disc(\bias_{i_1} - \bias_{i_2}) + \dots \disc^{l-1}(\bias_{i_{l-1}} - \bias_{i_l}) \bigr) + \disc^l \chi = \\
&\disc^l\chi +(1 - \disc)\bias_{i} - (1 - \disc)^2\sum_{q = 1}^{l - 1}\disc^{q-1}\bias_{i_{q}} \, ,
\end{align*}
where $\chi$ is the value of $k$. Analogously, the value of $i$ in the game on $\dgraph^{\sigma,\tau'}$ is at least
\begin{align*}
&(1 - \disc)\beta + (1 - \disc)\bigl( (\bias_{i'_0} - \bias_{i'_1}) + \disc(\bias_{i'_1} - \bias_{i'_2}) + \dots \disc^{l'-1}(\bias_{i'_{l'-1}} - \bias_{i'_{l'}}) \bigr) + \disc^{l'} \chi = \\
&\disc^{l'} \chi + (1 - \disc)\beta + (1 - \disc)\bias_{i} - (1 - \disc)^2\sum_{q = 1}^{l' - 1}\disc^{q-1}\bias_{i'_{q}} \, .
\end{align*}
Therefore (dividing by $1 - \disc$ and subtracting the term $\frac{1}{1-\disc} \chi$) to get a contradiction it is enough to prove that
\begin{equation}\label{eq:disc_on_path}
\beta + \frac{\disc^{l'} -1}{1 - \disc} \chi - (1 - \disc)\sum_{q = 1}^{l' - 1}\disc^{q-1}\bias_{i'_{q}}  > \frac{\disc^l - 1}{1 - \disc} \chi - (1 - \disc)\sum_{q = 1}^{l - 1}\disc^{q-1}\bias_{i_{q}} \, .
\end{equation}
To do so, note that by \cref{discount_to_mean} we have $\abs{\chi} \le 2n(1-\disc)\supnorm{\hat{\payoff}}$, since $\tau'$ is an optimal response to $\sigma$. Furthermore, $\abs{\frac{\disc^l - 1}{1 - \disc}} = 1 + \disc + \dots + \disc^{l-1} < n$ and $\supnorm{\bias} \le n\supnorm{\hat{\payoff}}$ by \cref{one_cycle_graph}. Hence,
\begin{align*}
\left\lvert\frac{\disc^l - 1}{1 - \disc}\chi - (1 - \disc)\sum_{q = 1}^{l - 1}\disc^{q-1}\bias_{i_{q}}\right\rvert &\le 2n^2(1-\disc)\supnorm{\hat{\payoff}} + (1 - \disc)n^2\supnorm{\hat{\payoff}} \\
&= 3n^2(1-\disc)\supnorm{\hat{\payoff}} < \frac{1}{2}\beta \, .
\end{align*}
The same is true if we replace $l$ by $l'$ and so \cref{eq:disc_on_path} holds, giving the desired contradiction.

Hence, we have shown that $\tau$ is the unique best response policy to $\sigma$ in the discounted game. Analogously, $\sigma$ is the unique best response policy to $\tau$. This shows that $(\sigma,\tau)$ are optimal. Moreover, they are the only pair of optimal policies. Indeed, if we denote by $\xi$ the value of the discounted game, then any optimal policy $\tau'$ of Min must guarantee a payoff at most $\xi$ against $\sigma$, so $\tau'$ is not worse that $\tau$ against $\sigma$. This implies that $\tau' = \tau$ because $\tau$ is the unique best response to $\sigma$. The same argument holds for $\sigma$.
\end{proof}

\begin{figure}[t]
\centering
\begin{footnotesize}
  \begin{algorithmic}[1]
    \Procedure{$\SwitchMax$}{$\sigma,\tau,\disc$}
    \State $\gameval \gets \text{value of the zero-player game with fixed $(\sigma,\tau,\disc)$}$
    \For{$i \in \Maxvertices$}
    	\If{$\gameval_i \neq \max_{(i,j) \in \edges}\{(1 - \disc) \payoff_{ij} + \disc \gameval_j\} $}
		\State $k \gets \argmax_{(i,j) \in \edges}\{(1 - \disc) \payoff_{ij} + \disc \gameval_j\}$ \Comment{$k$ is any state in $\argmax$}
		\State $\sigma(i) \gets k$
	\EndIf
    \EndFor
    \State \Return $\sigma$
\EndProcedure
\end{algorithmic}
\vspace*{0.1cm}
  \begin{algorithmic}[1]
    \Procedure{$\SwitchMin$}{$\sigma,\tau,\disc$}
    \State $\gameval \gets \text{value of the zero-player game with fixed $(\sigma,\tau,\disc)$}$
    \For{$i \in \Minvertices$}
    	\If{$\gameval_i \neq \min_{(i,j) \in \edges}\{(1 - \disc) \payoff_{ij} + \disc \gameval_j\} $}
		\State $k \gets \argmin_{(i,j) \in \edges}\{(1 - \disc) \payoff_{ij} + \disc \gameval_j\}$ \Comment{$k$ is any state in $\argmin$}
		\State $\tau(i) \gets k$
	\EndIf
    \EndFor
    \State \Return $\tau$
\EndProcedure
\end{algorithmic}
\end{footnotesize}
\caption{Switching procedures used in policy iteration.} 
\label{alg:switches}
\end{figure}

\subsection{Algorithms in the real model}

Our main algorithm, working in the real model of computation, comes from combining the bound on the Blackwell threshold from \cref{conditioned_discount} with the estimate of condition number in \cref{cond_estimate} and the bounds on the number of iterations executed by the policy iteration algorithm equipped with the greedy all-switches rule in discounted games. To give a full presentation of the algorithm, we recall the details of the policy iteration algorithm for discounted games with fixed discount rate $\disc$, see, e.g., the survey \cite{payoffchapter} for more discussion. This algorithm relies on two switching subroutines, $\SwitchMax(\sigma,\tau,\disc)$ and $\SwitchMin(\sigma,\tau,\disc)$, whose pseudocodes are given in \cref{alg:switches}. Roughly speaking, the procedure $\SwitchMax$ takes $\sigma$ as a candidate policy of player Max, computes the value of the zero-player discounted game when $(\sigma,\tau)$ are fixed, and checks if this value satisfies the equations from \cref{th:shapley} when applied to the one-player game obtained by fixing only $\tau$. If this is not the case, then $\SwitchMax$ improves $\sigma$ by switching some of the edges using these equations as a guideline. There are many possible ways to switch the edges, but the greedy all-switches rule given in \cref{alg:switches} is perhaps the most natural: we perform a switch for every state in which the equations from \cref{th:shapley} are not satisfied, and we do it in a greedy manner, by going to an edge that achieves the maximum in this equation. Intuitively, this should improve the value that Max obtains against $\tau$, and it can be proven that this is indeed the case. We repeat the procedure until we cannot make any more switches, which means that we have found an optimal response to $\tau$. The procedure $\SwitchMin$ works in the same way. Solving two-player games proceeds by performing switches at two levels: we start with an initial policy $(\sigma_0,\tau_0)$, we find the optimal response $\sigma_1$ to $\tau_0$ using policy iteration for player Max, then we perform a single switch for player Min to get $\tau_1$, we find an optimal response $\sigma_2$ to $\tau_1$, preform a single switch for player Min to get $\tau_3$ etc. The pseudocode of this procedure is given in \cref{alg:policy_iteration}. We use the following theorem, which improves the earlier bounds of \cite{ye2011simplex,hansen2013strategy}.

\begin{figure}[t]
\centering
\begin{footnotesize}
  \begin{algorithmic}[1]
    \Procedure{$\DiscPI$}{$\sigma,\tau,\disc$} \Comment{$(\sigma,\tau)$ are any initial policies}
    \State $\sigma' \gets \sigma$, $\tau' \gets \tau$
    \Repeat
         \State $\tau \gets \tau'$
   	 \Repeat
        		\State $\sigma \gets \sigma'$
    		\State $\sigma' \gets \SwitchMax(\sigma,\tau,\disc)$
	  \Until{$\sigma = \sigma'$} \Comment{$\sigma$ is the best response to $\tau$}
	  \State $\tau' \gets \SwitchMin(\sigma,\tau,\disc)$
    \Until{$\tau = \tau' $}
    \State \Return $(\sigma, \tau)$ \Comment{$(\sigma,\tau)$ are optimal}
\EndProcedure
\end{algorithmic}
\end{footnotesize}
\caption{Policy iteration algorithm for discounted games.} 
\label{alg:policy_iteration}
\end{figure}

\begin{theorem}[\cite{AkianGaubert:2013}]\label{discount_iter_bound}
The policy iteration algorithm equipped with the greedy all-switches rule (\cref{alg:policy_iteration}) is correct and halts after performing $O\bigl((\frac{m}{1- \disc})^2 \log^2 \frac{1}{1 - \disc}\bigr)$ switches.
\end{theorem}
\begin{remark}
We note that the bound in \cref{discount_iter_bound} counts the switches of both players. The original bound of $O\bigl(\frac{m}{1- \disc} \log \frac{1}{1 - \disc}\bigr)$ stated in \cite[Theorem~5]{AkianGaubert:2013} counts only the switches of player Min. Since we can apply the same bound of \cite[Theorem~5]{AkianGaubert:2013} to all the times that policy iteration is used to find the optimal response of Max to the current policy of Min, we get the squared bound of \cref{discount_iter_bound}.
\end{remark}
\begin{remark}
We note that finding the value of a zero-player discounted game (which is necessary to perform $\SwitchMax$ and $\SwitchMin$) can be done by solving a linear system of equations given in \cref{th:shapley}.
\end{remark}

\begin{figure}[t]
\centering
\begin{footnotesize}
  \begin{algorithmic}[1]
    \Procedure{$\IncDiscPI$}{$\sigma,\tau$} \Comment{$(\sigma,\tau)$ are any initial policies}
    \State $\disc \gets 0$
    \Repeat
         \State $\disc \gets (1 + \disc)/2$
         \State $(\sigma, \tau) \gets \DiscPI(\sigma,\tau,\disc)$
         \State\label{line:value-blackwell-bias} $(\gameval, \bias) \gets $ value and Blackwell bias of the zero-player mean-payoff game with fixed $(\sigma,\tau)$ \Comment{\textit{if the value of the zero-player game is not constant for all states, pick any $(\gameval,\bias)$}}
      \Until $(\gameval,\bias)$ solves the ergodic equation of the two-player mean-payoff game
      \State \Return $(\sigma, \tau)$ \Comment{$(\sigma,\tau)$ are optimal in the mean-payoff game}
\EndProcedure
\end{algorithmic}
\end{footnotesize}
\caption{Policy iteration algorithm for mean-payoff games with polynomial smoothed complexity.} 
\label{alg:increasing_discount_mpg}
\end{figure}

We can now give our main algorithm for mean-payoff games. Following \cite{beier2004typical,roglin2007smoothed}, we say that an algorithm has \emph{polynomial smoothed complexity} if there exists a polynomial $\poly(x_1,x_2,x_3,x_4)$ such that for all $\varepsilon \in \interval[open left]{0}{1}$ the probability that running time of the algorithm exceeds $\poly(n,m,\phi,\frac{1}{\varepsilon})$ is at most $\varepsilon$. Our algorithm with polynomial smoothed complexity is then given in \cref{alg:increasing_discount_mpg}. This algorithm simply uses the discounted policy iteration with increasing discount until it finds a pair of optimal policies. 

\begin{theorem}\label{main_algo}
The algorithm from \cref{alg:increasing_discount_mpg} halts for every input and outputs a pair of optimal polices in the mean-payoff game. Furthermore, it has polynomial smoothed complexity.
\end{theorem}

Before giving a proof, we make a remark about the implementation of the algorithm.

\begin{remark}
Computing the value and the Blackwell bias of a zero-player stochastic game (line \ref{line:value-blackwell-bias} of the procedure {\sc IncreasingDiscountPI}) can be done by finding the recurrent classes of the underlying Markov chain and solving a linear system~\cite[Section~8.2.3]{puterman}. In the case of deterministic games, this procedure simplifies in such a way that it can be done in linear complexity.
\end{remark}

\begin{proof}[Proof of \cref{main_algo}]
To prove that this algorithm is correct, suppose that it halts. Then, the policies $(\sigma,\tau)$ are induced by $\bias$, so $(\sigma,\tau)$ are optimal by \cref{bias_induced_optimality}. To prove that the algorithm halts on every input, note that after sufficiently many iterations of the main loop, the discount rate $\disc$ becomes greater than the Blackwell threshold whose existence is given in \cref{th:liggett_lippman}. Hence, the policies $(\sigma,\tau)$ found by policy iteration at that discount rate are Blackwell optimal and $\gameval$ is the value of the two-player game. Furthermore, the definition of the Blackwell bias from \cref{th:kohlberg} implies that the Blackwell bias of the two-player game is the same as the Blackwell bias of the zero-player game obtained by fixing $(\sigma,\tau)$, so $\bias$ is this Blackwell bias. In particular, the algorithm halts and outputs $(\sigma,\tau)$. To prove that the algorithm has polynomial smoothed complexity, we bound the number of switches it performs. Let $\payoff \in \Generic$ and denote $\Cond \coloneqq \max\{1,\Cond(\payoff)\}$. As proven above, the algorithm stops as soon as $\disc$ is greater than the Blackwell threshold, which, by \cref{conditioned_discount}, happens for $\disc > 1 - \frac{1}{6n^2\Cond}$. In particular, the algorithm makes $O\bigl(\log(n \Cond)\bigr)$ updates of the discount factor. Combining this with the bound of \cref{discount_iter_bound}, we get that the algorithm performs
\[
O\Bigl(\log(n \Cond) (m n^2 \Cond)^2 \log^2(n^2\Cond) \Bigr) = O\bigl(n^4m^2\Cond^2\log^3(n\Cond)\bigr)
\]
switches and updates of $\disc$. If we fix $\varepsilon > 0$, then \cref{cond_estimate} shows that we have $\Cond(r) < \frac{8m}{\varepsilon}(\phi + \sqrt{\frac{2m}{\varepsilon}})$ with probability at least $1 - \varepsilon$. In particular,  $\Cond \le 13(\frac{m}{\varepsilon})^{3/2}(\phi + 1)$, so the total number of switches and updates of $\disc$ is $O\bigl(n^4m^5(\phi + 1)^2\varepsilon^{-3} \log^3(\frac{nm(\phi+1)}{\varepsilon})\bigr)$. Every switch and update of $\disc$ can be done in polynomial complexity, so the algorithm has polynomial smoothed complexity. 
\end{proof}

\begin{figure}[t]
\centering
\begin{footnotesize}
  \begin{algorithmic}[1]
    \Procedure{$\IncDiscPItwo$}{$\sigma,\tau,\bar{\disc}$} \Comment{$(\sigma,\tau)$ are any initial policies, $\bar{\disc}$ is the target discount}
    \State $\disc \gets 0$
    \Repeat
         \State $\disc \gets \min\{(1 + \disc)/2, \bar{\disc}\}$
         \State $(\sigma, \tau) \gets \DiscPI(\sigma,\tau,\disc)$
          \State $\sigma' \gets \SwitchMax(\sigma,\tau,\bar{\disc})$
          \State $\tau' \gets \SwitchMin(\sigma,\tau,\bar{\disc})$
      \Until $\sigma' = \sigma$ and $\tau' = \tau$
      \State \Return $(\sigma, \tau)$ \Comment{$(\sigma,\tau)$ are optimal in the discounted game}
\EndProcedure
\end{algorithmic}
\end{footnotesize}
\caption{Policy iteration algorithm for discounted games with polynomial smoothed complexity.} 
\label{alg:increasing_discount_disc}
\end{figure}

The analogous algorithm for discounted games (\cref{alg:increasing_discount_disc}) is obtained by making a simple modification of the previous algorithm.

\begin{theorem}\label{main_algo_disc}
The algorithm from \cref{alg:increasing_discount_disc} halts for every input and outputs a pair of optimal polices in the discounted game with discount rate $\bar{\disc}$. Furthermore, it has polynomial smoothed complexity.
\end{theorem}
\begin{proof}
To show that the algorithm is correct, note that it halts only when $\sigma' = \sigma$ and $\tau' = \tau$, which implies that $(\sigma,\tau)$ are optimal in the discounted game by \cref{th:shapley}. To prove that it halts on every input, note that we always have $\disc \le \bar{\disc}$. Furthermore, after sufficiently many updates of $\disc$, we either have $\disc = \bar{\disc}$ or $\disc$ is greater than the Blackwell threshold. In both cases, the pair of policies returned by $\DiscPI$ is optimal in the game with discount rate $\bar{\disc}$. Since the algorithm halts as soon as $\disc$ exceeds the Blackwell threshold, it has polynomial smoothed complexity by the same reasoning as in the proof of \cref{main_algo}.
\end{proof}

\subsection{Algorithms in the oracle model}

We now move on to the oracle model of computation. To begin, we consider the case of mean-payoff games. In the oracle model, we do not have access to the exact values of the random weights and we can only have access to truncated weights. Hence, we need a procedure that certifies that a pair of policies found by solving truncated games is optimal in the game with untruncated, real weights. Such a procedure was already proposed in \cite{BorosElbassioniFouzGurvich:2011}, and we give a variant of this method below.

\begin{lemma}\label{approx_mpg}
Fix $\payoff \in \R^m$ and let $\tilde{\payoff} \in \R^m$, $\varepsilon > 0$ be such that $\supnorm{\payoff - \tilde{\payoff}} \le \varepsilon$. Let $(\sigma,\tau) \in \Xi$ and let $\edges^{\sigma,\tau} \subseteq \edges$ denote the edges used by $(\sigma,\tau)$. Consider the vectors $\payoff^{(1)}_{ij}, \payoff^{(2)}_{ij} \in \R^m$ defined as
\[
\payoff^{(1)}_{ij} \coloneqq \begin{cases}
\tilde{\payoff}_{ij} &\text{if $(i,j) \in \edges^{\sigma,\tau}$},\\
\tilde{\payoff}_{ij} - 2n\varepsilon &\text{otherwise},\\
\end{cases}
\qquad
\payoff^{(2)}_{ij} \coloneqq \begin{cases}
\tilde{\payoff}_{ij} &\text{if $(i,j) \in \edges^{\sigma,\tau}$},\\
\tilde{\payoff}_{ij} + 2n\varepsilon &\text{otherwise}.\\
\end{cases}
\]
If $(\sigma,\tau)$ are optimal in the mean-payoff games with weights $\payoff^{(1)}$ and $\payoff^{(2)}$, then they are optimal in the mean-payoff game with weights $\payoff$.
\end{lemma}

\begin{proof}
Suppose that $\sigma$ is not a best response (of the maximizing player) to $\tau$ in the game with weights $\payoff$. Then, there exists $i \in [n]$ and a policy $\tilde{\sigma}$ such that $\tilde{\gameval} > \gameval$, where $\tilde{\gameval}$ is the value of the game obtained by fixing $(\tilde{\sigma},\tau)$ and starting at $i$ and $\gameval$ is the value of the game obtained by fixing $(\sigma,\tau)$ and starting at $i$. In particular, the game given by $(\tilde{\sigma},\tau)$ ends in a different cycle than the cycle of $\dgraph^{\sigma,\tau}$. Note that, since $\supnorm{\payoff - \tilde{\payoff}} < \varepsilon$, the value of the game given by $(\sigma,\tau)$ with weights $\payoff^{(2)}$ is at most $\gameval + \varepsilon$. Also, the value of the game given by $(\tilde{\sigma},\tau)$ with weights $\payoff^{(2)}$ is at least $(\tilde{\gameval} - \varepsilon) + 2\varepsilon = \tilde{\gameval} + \varepsilon > \gameval + \varepsilon$, because the final cycle of this game contains at least one edge with weight $\tilde{\payoff}_{kl} + 2n\varepsilon$. Hence, $\sigma$ is not a best response to $\tau$ in the game with weights $\payoff^{(2)}$. We analogously prove that if $\tau$ is not a best response (of the minimizing player) to $\sigma$ in the game with weights $\payoff$, then $\tau$ is not a best response to $\sigma$ in the game with weights $\payoff^{(1)}$.
\end{proof}

\begin{figure}[t]
\centering
\begin{footnotesize}
  \begin{algorithmic}[1]
    \Procedure{$\IncDiscPIor$}{$\sigma,\tau$} \Comment{$(\sigma,\tau)$ are any initial policies}
    \State $\disc \gets 0$
    \State $\varepsilon \gets 1$
    \Repeat
         \State $\disc \gets (1 + \disc)/2$
         \State $\varepsilon \gets \varepsilon/2$
         \State $\tilde{\payoff} \gets$ weights such that $\supnorm{\payoff - \tilde{\payoff}} \le \varepsilon$  \Comment{$\tilde{\payoff}$ is obtained using the oracle}
         \State $(\sigma, \tau) \gets \DiscPI(\sigma,\tau,\disc)$ \Comment{We run policy iteration with weights $\tilde{\payoff}$}
         \If{$(\sigma,\tau) \not\in \Xi$}
         \State \textbf{continue}
         \EndIf
         \State $\payoff^{(1)}, \payoff^{(2)} \gets$ weights as in \cref{approx_mpg} for $(\sigma,\tau,\tilde{\payoff},\varepsilon)$
         \State $(\gameval^{(1)}, \bias^{(1)}) \gets $ value and Blackwell bias of the zero-player game given by $(\sigma,\tau, \payoff^{(1)})$
           \State $(\gameval^{(2)}, \bias^{(2)}) \gets $ value and Blackwell bias of the zero-player game given by $(\sigma,\tau, \payoff^{(2)})$
      \Until $(\gameval^{(1)}, \bias^{(1)})$ and $(\gameval^{(2)}, \bias^{(2)})$ solve the ergodic equations of the two-player games with weights $\payoff^{(1)}$ and $\payoff^{(2)}$ respectively
      \State \Return $(\sigma, \tau)$ \Comment{$(\sigma,\tau)$ are optimal in the mean-payoff game}
\EndProcedure
\end{algorithmic}
\end{footnotesize}
\caption{Policy iteration algorithm for mean-payoff games with polynomial smoothed complexity in the oracle model.} 
\label{alg:increasing_discount_mpg_approx}
\end{figure}

Using the lemma above, we can now adapt the algorithm from \cref{alg:increasing_discount_mpg} to the oracle model. This is given in \cref{alg:increasing_discount_mpg_approx}.

\begin{theorem}\label{main_algo_mpg_aprox}
The algorithm from \cref{alg:increasing_discount_mpg_approx} halts with probability $1$ and outputs a pair of optimal polices in the mean-payoff game. Furthermore, it has polynomial smoothed complexity.
\end{theorem}
\begin{proof}
To prove that the algorithm is correct, note that when it halts, then $(\sigma,\tau)$ belong to $\Xi$ and are bias-induced in the games with weights $\payoff^{(1)},\payoff^{(2)}$ so they are optimal in these games by \cref{bias_induced_optimality} and in the game with weights $\payoff$ by \cref{approx_mpg}. To prove that the algorithm halts with probability one, suppose that $\payoff \in \Generic$ and let $(\bar{\sigma},\bar{\tau}) \in \Xi$ be such that $\payoff \in \polyh^{\bar{\sigma},\bar{\tau}}$. Since $\polyh^{\bar{\sigma},\bar{\tau}}$ is open, there exists $\delta > 0$ such that $B_{\infty}(\payoff,\delta) \subseteq \polyh^{\bar{\sigma},\bar{\tau}}$ and such that the function $\hat{\payoff} \to \Cond(\hat{\payoff})$ is continuous on $B_{\infty}(\payoff,\delta)$. In particular, this function is bounded by some $M > 0$ on this ball. After sufficiently many updates of $\varepsilon$ the vectors $\payoff, \tilde{\payoff}, \payoff^{(1)},\payoff^{(2)}$ belong to the ball $B_{\infty}(\payoff,\delta)$ and the discount factor $\disc$ is above $1 - \frac{1}{6n^2M}$. Hence, it is above the Blackwell threshold for the game with weights $\tilde{\payoff}$ by \cref{conditioned_discount}. Since $\bar{\sigma},\bar{\tau}$ are the unique bias-induced policies in the games with weights $\tilde{\payoff},\payoff^{(1)},\payoff^{(2)}$, they are the unique Blackwell-optimal policies in these games by \cref{bias_induced_optimality}. Hence, $\DiscPI$ will output $\bar{\sigma},\bar{\tau}$ and the vectors $\bias^{(1)}$ and $\bias^{(2)}$ will be the Blackwell biases of games with weights $\payoff^{(1)},\payoff^{(2)}$. In particular, the algorithm will halt. To show that the algorithm has polynomial smoothed complexity, fix $\rho > 0$ and let $\delta \coloneqq \rho/(16(2n+1)m\phi)$. By \cref{cond_estimate_approx}, with probability at least $1 - \rho$, we have $\supnorm{\payoff} < 1 + \frac{1}{\phi}\sqrt{\frac{2m}{\rho}}$, the whole ball $B_{\infty}(\payoff, \delta)$ is in a single polyhedron $\polyh^{\bar{\sigma},\bar{\tau}}$, and $\Cond(\hat{\payoff}) < \frac{16m}{\rho}(\phi + \sqrt{\frac{2m}{\rho}}) + \frac{1}{n+1}$ for any $\hat{\payoff}$ in this ball. In particular, by \cref{conditioned_discount} after $O\bigl(\log(\frac{nm(\phi+1)}{\rho})\bigr)$ updates of $\disc$ and $\varepsilon$ we get $\tilde{\payoff},\payoff^{(1)},\payoff^{(2)} \in B_{\infty}(\payoff, \delta)$ and $\disc$ is above the Blackwell threshold for the game with weights $\tilde{\payoff}$. In particular, the algorithm will stop in $O\bigl(\log(\frac{nm(\phi+1)}{\rho})\bigr)$ updates of $\varepsilon$ and $\disc$. We bound the complexity of this algorithm as in the proof of \cref{main_algo}. Furthermore, we note that the algorithm makes $O\bigl(m\log(\frac{nm(\phi+1)}{\rho})\bigr)$ calls to the oracle that provides the bits of weights $\payoff$. (The inequality $\supnorm{\payoff} < 1 + \frac{1}{\phi}\sqrt{\frac{2m}{\rho}}$ ensures that the integer part of any weight has $O\bigl(\log(\frac{m(\phi+1)}{\rho})\bigr)$ bits.) Hence, the algorithm has polynomial smoothed complexity.
\end{proof}

To finish, let us discuss the problem of solving discounted games in the oracle model. This is a more nuanced task than solving the mean-payoff games. Indeed, for discounted games we need to pay special attention to the case in which the discount factor is below the Blackwell threshold. This is not a problem in the real model of computation, since we can just solve such games exactly, but it brings additional technical difficulties in the oracle model. We have to first specify how the random weights and the discount factor are chosen. Here, we suppose that the ``adversary'' (the person creating the input) first chooses the discount factor, which is given explicitly, and then picks weights of the graph at random (and the bits of the weights are given by an oracle). We note that this model is easier to analyze than the inverse model, in which the adversary can pick the discount factor after the weights are fixed. Indeed, in the latter situation, the adversary could choose a discount rate at which the optimal policy is about to change, so that two policies have a very similar resulting payoff. It is not clear if such a situation could be handled in the oracle model. If we assume that the discount factor is chosen first, then we still need to show that by truncating the weights to $O(\log(n\phi))$ bits we recover the optimal policies. For a fixed discount factor $\disc$, this can be proven by repeating the analysis that we did in the mean payoff case, using the equation from \cref{th:shapley} instead of the ergodic equation. More precisely, for each pair of policies $(\sigma,\tau)$ we define the sets

\begin{align*}
\polyh^{\sigma, \tau}_{\disc} \coloneqq \{\payoff \in \R^m \colon &(\sigma,\tau) \text{ are the unique optimal policies in the discounted game } \\
&\text{with weights $\payoff$ and discount factor $\disc$} \} \, .
\end{align*}

We prove that the sets $\polyh^{\sigma, \tau}_{\disc}$ are open and that the set $\R^m \setminus (\cup_{\sigma,\tau} \polyh^{\sigma,\tau}_{\disc})$ has Lebesgue measure zero. Then, by introducing the discounted analogues of the variables $Z_{ij}$, we prove that the difference between any edge that achieves the maximum/minimum in the equations of \cref{th:shapley} and the ``second best'' edge is large. Furthermore, we note that a value of a discounted zero-player game is $1$-Lipschitz in the weights of the graph. (This is a discounted analogue of the estimates from \cref{linear_functions}.) Putting all of these facts together, we get the following lemma, which is an analogue of \cref{robust_bias_policy}.

\begin{lemma}\label{robust_discounted_policy}
There exists a polynomial $\poly(\cdot)$ such that for every $\varepsilon > 0$ we have
\[
\Prob\biggl(\exists (\sigma, \tau), B_{\infty}\Bigr(\payoff, 1/\poly\bigl(n,m,\phi,\varepsilon^{-1},(1-\disc)^{-1}\bigr) \Bigr) \subseteq \polyh^{\sigma, \tau}_{\disc} \biggr) \ge 1 - \varepsilon \, .
\]
\end{lemma}

We moreover have a discounted analogue of \cref{approx_mpg}.

\begin{lemma}\label{approx_disc}
Fix $\payoff \in \R^m, \disc \in \interval[open]{0}{1}$, and let $\tilde{\payoff} \in \R^m$, $\varepsilon > 0$ be such that $\supnorm{\payoff - \tilde{\payoff}} \le \varepsilon$. Denote $\delta \coloneqq \frac{2\varepsilon}{(1 - \disc)\disc^n}$, let $(\sigma,\tau)$ be any policies, and let $\edges^{\sigma,\tau}$ denote the edges used by $(\sigma,\tau)$. Consider the vectors $\payoff^{(1)}_{ij}, \payoff^{(2)}_{ij} \in \R^m$ defined as
\[
\payoff^{(1)}_{ij} \coloneqq \begin{cases}
\tilde{\payoff}_{ij} &\text{if $(i,j) \in \edges^{\sigma,\tau}$},\\
\tilde{\payoff}_{ij} - \delta &\text{otherwise},\\
\end{cases}
\qquad
\payoff^{(2)}_{ij} \coloneqq \begin{cases}
\tilde{\payoff}_{ij} &\text{if $(i,j) \in \edges^{\sigma,\tau}$},\\
\tilde{\payoff}_{ij} + \delta &\text{otherwise}.\\
\end{cases}
\]
If $(\sigma,\tau)$ are optimal in the discounted games with weights $\payoff^{(1)}$ and $\payoff^{(2)}$, then they are optimal in the discounted game with weights $\payoff$.
\end{lemma}


\begin{figure}[t]
\centering
\begin{footnotesize}
  \begin{algorithmic}[1]
    \Procedure{$\IncDiscPIortwo$}{$\sigma,\tau,\bar{\disc}$} \Comment{$(\sigma,\tau)$ are any initial policies, $\bar{\disc}$ is the target discount}
    \State $\disc \gets 0$
    \State $\varepsilon \gets 1$
    \Repeat
         \State $\disc \gets \min\{(1 + \disc)/2, \bar{\disc}\}$
         \State $\varepsilon \gets \varepsilon/2$
         \State $\tilde{\payoff} \gets$ weights such that $\supnorm{\payoff - \tilde{\payoff}} \le \varepsilon$
         \State $(\sigma, \tau) \gets \DiscPI(\sigma,\tau,\disc)$ \Comment{We run policy iteration with weights $\tilde{\payoff}$}
         \State $\payoff^{(1)}, \payoff^{(2)} \gets$ weights as in \cref{approx_disc} for $(\sigma,\tau,\tilde{\payoff},\bar{\disc},\varepsilon)$
      \Until $(\sigma,\tau)$ are optimal in the games with weights $\payoff^{(1)}, \payoff^{(2)}$ and the discount factor $\bar{\disc}$
      \State \Return $(\sigma, \tau)$ \Comment{$(\sigma,\tau)$ are optimal in the discounted game}
\EndProcedure
\end{algorithmic}
\end{footnotesize}
\caption{Policy iteration algorithm for discounted games with polynomial smoothed complexity in the oracle model.} 
\label{alg:increasing_discount_disc_approx}
\end{figure}

With these lemmas, we can prove that the algorithm from \cref{alg:increasing_discount_disc_approx} solves discounted games in polynomial smoothed complexity in the oracle model.

\begin{theorem}\label{main_algo_disc_approx}
The algorithm from \cref{alg:increasing_discount_disc_approx} halts with probability $1$ and outputs a pair of optimal polices in the discounted game. Furthermore, it has polynomial smoothed complexity.
\end{theorem}
\begin{proof}
To show that the algorithm is correct, note that it halts only when $\sigma, \tau$ are optimal in games with weights $\payoff^{(1)},\payoff^{(2)}$, which by \cref{approx_disc} implies that they are optimal in the discounted game with weights $\payoff$ and discount rate $\bar{\disc}$. To prove that the algorithm halts with probability one, suppose that $\payoff \in \Generic \cap (\cup \polyh^{\sigma,\tau}_{\bar{\disc}})$ and let $\bar{\sigma},\bar{\tau},\hat{\sigma},\hat{\tau}$ be such that $\payoff \in \polyh^{\bar{\sigma},\bar{\tau}} \cap \polyh_{\bar{\disc}}^{\hat{\sigma},\hat{\tau}}$. Furthermore, let $\delta,M > 0$ be such that the ball $B_{\infty}(\payoff, \delta)$ is included in $\polyh^{\bar{\sigma},\bar{\tau}}$ and $\Cond(\check{\payoff}) \le M$ for all $\check{\payoff} \in B_{\infty}(\payoff, \delta)$. (Such a ball exists as in the proof of \cref{main_algo_mpg_aprox}.) After sufficiently many updates of $\disc$ and $\varepsilon$ we will have $\disc = \bar{\disc}$ and $\tilde{\payoff},\payoff^{(1)},\payoff^{(2)} \in \polyh_{\bar{\disc}}^{\hat{\sigma},\hat{\tau}}$, because the set $\polyh_{\bar{\disc}}^{\hat{\sigma},\hat{\tau}}$ is open. Then, the algorithm will stop and output $\hat{\sigma},\hat{\tau}$. This already proves that the algorithm halts with probability one. However, this is not the only situation in which the algorithm halts. The second important situation arises when $\disc \ge 1 - \frac{1}{6M^2}$ and $\varepsilon$ is such that $\tilde{\payoff},\payoff^{(1)},\payoff^{(2)} \in B_{\infty}(\payoff, \delta)$. In this situation, $\bar{\sigma},\bar{\tau}$ are the only Blackwell-optimal policies of the games with weights $\tilde{\payoff},\payoff^{(1)},\payoff^{(2)}$, and, by \cref{conditioned_discount}, $\disc$ is above the Blackwell threshold for all of these games. Since $\disc \le \bar{\disc}$, the algorithm will halt in this situation and output $\bar{\sigma},\bar{\tau}$. To prove that the algorithm has polynomial smoothed complexity, fix $\rho > 0$ and let $\xi \coloneqq 1/\poly\bigl(n,m,\phi,\rho^{-1},(1-\bar{\disc})^{-1}\bigr)$, where $\poly$ is the polynomial from \cref{robust_discounted_policy} and let $\delta \coloneqq \rho/(16(2n+1)m\phi)$. By \cref{robust_discounted_policy} and \cref{cond_estimate_approx}, with probability at least $1 - 2\rho$ the ball $B_{\infty}(\payoff, \xi)$ is in a single set $\polyh_{\bar{\disc}}^{\hat{\sigma},\hat{\tau}}$, the ball $B_{\infty}(\payoff,\delta)$ is in a single polyhedron $\polyh^{\bar{\sigma},\bar{\tau}}$, we have $\Cond(\check{\payoff}) < \frac{16m}{\rho}(\phi + \sqrt{\frac{2m}{\rho}}) + \frac{1}{n+1}$ for any $\check{\payoff}$ in this ball, and $\supnorm{\payoff} < 1 + \frac{1}{\phi}\sqrt{\frac{2m}{\rho}}$. Let $M \coloneqq \frac{16m}{\rho}(\phi + \sqrt{\frac{2m}{\rho}}) + \frac{1}{n+1}$. After $O\bigl(\log(\frac{nm(\phi+1)}{\rho})\bigr)$ updates of $\disc$ we either have $\disc > 1 - \frac{1}{6M^2}$ or $\disc = \bar{\disc} \le 1 - \frac{1}{6M^2}$. In the latter case, we have $\frac{1}{1 - \bar{\disc}} \le 6M^2$, so within $O\bigl(\log(\frac{nm(\phi+1)}{\rho})\bigr)$ further updates of $\varepsilon$ we get $\tilde{\payoff},\payoff^{(1)},\payoff^{(2)} \in B_{\infty}(\payoff, \xi)$. At this point, the algorithm stops and outputs $\hat{\sigma},\hat{\tau}$. Moreover, the number of switches it performs is bounded as in the proof of \cref{main_algo}. In the former case, we have $\disc > 1 - \frac{1}{6M^2}$ and within $O\bigl(\log(\frac{nm(\phi+1)}{\rho})\bigr)$ updates of $\varepsilon$ we get $\tilde{\payoff} \in B_{\infty}(\payoff, \delta)$. At this point, $\DiscPI$ outputs $\bar{\sigma},\bar{\tau}$ by \cref{conditioned_discount}. Furthermore, note that in the subsequent iterations of the main loop, $\DiscPI$ will never perform a single switch: the policies $\bar{\sigma},\bar{\tau}$ remain optimal for all $\disc > 1 - \frac{1}{6M^2}$ and $\supnorm{\tilde{\payoff} - \payoff} \le \varepsilon$. Hence, the number of switches performed by the algorithm is bounded as in the proof of \cref{main_algo}. Afterwards, within $O\bigl(\log(\frac{nm(\phi+1)}{\rho(1-\bar{\disc})})\bigr)$ updates of $\varepsilon$ we will have $\payoff^{(1)},\payoff^{(2)} \in B_{\infty}(\payoff, \delta)$, so the algorithm will stop and output $\bar{\sigma},\bar{\tau}$. The complexity of each of these iterations is dominated by the complexity of solving a linear system (one system is solved by $\DiscPI$ to check that $\bar{\sigma},\bar{\tau}$ are optimal in the game with weights $\tilde{\payoff}$ and two more are needed to check if $\bar{\sigma},\bar{\tau}$ are optimal in games with weights $\payoff^{(1)},\payoff^{(2)}$). Hence, the algorithm has polynomial smoothed complexity.
\end{proof}

\section{Conclusion}

\newcounter{nparcounter}
\setcounter{nparcounter}{0}


\NewDocumentCommand{\npar}{o}{%
  \IfValueTF{#1}{%
    \medskip\noindent\refstepcounter{nparcounter}\textbf{\thenparcounter}\quad \textit{#1}. %
  }{%
    \medskip\noindent\refstepcounter{nparcounter}\textbf{\thenparcounter}. %
  }%
}

We gave an analysis of two-player discounted and mean-payoff games, that led to a condition number, and a policy-iteration algorithm which is efficient on well-conditioned inputs. We showed that random inputs are well-conditioned with high probability. A few remarks are in order.

\npar Our techniques work for two-player games played on ergodic graphs. In non-ergodic graphs, the value $\lambda_i$ is not necessarily the same at each vertex $i$. A folklore reduction, appearing for example in \cite{ChatterjeeHenzingerKrinningerNanongkai:2014}, shows that computing the value vector of a non-ergodic game reduces to computing the value of an ergodic game. So one can ask if our algorithms can be used on non-ergodic games. The answer is not obvious. The reduction proceeds in rounds, where in the first round one finds, say, the largest coordinate $\lambda_i$ of the value vector, and then discards the node $i$ (which requires some care) and repeats. Now, if one takes a non-ergodic game $\dgraph$ with sufficiently random payoffs, and applies this reduction, the resulting game is sufficiently random at the first round, but it is not clear what happens in the succeeding rounds. So, as far as we can tell, the question remains open: \emph{Do deterministic two-player discounted and mean-payoff games have polynomial smoothed complexity, when played on non-ergodic graphs?} A possible way of answering this question is by doing an analysis of Blackwell-optimal policies in the non-ergodic case, similar to what we have done here for the ergodic case.

\npar Allamigeon, Gaubert, Katz and Skomra \cite{AllamigeonGaubertKatzSkomra:2022} show that a certain value-iteration algorithm runs efficiently on all ergodic instances with value $\lambda$ bounded away from zero. They use $\frac{\max_i u_i - \min_i u_i}{|\lambda|}$ as a condition number. Can we use their result to show that value iteration has polynomial smoothed complexity? I.e., is a sufficiently-random instance well-conditioned as per their condition number? This was the central question left unanswered in their paper, and we tried to solve it, or provide a counter-example, but have so far failed to do so.

\npar Our policy-iteration rule is not one of the standard rules (Howard, lexicographic, RandomFacet, \emph{etc}). Do these standard rules also have polynomial smoothed complexity on deterministic two-player games? How about other ``combinatorial'' algorithms? 

\npar Can we extend our results to \emph{stochastic} two-player games? The counter-example of Christ and Yannakakis shows that the Howard all-switches rule does not have polynomial smoothed complexity on stochastic two-player games. This seems to indicate that the stochastic setting is more delicate. On the other hand, our policy iteration rule is different to Howard's. So one could tentatively ask: is there a smoothed counter-example to the Howard rule also in the deterministic (say, two-player) setting? This would show that our policy-iteration rule cannot be replaced by the Howard rule.

\npar How about other problems in $\mathsf{UEOPL}$? Some of these problems are combinatorial, and do not seem to be amenable to smoothed analysis. But one can consider, for example, the P-Matrix Linear Complementarity Problem (P-LCP, see \cite[Section 4.3]{fearnley2020unique}), and ask: \emph{does it have polynomial smoothed complexity?} More broadly speaking, one can make the conjecture that \emph{every problem in $\mathsf{UEOPL}$ becomes easy under a suitable notion of perturbation}. This conjecture is broad and imprecise, but it might be an interesting starting point for further research.


\bibliography{phd_bibliography}
\bibliographystyle{alpha}

\end{document}
